\documentclass[11pt]{article} 
\usepackage{xspace}
\usepackage{float}
\usepackage{marginnote}
\usepackage{misc} 
\usepackage{amsmath}
\usepackage{amssymb}
\usepackage{enumerate}
\usepackage{epstopdf}
\usepackage{times}
\newcommand{\poly}{\mathrm{poly}}
\newcommand{\mnote}{\marginnote}

\usepackage{graphicx}
\usepackage{subfigure}
\usepackage{amsfonts}
\usepackage{amsmath}
\usepackage[svgnames,table]{xcolor}
\usepackage[symbol]{footmisc}
\renewcommand*{\thefootnote}{\arabic{footnote}}
\usepackage{booktabs}

\usepackage{euscript}
\usepackage{ifthen}

\usepackage{algorithm}
\usepackage[noend]{algpseudocode}
\usepackage{eqparbox}

\usepackage{color}
\newcommand{\boundary}{\EuScript{K}}

\newcommand{\ty}[1]{\tilde{y}(#1)}
\newcommand{\tys}[1]{\tilde{y}^*(#1)}

\newcommand{\ignore}[1]{}
\newcommand{\sync}{\textsc{Sync}}

\newcommand{\dist}{{\sf c}}

\newcommand{\wt}{{\tt \sf c}}

\newcommand{\Mi}[1]{M^{(#1)}}

\newcommand{\Gr}{G}
\newcommand{\R}{{\cal R}}

\newcommand{\opt}{\textsc{Opt}}

\newcommand{\dir}[1]{\overrightarrow{#1}}

\newcommand{\searchandswitch}{\textsc{SearchAndSwitch}}
\newcommand{\switch}{\textsc{Switch}}

\newcommand{\construct}{\textsc{Construct}}
\newcommand{\hungarian}{\textsc{HungarianSearch}}

\newcommand{\reduce}{\textsc{Reduce}}
\newcommand{\reduceslack}{\textsc{ReduceSlack}}

\newcommand{\fastmatch}{\textsc{FastMatch}}

\newcommand{\numphases}{\sqrt{n}/r^{1/4}}
\newcommand{\unmatchedrem}{\sqrt{n} / r^{1/4}}
\newcommand{\bruteforceconstruct}{O(\sqrt{r}(m_j + n_j \log{n_j}))}

\newcommand{\steptwoweight}{\lceil \numphases \rceil  \sqrt{r}}
\newcommand{\inactiveH}{B^\mathcal{I}_H} 
\newcommand{\activeH}{B^\mathcal{A}_H}

\newcommand{\csetsymdiff}{\mathcal{C}_\oplus}
\newcommand{\psetsymdiff}{\mathcal{P}_\oplus}

\newcommand{\project}[3]{\dir{P}_{#1,#2,#3}}

\newcommand{\yjs}{y^*_j}
\newcommand{\sjs}{s^*}

\usepackage{todonotes}

\newboolean{disable-comments} 
\newboolean{disable-todos}
 
\setboolean{disable-comments}{true}
\setboolean{disable-todos}{true}

\newcommand{\tnote}[1]{\ifthenelse{\boolean{disable-comments}}{}{ \mnote{#1}}}   

\newcommand{\todonote}[1]{\ifthenelse{\boolean{disable-todos}}{}{ \todo[inline]{#1}}}

\newcommand{\CTotalIteration}{O((n^2/r) \log^2{r})}
\newcommand{\CTotalDelta}{O(n\sqrt{r}\log{r}\log{n})}

\newcommand{\CIteration}{O((n/\sqrt{r})\log^2{r})}

\newtheorem{lemma}{Lemma}[section]

\newtheorem{cor}[lemma]{Corollary}

\newtheorem{definition}[lemma]{Definition}

\begin{document}

\begin{titlepage}

\title{A Faster Algorithm for Minimum-Cost Bipartite Matching in Minor-Free Graphs \thanks{This work is supported by the National Science Foundation under grant NSF-CCF 1464276.}}
\author{ Nathaniel Lahn\thanks{Department of Computer Science, Virginia Tech. Email: \texttt{lahnn@vt.edu}}
\and
Sharath Raghvendra\thanks{Department of Computer Science, Virginia Tech. Email:
 \texttt{sharathr@vt.edu}}    
}
\date{}
\maketitle

\maketitle

\begin{abstract}
We give an $\tilde{O}(n^{7/5} \log (nC))$-time\footnote{We use $\tilde{O}(\cdot)$ to suppress logarithmic terms throughout the paper.} algorithm to compute a minimum-cost maximum cardinality matching (optimal matching) in $K_h$-minor free graphs with $h=O(1)$ and integer edge weights having magnitude at most $C$. This improves upon the $\tilde{O}(n^{10/7}\log{C})$ algorithm of Cohen et al. [SODA 2017] and the $O(n^{3/2}\log (nC))$ algorithm of Gabow and Tarjan [SIAM J. Comput. 1989].

For a graph with $m$ edges and $n$ vertices, the well-known Hungarian Algorithm computes a shortest augmenting path in each phase in $O(m)$ time, yielding an optimal matching in $O(mn)$ time. The Hopcroft-Karp [SIAM J. Comput. 1973], and Gabow-Tarjan [SIAM J. Comput. 1989] algorithms compute, in each phase, a maximal set of vertex-disjoint shortest augmenting paths (for appropriately defined costs) in $O(m)$ time. This reduces the number of phases from $n$ to $O(\sqrt{n})$ and the total execution time to $O(m\sqrt{n})$.

In order to obtain our speed-up, we relax the conditions on the augmenting paths and iteratively compute, in each phase, a set of carefully selected augmenting paths that are not restricted to be shortest or vertex-disjoint. As a result, our algorithm computes substantially more augmenting paths in each phase, reducing the number of phases from $O(\sqrt{n})$ to $O(n^{2/5})$. By using small vertex separators, the execution of each phase takes $\tilde{O}(m)$ time on average.
For planar graphs, we combine our algorithm with efficient shortest path data structures to obtain a minimum-cost perfect matching in $\tilde{O}(n^{6/5} \log{(nC)})$ time. This improves upon the recent $\tilde{O}(n^{4/3}\log{(nC)})$ time algorithm by Asathulla et al. [SODA 2018].

\end{abstract}
\end{titlepage}
\section{Introduction}

Consider a bipartite graph $\Gr(A \cup B, E)$ with $|A|=|B|=n$. For the edge set $E\subseteq A\times B$, let every edge $(a,b)\in E$ have a \emph{cost} specified by $\dist(a,b)$. A \emph{matching} $M\subseteq E$ is a set of vertex-disjoint edges whose cost $\wt(M)$ is given by $\sum_{(a,b)\in M}\dist(a,b)$. $M$ is a \emph{maximum cardinality} matching if  $M$ is the largest possible set of vertex-disjoint edges. A \emph{minimum-cost maximum cardinality matching} is a maximum cardinality matching with the smallest cost.  \ignore{In this paper, we design an $\tilde{O}(n^{7/5})$ time algorithm to compute a minimum-cost maximum cardinality matching for the case where $\Gr$ is a $K_h$-minor free graph.}

In this paper, we present an efficient algorithm for any graph that admits an \emph{$r$-clustering}. A \emph{clustering} of a graph $G$ is a partitioning of $G$ into edge disjoint pieces. A vertex is a \emph{boundary} vertex if it participates in more than one piece of this clustering. For a parameter $r > 0$, an $r$-clustering of a graph is a clustering of $G$ into edge-disjoint pieces $\{\R_1,\ldots,\R_k\}$ such that $k = \tilde{O}(n/\sqrt{r})$, every piece $\R_j$ has at most $O(r)$ vertices, and each piece has $\tilde{O}(\sqrt{r})$ boundary vertices. Furthermore, the total number of boundary vertices, counted with multiplicity, is $\tilde{O}(n/\sqrt{r})$. 

For any directed graph $G$ with an $r$-clustering, one can \emph{compress} $G$ to a graph $H$ as follows. The vertex set of $H$ is the set of boundary vertices and we add an edge in $H$ if the two boundary vertices are connected by a directed path inside one of the pieces. It is easy to see that $H$ has $\tilde{O}(n/\sqrt{r})$ vertices and $\tilde{O}(n)$ edges. 

\ignore{We would like to note that for any graph with an efficiently computable $O(\sqrt{n})$ separator on all subgraphs, an $r$-clustering can be efficiently computed by recursively applying the weighted separator theorems of Djidjev and Gilbert \cite{weighted_separators}. Therefore, both planar graphs and $K_h$-minor free graphs admit an efficiently computable $r$-clustering. }

In this paper, we design an $\tilde{O}(mr + m \unmatchedrem)$ time algorithm to compute minimum-cost matching  for bipartite graphs that admit an $r$-clustering. Setting $r=n^{2/5}$ minimizes the running time to $\tilde{O}(mn^{2/5})$.  For several natural classes of graphs such as planar graphs and $K_h$-minor free graphs, there are fast algorithms to compute an $r$-clustering for any given value of $r$. For such graphs, we obtain faster minimum-cost matching algorithms.
\ignore{
\paragraph{Unweighted Matching} In arbitrary bipartite graphs with $n$ vertices and $m$ edges, Ford and Fulkerson's algorithm~\cite{ford_fulkerson} 
iteratively computes, in each phase, an augmenting path in $O(m)$ time, leading to an
 execution time of $O(mn)$. Hopcroft and Karp's algorithm~\cite{hk_sicomp73}, in each phase, computes a maximal set of vertex-disjoint shortest augmenting paths in $O(m)$ time. This reduces the number of phases from $n$ to $O(\sqrt{n})$ and the total execution time to $O(m\sqrt{n})$. More recently, two new methods have been introduced for computing maximum matchings. Using the electric flow-based paradigm a maximum cardinality matching can be computed in $\tilde{O}(m^{10/7})$\cite{madry2013navigating} time. Using the matrix multiplication approach, a maximum cardinality matching can be computed in $\tilde{O}(n^\omega)$. Here, $\omega$ is the matrix multiplication complexity exponent, which is $O()$.
}

\paragraph{Previous work.}  In an arbitrary bipartite graph with $n$ vertices and $m$ edges, Ford and Fulkerson's algorithm~\cite{ford_fulkerson} 
iteratively computes, in each phase, an augmenting path in $O(m)$ time, leading to a maximum cardinality matching in $O(mn)$ time. Hopcroft and Karp's algorithm (HK-Algorithm)~\cite{hk_sicomp73}, in each phase, computes a maximal set of vertex-disjoint shortest augmenting paths in $O(m)$ time. This reduces the number of phases from $n$ to $O(\sqrt{n})$ and the total execution time to $O(m\sqrt{n})$. For planar graphs, multiple-source multiple-sink max-flow can be computed in $O(n\log^3{n})$ time~\cite{multiple_planar_maxflow}. This also gives an $O(n\log^3{n})$ algorithm for maximum cardinality bipartite matching on planar graphs. Such matching and flow algorithms in planar graphs are based on a reduction to computing shortest paths in the planar dual graph. However, it is unclear how such a reduction extends to minimum cost matching or $K_h$-minor free graphs.

In weighted graphs with $n$ vertices and $m$ edges, the well-known Hungarian method computes a minimum-cost maximum cardinality matching in  $O(mn)$ time~\cite{hungarian_56} by iteratively computing a shortest augmenting path.  Gabow and Tarjan designed a cost-scaling algorithm (GT-Algorithm) to compute a minimum-cost perfect matching in $O(m\sqrt{n}\log (nC))$, where $C$ is the largest cost on any edge of the graph~\cite{gt_sjc89}. Their method, like the Hopcroft-Karp algorithm, computes a maximal set of vertex disjoint shortest (for an appropriately defined augmenting path cost) augmenting paths in each phase. To assist in computing these paths, they introduce an error of $+1$ to the cost of each matching edge, giving a total error of $O(n)$ for any matching. Using the scaling approach, they are able to compute an optimal matching despite these errors. Furthermore, they are able to show that the total error in all augmenting paths computed during each scale is $O(n\log{n})$, which in turn bounds the total number of edges on all the augmenting paths by $O(n\log{n})$.

Recently, Asathulla~\etal~\cite{soda-18} gave an $\tilde{O}(n^{4/3}\log{(nC)})$ scaling algorithm (AKLR-Algorithm) for minimum-cost perfect matching in planar graphs. We have outlined their approach in Algorithm \ref{alg:planar}. Initially, their algorithm executes $O(\sqrt{r})$ iterations of the GT-Algorithm in order to match all but $O(n/\sqrt{r})$ vertices (line \ref{alg:planar_line:preproc}). Next, they use the $r$-clustering to construct the compressed residual graph $H$ (line \ref{alg:planar_line:construct}) consisting of $O(n/\sqrt{r})$ vertices and $O(n)$ edges. Each edge of the compressed residual graph is given a weight equal to that of the shortest path between the two vertices in the corresponding piece.

\begin{figure}
\label{fig:runtimes}
\begin{tabular}{*4l}\toprule
  & \textbf{Arbitrary Graphs} & \textbf{$K_h$-Minor Free Graphs} & \textbf{Planar Graphs}\\\midrule
 \textbf{Strongly Polynomial}	&$O(mn)$~~\cite{hungarian_56}	 					&      $\tilde{O}(n^{3/2})$~~\cite{lt}        &       $\tilde{O}(n^{3/2})$~~\cite{lt}       	\\
 \textbf{Cost Scaling Algorithms}			&   $\tilde{O}(m\sqrt{n}\log{(nC)})$~~\cite{gt_sjc89}&      $O(n^{3/2}\log{(nC)})$~~\cite{gt_sjc89}      	&       $\tilde{O}(n^{4/3}\log{(nC)})$~~\cite{soda-18}\\
 \textbf{Matrix Multiplication} &   $\tilde{O}(n^{\omega}C)$~~\cite{fast_matrix_weighted} 	 	&      $\tilde{O}(n^{\omega/2}C)$~~\cite{kawarabayashi2010separator,fast_matrix_matching_planar,fast_matrix_weighted,wulff2011separator}  	&       $\tilde{O}(n^{\omega/2}C)$~~\cite{fast_matrix_matching_planar,fast_matrix_weighted}      \\
 \textbf{Electric Flow}			&   $\tilde{O}(m^{10/7}\log{C})$~~\cite{cohen2017negative}       	&      $\tilde{O}(n^{10/7}\log{C})$~~\cite{cohen2017negative}     	&       $\tilde{O}(n^{10/7}\log{C})$~~\cite{cohen2017negative}   		\\
 \textbf{Our Results}			&   \multicolumn{1}{c}{--}										&      $\tilde{O}(n^{7/5}\log{(nC)})$      	&       $\tilde{O}(n^{6/5}\log{(nC)})$			\\
\bottomrule
\end{tabular}
\caption{Comparison of execution times of various matching algorithms. Here $C$ is the largest cost edge; $\omega<2.373$ is the exponent of matrix multiplication complexity. Our algorithm is the fastest for $K_h$-minor free graphs (resp. planar graphs) when $C= 2^{O(n^{0.1})}$ (resp. $C=2^{O(n^{0.3})}$) and $C = \tilde{\Omega}(n^{0.22})$ (resp. $C = \tilde{\Omega}(n^{0.02})$).}
\end{figure}

The AKLR-Algorithm iteratively finds a single shortest augmenting path each iteration (line \ref{alg:planar_line:search}) and augments the matching along this path (line \ref{alg:planar_line:augment}). Their use of planar shortest path data structures allows each augmenting path to be found quickly, in time proportional to the number of vertices of the compressed graph, i.e., $\tilde{O}(n/\sqrt{r})$. Augmenting along a path forces portions of the compressed residual graph  to be updated. In order to limit the number of these updates, they introduce an error of $\sqrt{r}$ on each boundary vertex of the compressed graph. Like the GT-Algorithm, they are able to show that the total error of all augmenting paths computed by the AKLR-Algorithm is $O(n\log n)$ and so, the augmenting paths can use at most $O((n/\sqrt{r})\log{n})$ boundary vertices, which in turn bounds the total number of updates. Furthermore, only $O(n/\sqrt{r})$ vertices have an error of $\sqrt{r}$ and so, the total error of the optimal matching is still $O(n)$, and is removed using the scaling approach. 

Two other approaches for weighted matching problems include the randomized matrix multiplication approach ($\tilde{O}(n^{\omega/2}C)$ time) ~\cite{fast_matrix_weighted} and the electric-flow based approach ($\tilde{O}(m^{10/7}\log{C}$) time)~\cite{cohen2017negative}. Here, $\omega < 2.373$ is the exponent of matrix multiplication complexity. For a comparison of various approaches' running times, see Figure \ref{fig:runtimes}.
\renewcommand*{\thefootnote}{\fnsymbol{footnote}}
\begin{algorithm}
\caption{A scale of the AKLR-Algorithm with complexities.}
\label{alg:planar}
\begin{algorithmic}[1]
\State $M \gets \emptyset$\Comment{\bf Time}{}
\State \label{alg:planar_line:preproc}\textbf{Preprocessing Step}: Run $\sqrt{r}$ iterations of GT-Algorithm $\cdot\cdot\cdot\cdot\cdot\cdot\cdot$ \Comment{$\tilde{O}(n\sqrt{r})$}{}
\State\label{alg:planar_line:construct}Compute the compressed graph $H$ $\cdot\cdot\cdot\cdot\cdot\cdot\cdot\cdot\cdot\cdot\cdot\cdot\cdot\cdot\cdot\cdot\cdot\cdot\cdot\cdot\cdot\cdot\cdot\cdot\cdot\cdot\cdot\cdot\cdot$ \Comment{$\tilde{O}(n)$}{}
\For{$i$ from 1 to $O(n/\sqrt{r})$}
\State \label{alg:planar_line:search}$P \gets $ \hungarian($H$) $\cdot\cdot\cdot\cdot\cdot\cdot\cdot\cdot\cdot\cdot\cdot\cdot\cdot\cdot\cdot\cdot\cdot\cdot\cdot\cdot\cdot\cdot\cdot\cdot\cdot\cdot\cdot\cdot\cdot\cdot\cdot$ \Comment{$\tilde{O}(n/\sqrt{r})$ per iteration}{}
\State \label{alg:planar_line:augment} Augment $M$ along $P$ and update $H$ $\cdot\cdot\cdot\cdot\cdot\cdot\cdot\cdot\cdot\cdot\cdot\cdot\cdot\cdot\cdot\cdot\cdot\cdot\cdot\cdot\cdot\cdot\cdot\cdot\cdot$ \Comment{$\tilde{O}(r)\footnotemark[1]$ per iteration }{}
\EndFor\label{planarPhaseLoop}
\State \textbf{return} $M$
\end{algorithmic}
\end{algorithm}
\footnotetext[1]{The given time is averaged over all iterations.}
\renewcommand*{\thefootnote}{\arabic{footnote}}


\subsection{Our results}
In this paper, we design an  algorithm to compute minimum-cost perfect matching in bipartite graphs with an $r$-clustering. Our algorithm runs in $\tilde{O}(mr + m \unmatchedrem)$ time. For $r=n^{2/5}$, we obtain an $\tilde{O}(mn^{2/5})$ time algorithm to compute the optimal matching. As consequences, we obtain the following results:

\begin{itemize}
\item For $K_h$-minor free graphs, we obtain an $\tilde{O}(n^{7/5}\log{(nC)})$ time algorithm to compute the minimum-cost matching. In comparison, a min-cost matching can be computed in $\tilde{O}(n^{10/7} \log{C})$ time~\cite{cohen2017negative}.
\item For planar graphs, our approach leads to an execution time of $\tilde{O}(n^{6/5}\log{(nC)}) $ improving the previous $\tilde{O}(n^{4/3}\log{(nC)})$ time algorithm by Asathulla~\etal~\cite{soda-18}.

\ignore{ 
\item Although the main focus of this paper is an $\tilde{O}(n^{7/5}\log{nC})$ algorithm for $H$-minor free graphs, the result extends to other classes of graphs closed under an $O(\sqrt{n})$ balanced separator property. By recursively using balanced separators, each piece of the graph can be recursively split such that both the numbers of vertices and boundary vertices per piece decreases geometrically throughout each level of the recursion. This standard technique is applicable to any graph with a separator theorem based on a result by Djidjev  and Gilbert \cite{djidjev_gilbert_separators}. On such graphs with $m$ vertices and $n$ edges where a balanced $O(\sqrt{n})$-sized separator can be computed quickly, our algorithm takes $\tilde{O}(mn^{2/5}\log{C})$ time.\footnote{To accommodate such graphs, our presentation will take into account the number of edges throughout the paper.}

}
\end{itemize} 

An $r$-clustering can be quickly constructed for any graph with $m$ edges, $n$ vertices, and an efficiently computable $O(\sqrt{n})$-sized separator on all subgraphs. For such a graph, our algorithm computes a minimum-cost matching in $\tilde{O}(mn^{2/5}\log{(nC)})$ time.\footnote{To accommodate such graphs, our presentation will take into account the number of edges throughout the paper.}
The reduction of Gabow and Tarjan from maximum cardinality minimum-cost matching to minimum-cost perfect matching preserves the  $r$-clustering in the input graph. Therefore, we can  use the same reduction to also compute a minimum-cost maximum cardinality matching in $\tilde{O}(mn^{2/5})$ time. Our results are based on a new approach to speed-up augmenting path based matching algorithms, which we describe next.

\begin{algorithm}
\caption{A scale of our algorithm for $K_h$-minor free graphs with complexities.}\label{alg:hmfree}
\begin{algorithmic}[1]
\State $M \gets \emptyset$\Comment{{\bf Time}}{}
\State \textbf{Preprocessing Step}: Run $\sqrt{r}$ iterations of GT-Algorithm $\cdot\cdot\cdot\cdot\cdot\cdot\cdot$ \Comment{$\tilde{O}(m\sqrt{r})$}{}
\State Compute compressed residual graph $H$ $\cdot\cdot\cdot\cdot\cdot\cdot\cdot\cdot\cdot\cdot\cdot\cdot\cdot\cdot\cdot\cdot\cdot\cdot\cdot\cdot\cdot\cdot\cdot\cdot\cdot\cdot\cdot$ \Comment{$\tilde{O}(m\sqrt{r})$}{}
\For{$i$ from 1 to $O(\sqrt{n}/r^{1/4})$} \label{step2startline}
\State Execute \fastmatch\ to find many augmenting paths, $\cdot\cdot\cdot\cdot\cdot\cdot\cdot\cdot\cdot$ \Comment{$\tilde{O}(m)^*$ per iteration}{}
\EndFor\label{phaseLoop}
\State\label{step2endline}~~~~~Augment and update $H$ for all paths $\cdot\cdot\cdot\cdot\cdot\cdot\cdot\cdot\cdot\cdot\cdot\cdot\cdot\cdot\cdot\cdot\cdot\cdot\cdot\cdot\cdot\cdot\cdot\cdot\cdot$ \Comment{$\tilde{O}(mr^{5/4}/\sqrt{n}))^*$ per iteration}{}
\For{$i$ from 1 to $O(\sqrt{n}/r^{1/4})$}\label{step3startline}
\State $P \gets$ \hungarian($G$) $\cdot\cdot\cdot\cdot\cdot\cdot\cdot\cdot\cdot\cdot\cdot\cdot\cdot\cdot\cdot\cdot\cdot\cdot\cdot\cdot\cdot\cdot\cdot\cdot\cdot\cdot\cdot\cdot\cdot\cdot\cdot$ \Comment{$\tilde{O}(m)$ per iteration}{}
\State\label{step3endline} Augment $M$ along $P$ $\cdot\cdot\cdot\cdot\cdot\cdot\cdot\cdot\cdot\cdot\cdot\cdot\cdot\cdot\cdot\cdot\cdot\cdot\cdot\cdot\cdot\cdot\cdot\cdot\cdot\cdot\cdot\cdot\cdot\cdot\cdot\cdot\cdot\cdot\cdot\cdot\cdot\cdot\cdot$ \Comment{$\tilde{O}(n)$ per iteration}{}
\EndFor\label{postprocLoop}
\State \textbf{return} $M$
\end{algorithmic}
\end{algorithm}

\subsection{Our approach}     
The HK, Hungarian, GT, and AKLR algorithms rely upon computing, in each phase, one or more vertex-disjoint minimum-cost augmenting paths, for an appropriate cost definition. To assist in computing these paths, each algorithm defines a weight on every vertex.

For instance, the HK-Algorithm assigns a layer number to every vertex by conducting a BFS from the set of free vertices in the residual graph.   Any augmenting path that is computed in a \emph{layered graph} -- a graph consisting of edges that go from a vertex of some layer $i$ to layer $i+1$ -- is of minimum length. Similarly, the Hungarian, GT and AKLR algorithms  assign a dual weight to every vertex satisfying a set of constraints, one for each edge. Any augmenting path in an \emph{adimissible graph} -- containing edges for which the dual constraints are ``tight'' and have zero slack -- is a minimum-cost augmenting path for an appropriate cost. Hungarian and AKLR Algorithms iteratively compute such augmenting paths and augment the matching along these paths.

The GT-Algorithm (resp. HK-Algorithm) computes, in each phase, a maximal set of vertex-disjoint augmenting paths in the admissible graph (resp. layered graph) by iteratively conducing a partial DFS from every free vertex. Each such DFS terminates early if an augmenting path is found. Moreover, every vertex visited by this search is immediately discarded from all future executions of DFS for this phase. This leads to an $O(m)$ time procedure to obtain a maximal set of vertex-disjoint augmenting paths.

In order to obtain a speed-up, we deviate from these traditional matching algorithms as follows. In each phase, we compute substantially more augmenting paths that are not necessarily minimum-cost or vertex-disjoint. We accomplish this by allowing the admissible graph to have certain edges with positive slack. We then conduct a partial-DFS on this admissible graph. Unlike traditional methods, we do not discard vertices that were visited by the DFS and instead allow them to be reused. As a result, we discover more augmenting paths. Revisits increase the execution time per phase. Nonetheless, using the existence of an $r$-clustering, we bound the amortized execution time by $O(m)$ per phase.

So, how do we guarantee that our algorithm computes significantly more augmenting paths in each phase?  The HK-Algorithm measures progress made by showing that the length of the shortest augmenting path increases by at least one at the end of each phase. After $\sqrt{n}$ phases, using the fact that the length of the shortest augmenting path is at least $\sqrt{n}$, one can bound the total number of free vertices by $O(\sqrt{n})$.  

In the GT-Algorithm, the dual weights assist in measuring this progress. Gabow and Tarjan show that the free vertices of one set, say $B$, increases by at least $1$ in each phase whereas the dual weights of free vertices of $A$ always remains $0$. After $\sqrt{n}$ phases, using the fact that the dual weights of all free vertices is at least $\sqrt{n}$, one can bound the total number of free vertices by $\sqrt{n}$. Note that this observation uses the fact that at the beginning of each scale, the cost of the optimal matching is $O(n)$. 

In our algorithm, which is also based on the scaling paradigm, we achieve a faster convergence by aggressively increasing the dual weight of free vertices of $B$ by $O(n^{1/5})$ while maintaining the dual weights of free vertices of $A$ at $0$. So, the progress made in one phase of our algorithm is comparable to the progress made by $O(n^{1/5})$ phases of GT-Algorithm. As a result, at the end of $O(n^{2/5})$ phases, the dual weight of every free vertex is at least $n^{3/5}$ and the number of free vertices remaining is no more than $O(n^{2/5})$. Each of the remaining can be matched in $O(m)$ time by conducting a simple Hungarian Search leading to an execution time of $O(mn^{2/5})$.

\ignore{In the context of weighted matching, using the scaling paradigm, it is often the case that edges have non-negative integer costs and the optimal matching cost for a scale is $O(n)$. In each scale, Gabow and Tarjan's algorithm iteratively computes a set of vertex-disjoint augmenting paths and  augments the matching along these paths. Each iteration, the dual weight of every unmatched vertex of $B$ increases by at least $1$ whereas the dual weight of every unmatched vertex of $A$ remains $0$. After $\sqrt{n}$ iterations, the dual weight of each free vertex of $B$ is at least $\sqrt{n}$. It can be shown that the sum of the dual weights of all unmatched vertices cannot exceed the cost of the optimal matching.}
\ignore{
\begin{eqnarray}
\sum_{v\in A_F\cup B_F}y(v)\le |B_F|\sqrt{n} \le \wt(M_{\opt}) \le O(n).\label{eq:convergence}\end{eqnarray}

Therefore, it follows that there are at most $O(\sqrt{n})$ unmatched vertices remaining. Each of these unmatched vertices can then be matched iteratively in $O(m)$ time leading to a total execution time of $O(m\sqrt{n})$. To further improve the execution time, one possibility is to aggressively increase the dual weights of free vertices of $B$ by a value $\Omega(n^{\eps})$ for some constant  $0<\eps<1$. In this case, after $n^{(1-\eps)/2}$ iterations, there are $O(n^{(1-\eps)/2})$ free vertices remaining. If each iteration could be executed in $O(m)$ time, the resulting execution time would be $O(mn^{(1-\eps)/2})$. 
}


\ignore{\begin{figure}[H]
  \centering
  \includegraphics[width=5cm,scale=1]{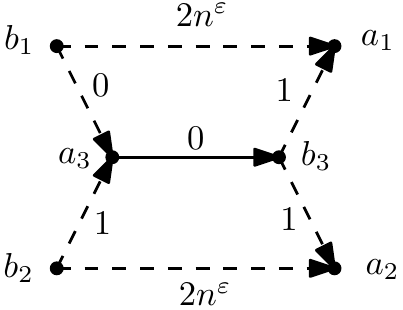}
  \caption{The edges weights constitute slacks. If the dual weights of the free vertices $b_1$ and $b_2$ were both increased by $n^\eps$ then a minimum slack augmenting path $\langle b_1, a_3, b_3, a_1 \rangle$ may be found. However, to maintain feasibility, the dual weight of $a_1$ would have to decrease by $n^\eps - 1$. Negative dual weights for free vertices negatively affects convergence. }
\label{fig:NotAllAtOnce}
\end{figure}

Such an  aggressive increase of dual weights  in arbitrary graphs is incredibly challenging. For an example difficulty, see Figure \ref{fig:NotAllAtOnce}. For graphs that admit an $r$-clustering, however, we are able to apply this strategy and obtain faster matching algorithms. First, we allow an error of $\delta_{uv} > \sqrt{r}$ for each boundary edge $(u,v)$. Similar to Asathulla~\etal,~\cite{soda-18} we guarantee that this introduces a total error of at most $O(n)$ which is acceptable in the scaling paradigm.  Each iteration, we increase the dual weights of free vertices of $B$ by $\sqrt{r}$, using the errors on the boundaries of pieces as `buffers' to absorb the change. This technique allows free vertices of $A$ to continue having dual weights of $0$.   }

Next, we present an overview of our algorithm.

\subsection{Overview of the algorithm} 
Our algorithm uses a bit-scaling framework similar to that of the AKLR-Algorithm. Algorithm~\ref{alg:hmfree} provides an overview of each scale of our algorithm. We split Algorithm~\ref{alg:hmfree} into three steps. The \emph{first step}, corresponding to line 2, we call the preprocessing step. We reference lines 3--6 as the \emph{second step} and lines 7--9 as the \emph{third step}. Note that lines 1--3 and 7--10 are almost identical (with minor differences in implementation) to lines 1--3 and 4--7 of Algorithm~\ref{alg:planar}. 
Similar to the AKLR-Algorithm, after the preprocessing step, we have $O(n/\sqrt{r})$ free vertices remaining. The second step  takes this matching, iteratively calls \fastmatch\ and returns a matching with only $O(\sqrt{n}/r^{1/4})$ free vertices. An execution of \fastmatch\ is similar to a phase of the GT-Algorithm.
 In lines 7--9, we match the remaining free vertices by simply finding one augmenting path at a time.\\

\noindent We give an overview of the second step next. For full details, see Section \ref{subsec:step2}.
\begin{itemize}
\item We associate a slack with every edge of the residual graph $G$ and its compressed representation $H$.  The \emph{projection} of  an edge $(u,v)$ of the compressed residual graph $H$ is a path with the smallest total slack between $u$ and $v$ inside any piece $\R_j$ of the $r$-clustering. The slack of edge $(u,v)$ is simply the total slack on the projection.  For any edge (in $G$ or $H$) directed from $u$ to $v$, we define slack so that if the dual weight of $u$ increases (in magnitude) by $c$ the slack on $(u,v)$ decreases by $c$ whereas if the dual weight of $v$ increases (in magnitude) by $c$, so does the slack. In the residual graph, we say that any edge between two non-boundary vertices is admissible only if it has a zero slack. However, we allow for admissible edges of the residual graph that are incident on a boundary vertex to have a slack of $\sqrt{r}$. As a result, an augmenting path of admissible edges in the residual graph need not be a shortest augmenting path. For an edge $(u,v)$ of the compressed residual graph, we define it to be admissible if it has a slack at most $\sqrt{r}$. 

\item The \fastmatch\ procedure conducts a DFS-style search on the admissible graph of $H$ from every free vertex $v$. For every vertex  of $H$ that does not lead to an augmenting path, this search procedure raises its dual weight (magnitude)  by $\sqrt{r}$. As a result, the DFS either finds an augmenting path and matches $v$ or raises the dual weight of $v$ by $\sqrt{r}$ as desired.
When the search procedure finds an augmenting path $P$ of admissible edges in $H$, it adjusts the dual weights (using a procedure called \sync) and projects $P$  to an admissible augmenting path in the residual graph of $G$.
\item   Unlike the GT or HK algorithms, our algorithm immediately augments along this admissible path and does not throw away vertices visited by the search.  This causes the augmenting paths computed during the \fastmatch\ procedure to not necessarily be vertex-disjoint. Furthermore, vertices of the graph $H$ can be visited multiple times within the same execution of the \fastmatch\ procedure. We describe an example of such a scenario; see Figure \ref{fig:create-new} and the discussion at the end of this section for more details. Due to these revisits, unlike the GT-Algorithm (where each phase takes $O(m)$ time), we cannot bound the time taken by the \fastmatch\ procedure.  We note, however, that every vertex visited by the search either lies on an augmenting path or has its dual weight (magnitude) increased by $\sqrt{r}$ . Therefore, any vertex $v$ of $H$ that is unsuccessfully visited  $\sqrt{n}/r^{1/4}$ times by the search will have a dual weight magnitude of at least $\sqrt{n}r^{1/4}$. In order to limit the number of visits of a vertex, whenever a vertex $v \in B$ whose dual weight exceeds $\sqrt{n}r^{1/4}$ is visited, the search procedure immediately computes a projection $\dir{P}$ of the current DFS search path. This projection forms an alternating path from a free vertex to $v$ in $G$. After setting $M \leftarrow M \oplus \dir{P}$, $v$ is a free vertex with dual weight (magnitude) at least $\sqrt{n}r^{1/4}$ . The vertex $v$ is then marked as inactive and will not participate in any future execution of the \fastmatch\ procedure.  The second step of the algorithm ends when all remaining free vertices become inactive (this happens after $O(\sqrt{n}/r^{1/4})$ executions of \fastmatch). Using the fact that the optimal matching has cost of $O(n)$, we can show that the number of inactive free vertices cannot exceed $\tilde{O}(\sqrt{n}/r^{1/4})$.  

\item Due to the fact that an augmenting path $P$   in $H$ computed by the search procedure need not be a path with a minimum cost, its projection $\dir{P}$ may be a non-simple path in the underlying graph $G$. To avoid creating such non-simple projections, when our DFS style search encounters a cycle $C$, the algorithm computes its projection $\dir{C}$ and flips the edges on the cycle immediately by setting $M \leftarrow M\oplus \dir{C}$. This modification requires us to update all pieces that contain edges of $\dir{C}$. When the search finds an augmenting path (resp. cycle) in the admissible graph of $H$, due to the active elimination of cycles, we can guarantee that its projection indeed a simple path (resp. cycle) of admissible edges.

\begin{figure}
  \centering
  \includegraphics[width=\textwidth,scale=1]{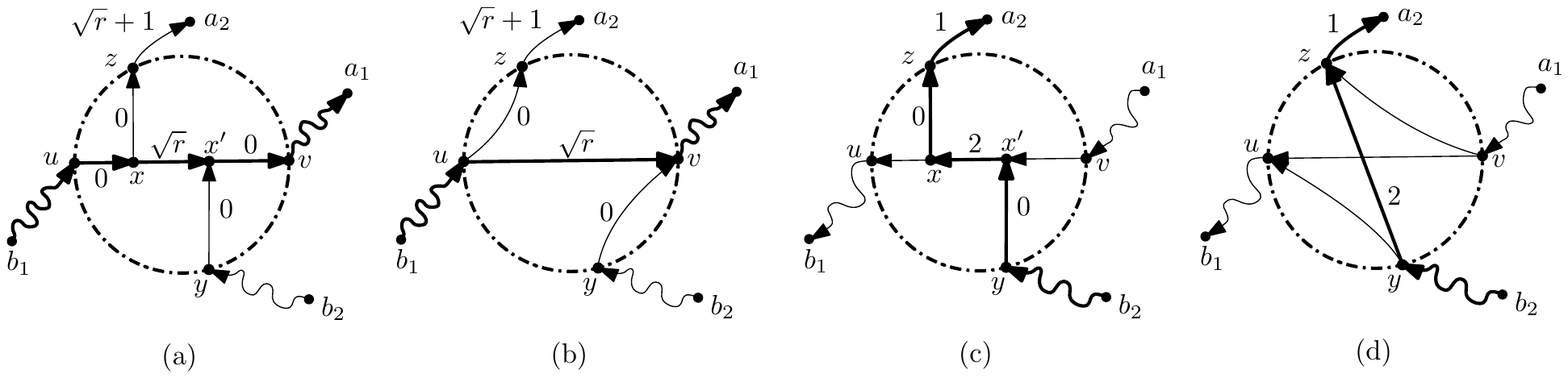}
\caption{
An example where a vertex $z$ is visited multiple times during a phase of \fastmatch. The status of a piece in the $r$-clustering and its compressed representation is given both before and after augmentation, with edge slacks. By augmenting along a path containing $(u,v)$, a new path is created using $(y, z)$, causing the revisit.
}
\ignore{
  \caption{This is an example situation where the algorithm may backtrack from a vertex multiple times in a phase. Here, (a) gives the initial state of the piece in $G$ and (b) gives the initial state of the piece in $H$. Consider if the end of the DFS search path is $u$. The edge $(u,z)$ in $H$ is explored, but no admissible outgoing edges from $z$ are found, and the algorithm backtracks from $z$. Then, the algorithm explores $(u,v)$ and eventually finds an augmenting path. However, later, a new search from a different free vertex of $B$ could begin and visit $y$. Since the directions of path from $u$ to $v$ in the residual graph were flipped, a new path from $y$ to $z$ was created. Therefore, $z$ could be visited again through $y$.}
}
\label{fig:create-new}
\end{figure}
\item The total time taken by the \fastmatch\ procedure can be attributed to  the time taken by the DFS style search to find augmenting paths, alternating paths, and alternating cycles (search operations) and the time taken to project, flip the edges, and update the compressed residual graph for the paths an cycles (update operations). The dual weight of any vertex cannot exceed $\sqrt{n}r^{1/4}$ and so, every vertex is visited by the search $O(\sqrt{n}/r^{1/4})$ times.  Since the compressed residual graph has $O(n)$ edges, the total time taken by the search operations is bounded by $\tilde{O}(n^{3/2} / r^{1/4})$. Like Asathulla \etal\ we must argue that the total number of compressed graph updates is small. However, unlike their algorithm, we must also account for alternating paths and cycles instead of just augmenting paths. Despite this, we show that the total length of all the cycles, alternating paths and augmenting paths computed in the compressed graph $H$  does not exceed $O((n/\sqrt{r}) \log n)$. The total time for the update operations is $\tilde{O}(mr)$. 
\ignore{
\item At the end of the second step, all free vertices of $B$ are inactive. Therefore, at the beginning of the third step, every free vertex has a dual weight of at least $\sqrt{n}r^{1/4}$. \ignore{Using an argument similar to~\eqref{eq:convergence}, we show that there are no more than $\tilde{O}(\sqrt{n}/r^{1/4})$ free vertices remaining.}Since the optimal matching has cost $O(n)$, we can show that there are $\tilde{O}(\sqrt{n}/r^{1/4})$ free vertices remaining. Each of these vertices are matched in $\tilde{O}(m)$ time by iteratively computing an augmenting path and augmenting the matching along the path in $\tilde{O}(m\sqrt{n}/r^{1/4})$ time.   }
   \end{itemize}

Summing up over all scales, and setting $r=n^{2/5}$ gives the claimed running time. 

\paragraph{Discovery of augmenting paths that are not vertex-disjoint.}
We present a scenario where an execution of \fastmatch\ procedure causes a vertex $z$ of $H$ to be visited twice leading to the discovery of two augmenting paths that are not vertex-disjoint; see Figure \ref{fig:create-new} (edge weights represent slacks). Here, (a) and (c) represent the states of the residual graph within a piece before and after augmenting along the first path. (b) and (d) give the compressed graph counterparts of (a) and (c) respectively. $b_1, b_2$ are free vertices of $B$ and $a_1,a_2$ are free vertices of $A$. Suppose that the \fastmatch\ procedure begins a DFS search from $b_1$, eventually adds $u$ to the search path, and explores the admissible edge $(u,z)$ with slack $0$ (Figure 2(b)).  The search procedure now adds $z$ to the search path. Now suppose the execution of DFS-style search from $z$ does not lead to any augmenting path of admissible edges and the search backtracks from $z$. At this time, the dual weight (magnitude) of $z$ increases by $\sqrt{r}$ making the edge $(z,a_2)$ admissible. After backtracking from $z$, the search proceeds along $(u,v)$ and finds an augmenting path from $b_1$ to $a_1$ in $H$. In order to create an admissible projection, that algorithm increases the dual weight (magnitude) of $u$ by $\sqrt{r}$ and as a result reduces the slack on $(x,x')$ to $0$.  After augmentation, the new edge $(y,z)$ (in Figure 2(d)) is also admissible. Now suppose, in the same execution of \fastmatch\ a different DFS begins from $b_2$ and eventually adds the vertex $y$ to the path. Both $(y,z)$ and $(z,a_2)$ are admissible. So, our search visits $z$ for a second time and also finds an augmenting path to $a_2$. Note that both augmenting paths found use $(x,x')$. Furthermore, $z$ was visited multiple times during a single execution of \fastmatch\ procedure. 
\ignore{
 The algorithm searches from $u$ and explores the admissible edge $(u,z)$ of $H$, eventually resulting in a backtrack from $z$. Upon backtracking from $z$, the dual weight of $z$ increases in magnitude, making all incoming edges inadmissible. The search eventually finds an augmenting path that uses the edge $(u,v)$. Prior to augmentation, the slack on $(u,v)$ and $(u,z)$ is reduced nearly to 0. After augmentation, a new path from $y$ to $z$ is created, and $z$ can be visited again via the new admissible edge $(y,z)$ of $H$ in a subsequent search from a different free vertex.
}

\medskip
\noindent
{\bf Organization.} The remainder of the paper is organized as follows.
In Section~\ref{sec:background} we present the background on the matching algorithms that serve as the building blocks for our approach. In particular, we present a new variant of Gabow and Tarjan's algorithm that our search procedure is based on. In Section~\ref{sec:preliminaries} we define the notion of an $\R$-feasible matching based on an $r$-clustering and related concepts. Section \ref{sec:compressed-graph} describes our definition of the compressed graph and compressed feasibility, which differs slightly from that in \cite{soda-18}. In Section~\ref{sec:scaling}, we describe the overall scaling algorithm. In section~\ref{algorithmforscale}, we present our algorithm for each scale. In Section~\ref{analysis} we prove the correctness and efficiency of our algorithm. In Section~\ref{sec:planar}, we combine our algorithm with shortest path data structures designed for planar graphs to achieve an $O(n^{6/5})$ time algorithm. We conclude in Section~\ref{sec:conclusion}. Proofs of some of the lemmas are presented in the appendix.

\section{Background}
\label{sec:background}
In this section, we  present definitions relevant to matching and give an overview of the Hungarian Algorithm. We use ideas from a new variant of the GT-Algorithm which we present in Section \ref{subsec:gt}.
\paragraph{Preliminaries on matching.}
Given a matching $M$ on a bipartite graph, an \emph{alternating path} (resp. cycle) is a simple path (resp. cycle) whose edges alternate between those in $M$ and those not in $M$. We refer to any vertex that is not matched in $M$ as a \emph{free vertex}. Let $A_F$ (resp. $B_F$)be the set of free vertices of $A$ (resp. $B$).  An \emph{augmenting path} $P$ is an alternating path between two free vertices. We can \emph{augment} $M$ by one edge along $P$ if we remove the edges of $P \cap M$ from $M$ and add $P \setminus M$ to $M$. After augmenting, the new matching is given by $M \leftarrow M \oplus P$, where $\oplus$ is the symmetric difference operator. For a matching $M$, we define a directed graph  called the residual graph $\Gr_M(A\cup B, E_M)$. We represent a directed edge from $a$ to $b$ as $\dir{ab}$. For every edge $(a,b) \in E\cap M$, we have an edge $\dir{ab}$ in $E_M$ and for every edge $(a,b)\in E\setminus M$, there is an edge $\dir{ba}$ in $E_M$. Note that $G$ and $G_M$ have the same vertex set and edge set with the edges of $G_M$ directed depending on their membership in the matching $M$. For simplicity in presentation, we treat the vertex set and the edge set of $G$ and $G_M$ as identical.  So, for example, a matching $M$ in the graph $G$ is also a matching in the graph $G_M$. It is easy to see that a  path $\dir{P}$ in $\Gr_M$ is a directed path if and only if this path is an alternating path in $\Gr$.

\paragraph{Hungarian Algorithm.}In the Hungarian Algorithm,  for every vertex $v$ of the graph $\Gr$, we maintain a dual weight $y(v)$. A \emph{feasible} matching consists of a matching $M$ and a set of dual weights $y(\cdot)$ on the vertex set such that for every edge $(u,v)$ with $u \in B$ and $v\in A$, we have

\begin{eqnarray}
y(u) + y(v) &\le & \dist(u,v), \label{eq:feas1}\\
y(u) + y(v) & = & \dist(u,v)\quad \mbox{for } (u,v) \in M. \label{eq:feas2}
\end{eqnarray}

\noindent For the Hungarian algorithm, we define the \emph{net-cost} of an augmenting path $P$ as follows:

$$\phi(P)=\sum_{(a,b)\in P\setminus M}\dist(a,b) -\sum_{(a,b)\in P\cap M}\dist(a,b).$$
We can also interpret the net-cost of a path as the increase in the cost of the matching due to augmenting it along $P$, i.e., $\phi(P)=\wt(M\oplus P) - \wt(M)$. We can extend the definition of net-cost to alternating paths and cycles in a straight-forward way.

The Hungarian algorithm starts with $M=\emptyset$. In each iteration, it computes a minimum net-cost augmenting path $P$ and updates $M$ to $M\oplus P$. The algorithm terminates when we there is a perfect matching.
During the course of the algorithm, it maintains the invariant that there are no alternating cycles with negative net-cost. 

It can be shown that any perfect matching $M$ is a min-cost perfect matching iff there is no alternating cycle with negative net-cost with respect to $M$. Any perfect matching $M$ that satisfies the feasibility conditions (\ref{eq:feas1}) and (\ref{eq:feas2}) has this property. Thus it is sufficient for the Hungarian algorithm to find a  feasible perfect matching. 

In order to find the minimum net-cost augmenting path, the Hungarian algorithm uses a simple Dijkstra-type search  procedure called the Hungarian Search. The Hungarian Search computes the minimum net-cost path as follows. For any edge $(u,v)$, let $\dist(u,v)-y(u)-y(v)$ be the \emph{slack} of $(u,v)$. Consider a directed graph $\Gr_M'$, which is the same as the residual graph $\Gr_M$  except the cost associated with each edge is equal to its slack. It can be shown that the minimum weight directed path in $\Gr_M'$ corresponds to the minimum net-cost augmenting path in $\Gr$. Since the slack on every edge is non-negative by feasibility condition \eqref{eq:feas1},  the graph $\Gr_M'$ does not have any negative cost edges. Therefore, we can simply use Dijkstra's algorithm to compute the minimum net-cost augmenting path. After this, the dual weights are updated in such a way that the invariants are satisfied. See~\cite{hungarian_56} for details.\tnote{There may be a better reference for Hungarian search than this one?} The Hungarian algorithm computes $O(n)$ augmenting paths each of which can be computed by Hungarian Search in $O(m)$ time. Therefore the total time taken is $O(mn)$. At any stage of the algorithm, the matching $M$ and the set of dual weights $y(\cdot)$  satisfy the following invariants:

\begin{itemize}
\item[(i)] $M$ and the set of dual weights $y(\cdot)$ form a feasible matching.
\item[(ii)] For every vertex $b\in B$, $y(b) \ge 0,$ and if $b$ is a free vertex then the dual weight $y(b)=\max_{v\in B} y(v)$.
\item[(iii)] For every vertex $a\in A$, $y(a) \le 0,$ and if $a$ is a free vertex then the dual weight $y(a)=0$.
\end{itemize}

\noindent  

The dual weights of vertices in $B$ are always non-negative and the dual weights of vertices in $A$ are non-positive, a property also satisfied by our algorithm. Next, we introduce a modified version of  Gabow and Tarjan's algorithm. Some elements of our algorithm will be based on this variant.

\subsection{Modified Gabow-Tarjan Algorithm}
\label{subsec:gt}
We begin by giving an overview of the algorithm and describe the steps it takes between two successive scales. After that, we present the algorithm inside each scale. The steps taken by this algorithm  are  different from Gabow and Tarjan's original algorithm and its correctness requires a proof. Instead of providing a proof of correctness here, we adapt this algorithm for our setting and provide a proof for our case directly. 

As in the Hungarian algorithm, the Gabow-Tarjan algorithm also maintains a dual weight for every vertex of $\Gr$. We  define a \emph{$1$-feasible matching} to consist of a matching $M$ and set of dual weights $y(v)$ such that for every edge between $u \in A$ and $v\in B$ we have

\begin{eqnarray}
& &y(u) + y(v) \le \dist(u,v) + 1, \label{eq:feas3}\\
& &y(u) + y(v) \ge  \dist(u,v) -1\quad \mbox{for } (u,v) \in M. \label{eq:feas4}
\end{eqnarray}
\ignore{The above conditions with the $+1$ removed from (\ref{eq:feas1}) are identical to the dual feasibility conditions of the linear program corresponding to optimal bipartite matching.}
 Gabow and Tarjan presented a similar feasibility constraints for the minimum-cost degree constrained subgraph (DCS) of a bipartite  multigraph \cite{gt_sjc89}. Note, however, that our definition of $1$-feasibility  is different from the original definition of   $1$-feasibility as given by Gabow and Tarjan:
\begin{eqnarray}
& &y(u) + y(v) \le \dist(u,v) + 1, \label{eq:gtfeas5}\\
& &y(u) + y(v) =  \dist(u,v) \quad \mbox{for } (u,v) \in M. \label{eq:gtfeas6}
\end{eqnarray}
 A \emph{$1$-optimal matching} is a perfect matching that is $1$-feasible. We define the \emph{slack} of an edge to be $\dist(u,v)-y(u) - y(v) + 1$ if $(u,v) \not\in M$ and $y(u) +y(v) - \dist(u,v) +1$ if $(u,v) \in M$. The following lemma relates $1$-optimal matchings to the optimal matchings.
\begin{lemma}
\label{optrel}
For a bipartite graph $\Gr(A\cup B, E)$ with an integer edge cost function $\dist$, let $M$ be a $1$-optimal matching and $M_{\opt}$ be the optimal matching. Then,
$\dist(M) \le \dist(M_{\opt}) + 2n.$
\end{lemma}

For every $(a,b) \in E$, suppose we redefine the edge weight to be $\dist^*(a,b) = (2n+1)\dist(a,b)$. This uniform scaling of edge costs preserves the set of optimal matchings and guarantees that any sub-optimal matching has a cost that is at least $2n+1$ greater than the optimal cost.
Thus, a $1$-optimal matching with the edge weights $\dist^*(\cdot,\cdot)$ corresponds to an optimal matching with the original edge weights $\dist(\cdot,\cdot)$. 

We now describe  the bit-scaling paradigm. For any edge $(u,v)$, let $b_1, b_2\ldots b_\ell$ be the binary representation of $\dist^*(u,v)$. Let $\dist_{i}(u,v)$ correspond to the most significant $i$ bits of $\dist^*(u,v)$. The Gabow-Tarjan Algorithm consists of \emph{scales}. The algorithm for any scale $i$  takes  a bipartite graph on $A, B,$ with a cost function $\dist_{i}(\cdot,\cdot)$, and a set of dual weights $y(v)$ for every vertex $v\in A \cup B$  as input. Let $\dist^i(u,v) = \dist_i(u,v) - y(u)-y(v)$. Then, $\dist^i(\cdot,\cdot)$ satisfies the following at the beginning of the $i$th scale:
\begin{itemize}
\item For every edge $(u,v)$, $\dist^i(u,v) \ge 1$, and,
\item The cost of a $1$-optimal matching with respect to $\dist^i(\cdot,\cdot)$ is $O(n)$.

\end{itemize}
Given such an input, the algorithm for each scale returns a perfect matching $M$ and a set of dual weights $y(\cdot)$ so that $M, y(\cdot)$ is a $1$-optimal matching.   The input to the  first scale is the graph $G(A \cup B, E)$ with the cost $\dist_1(\cdot,\cdot)$ and a set of   dual weights of $-1$ on every vertex of $B$ and  dual weights of $0$ on every vertex of $A$.  It is easy to see that $\dist^1(\cdot,\cdot)$ satisfies the two conditions. For any scale $i$, the algorithm  computes a matching $M$ and dual weights $y'(\cdot)$ so that $M, y'(\cdot)$ is a $1$-optimal matching with respect to the costs $c^i(\cdot)$. For every vertex $v \in A\cup B$, let $y(v)$ be the sum of the dual weight $y'(v)$ and the initial dual weight assigned at the start of the scale $i$ to $v$. Then, it can be shown that $M, y(\cdot)$ is $1$-optimal with respect to $\dist_i(\cdot,\cdot)$. 

For any $i \ge 1$, we use the  $1$-optimal matching $M, y(\cdot)$  returned by the algorithm for scale $i$ to generate an input for scale $i+1$ as follows. First, we set the slack (with respect to $\dist_i(\cdot,\cdot)$) of every edge of the $1$-optimal matching $M$ of scale $i$ to $0$. For any edge $(u,v) \in M$, this can be done by reducing the dual weight of one of its vertices, say $u$ by $s(u,v)$, i.e., $y(u) \leftarrow y(u) - s(u,v)$. Note that any reduction in dual weight of $u$ does not violate~\eqref{eq:feas3} or~\eqref{eq:feas4} and so $M, y(\cdot)$ remains $1$-optimal.    
After this, we transfer the dual weights from scale $i$ to scale $i+1$  by simply setting, for any vertex $v \in A \cup B$, $y(v) \leftarrow 2y(v)-2$. Therefore, at the beginning of scale $i+1$, the  reduced cost $\dist^{i+1}(u,v)=\dist_{i+1}(u,v) - y(u)- y(v)$ on every edge is at least $2$ and every edge $(u,v)$ in the $1$-optimal matching $M$ of scale $i$ has  $\dist^{i+1}(u,v) \le 6$. So, the cost of an optimal matching and also any $1$-optimal matching with respect to the $\dist^{i+1}(\cdot,\cdot)$  is upper bounded by $O(n)$ as desired.

 \paragraph{Algorithm for each scale.}Now, we present the algorithm for scale $i$. We refer to an edge $(u,v)$ as \emph{admissible} if
it has a slack of $0$, i.e., \begin{eqnarray*}
 y(u) + y(v) &=& \dist^{i+1}(u,v) +1 \qquad \textrm{if $(u,v) \not\in M,$}\\
y(u) + y(v) &=& \dist^{i+1}(u,v) - 1 \qquad \textrm{if $(u,v) \in M.$} 
 \end{eqnarray*} 
An \emph{admissible graph} is the set of admissible edges. The algorithm runs in two stages. The first stage of the algorithm executes $O(\sqrt{n})$ iterations. In each iteration, the algorithm  initiates a DFS search from each free vertex in $B_F$. If it finds an augmenting path in the admissible graph, then it augments the matching right away.  Consider a DFS\ initiated from $b \in B_F$. Let $P = \langle u_1, ..., u_k\rangle$ be the current path of the DFS with $u_1=b$.

\begin{itemize}
\item  If there is no admissible edges outgoing from $u_k$, then remove $u_k$ from $P$. If $u_k \in B$, set $y(u_k) \leftarrow y(u_k)+1$. Otherwise, $u_k \in A$, and set $y(u_k) \leftarrow y(u_k)-1$.

\item Otherwise, suppose there is an admissible edge from $u_k$ to  a vertex $v$. If $v$ is a free vertex of $A$, then the algorithm has found an augmenting path of admissible edges and it augments the matching along the path. Otherwise, it adds $v$ to the path as vertex $u_{k+1}$ and continues the search from $v$.

\end{itemize} This completes the description of a DFS search for augmenting paths.

In a single iteration, a DFS\ is initiated from every free vertex of $B_F$. At the end of the iteration, it can be shown that there are no augmenting paths consisting of only admissible edges and the dual weights of every free vertex $b \in B_F$ are increased by exactly $1$ when $b$ is removed from $P$.  After $\sqrt{n}$ iterations, the dual weights of vertices in $B_F$ will be  $\sqrt{n}$.  It can be shown that, the sum of dual weights of $A_F$ and $B_F$ cannot exceed the cost of a $1$-optimal matching, i.e., $O(n)$. Furthermore, the dual weights of vertices of $A_F$ are maintained as 0. So, 

$$ \sum_{v\in B_F\cup A_F} y(v) \ge |B_F|\sqrt{n} = O(n).$$  This bounds  $|B_F|$  by $O(\sqrt{n})$. After this, we iteratively (for $O(\sqrt{n})$ iterations) execute Hungarian Search in $O(m)$  time to find an augmenting path and augment the matching.  The total computation time for a single scale is $O(m \sqrt{n})$, and summed over all $O(\log (nC))$ scales, total time taken is  $O(m \sqrt{n}\log{(nC)})$. Gabow and Tarjan show that   the total length of all the augmenting paths found  is $O(n \log n)$, and a similar argument can be applied here.

There are two aspects in which this description of Gabow and Tarjan's algorithm differs from the original GT-Algorithm. First, in the first $\sqrt{n}$  phases, we avoid doing a Hungarian Search and only conduct several partial depth-first searches. Second, for any partial DFS, all vertices of $B$ (resp. $A$) that are visited by the DFS\ but do not lead to an augmenting path undergo an increase (resp. decrease) in their dual weight. The dual weights of $B$ start at $0$ and only increase during the algorithm. Therefore, dual weights of $B$ are non-negative. Similarly, the dual weights of $A$ start at $0$ and may reduce during the algorithm. So, dual weights of vertices in $A$ are\ non-positive. So,  if the partial DFS   visits any vertex, except for vertices that are along any augmenting path $P$, there will be an increase in the magnitude of its dual weight by $1$. Whenever a vertex is visited, perhaps each of its $\mathrm{deg}(v)$ neighbors could be explored. By careful analysis, the time taken to execute the first $\sqrt{n}$ phases can also be bounded by $O(\sum_{v \in A\cup B}\mathrm{deg}(v)|y(v)|+ \sum_{j=1}^n |P_j|) = O(m\sqrt{n} + n\log n)$, where $P_j$ is the $j^{th}$ augmenting path computed by the algorithm.

In the next section, we introduce definitions for feasibility, admissiblity and net-cost as used in our algorithm.

\section{Preliminaries}
\label{sec:preliminaries}
In this section, we introduce conditions for feasibility, admissibility and the definitions of slack and net-cost as used in our algorithm. Using these definitions, we establish critical properties (Lemma~\ref{lem:admissible-feasible} and~\ref{lem:ynetcost}) of paths in an admissible graph. These definitions and properties are based on an $r$-clustering, which we formally introduce next.

\paragraph{$r$-clustering of a graph $G$.} Consider a partitioning of any graph $G(V,E)$ into $l$  edge-disjoint subgraphs called \emph{pieces} and denoted by  $\R(G) =\{\R_1(V_1, E_1),\ldots, \R_l(V_l, E_l)\}$.  For each piece $\R_j$, the vertex set $V_j$ and the edge set $E_j$ of any piece $\R_j$ is the set  $E_j = \{(a,b) \mid a,b \in V_j\}$. Furthermore, $\bigcup_{j=1}^l V_j = V$ and $E=\bigcup_{j=1}^l E_j =E$.\  A vertex $v \in V$ which has edges incident from two or more pieces is called a \emph{boundary vertex}. An edge which is adjacent to one or more boundary vertices is a \emph{boundary edge}. Any other edge is an \textit{interior} edge. Let $\boundary_j$ denote the set of boundary vertices of $\R_j$ and let $\boundary = \bigcup_{j=1}^l \boundary_j$. 
\begin{definition}
\label{def:r-division}

A partition $\R(G) = \{\R_1(V_1, E_1),\ldots, \R_l(V_l, E_l)\}$ of a graph $G$
is an   $r$-clustering  if $l=O(n/\sqrt{r})$, for each $\R_j$, $|V_j|=O(r)$ and $|\boundary_j|=O(\sqrt{r})$. Furthermore, the total number of boundary vertices is $|\boundary|\le\sum_{j} |\boundary_j| = O(n/\sqrt{r})$. 
 Let $k_1, k_2$ be  constants such that $ \max_{j} {|\boundary_j|} \le k_1\sqrt{r}$ and  $|\boundary| \le k_2 n / \sqrt{r}$. Also, let $|V_j|=n_j$ and $|E_j|=m_j$.

\end{definition}

\noindent For any choice of $r$, an $r$-clustering can be computed on $K_h$-minor free graphs in $O(m \log n + n^{1 + \epsilon}\sqrt{r})$ time~\cite{wulff2011separator}.  Note that the $r$-clustering computed in \cite{wulff2011separator} has $\tilde{O}(n/\sqrt{r})$ pieces, $\tilde{O}(n/\sqrt{r})$ boundary nodes (counting multiplicities), and $\tilde{O}(\sqrt{r})$ boundary nodes per piece. For simplicity in exposition, we present for the case without the $\poly(\log{n})$ terms. After presenting our algorithm, we briefly describe in Section~\ref{subsec:hminor}, the necessary changes to account for these $\poly(\log{n})$ terms.  

\paragraph{Convention for notation.} Throughout this paper, we will deal with a bipartite graph $G$. For any vertex $u \in A\cup B$, throughout this paper we use $\lambda_u = -1$ if $u \in A$ and $\lambda_u = 1$ if $u \in B$.  For simplicity of analysis, we assume without loss of generality that $\sqrt{r}$ is an integer. Given a matching $M$ and a set of dual weights $y(\cdot)$, we refer to its residual graph by $G_M$. Note  that the vertex and edge sets of $G$ and $G_M$  are identical (except for the directions) and a matching, alternating path or an alternating cycle in $G$ is also a matching, directed path or a directed cycle in $G_M$. So, if there is any subset  $P$ of edges in $G$, we will also use $P$ to denote the same subset of edges in $G_M$, the directions of these edges are determined by whether or not an edge is in $M$ .  We will define a net-cost for an alternating path (or cycle) $P$ in our algorithm and denote it by $\phi(P)$. Any directed path or cycle in $G_M$ will inherit its net-cost from $G$. During the course of our algorithm, for any weighted and directed graph $K$, we will  use the notation $K'$ to be the graph identical to $K$ where the cost of any directed edge in the graph is replaced by its slack.
Recall that an $r$-clustering $\R(G)$ partitions the edges of $G$. Since $G$, $G_M$ and $G_M'$ have the same underlying set of edges, $\R(G)$ can be seen as an $r$-clustering of $G_M$ and $G_M'$ as well.

We next introduce a notion of feasibility that is based on an $r$-clustering. We assume that we are given an $r$-clustering, $\R(\Gr) =\{\R_1(V_1, E_1),\ldots, \R_l(V_l, E_l)\}$ with ($l = O(n/\sqrt{r})$).  Recall that we denote by $\boundary$ the boundary vertices of $\R$ and by $\boundary_j$ the set of boundary vertices in $\R_j$. For every edge $uv \in E_j$, we define a $0/1$ indicator variable $i_{uv}$ to be $1$ iff $uv$ is a boundary edge in $\R$. We define a value $\delta_{uv}$ to be $\max\{1, i_{uv}\frac{m_j n}{m\sqrt{r}}, 2 i_{uv}\sqrt{r}\}$. For any edge induced subgraph $G^*(V^*,E^*)$ of $\Gr(A\cup B,E)$, we say that a matching $M\subseteq E^*$ and a set of dual weights $y(\cdot)$ on vertices of $V^*$ are 
\emph{$\R$-feasible } if every edge $(u,v) \in E^*$ satisfies the following two conditions:

\begin{eqnarray}
& &y(u) + y(v) \le \dist(u,v) + \delta_{uv} \quad \mbox{for } (u,v) \notin M, \label{eq:feas5}\\
& &y(u) + y(v) \ge \dist(u,v) - \delta_{uv}  \quad \mbox{for } (u,v) \in M. \label{eq:feas6}
\end{eqnarray}

The algorithm in~\cite{soda-18} used a very similar definition except that it uses a different $\delta_{uv}$. Additionally, here we allow matching edges to violate the traditional feasibility constraints.  We can define the slack of an edge $(u,v)$ denoted by $s(u,v)$ with respect to the dual weights as
\begin{eqnarray}
& & s(u,v)=\dist(u,v) + \delta_{uv} - y(u) - y(v)\quad \mbox{for } (u,v) \notin M,\label{eq:slack1}\\
& &s(u,v) = y(u) + y(v) - \dist(u,v) + \delta_{uv}  \quad \mbox{for } (u,v) \in M. \label{eq:slack2}
\end{eqnarray}
We define \emph{admissible} edges next. Any boundary edge $(u,v)$ is admissible only if $s(u,v) \le \sqrt{r}$. Any edge $(u,v)$ that does not border a boundary vertex is admissible only if $s(u,v)=0$. Note that in both cases an admissible edge satisfies  $$\delta_{uv} \ge 2s(u,v).$$Our algorithm will compute admissible paths and cycles with respect to $\R$-feasible matchings. For any such path $P$, we update the matching $M$ by setting $M \leftarrow M\oplus P$. The following lemma shows that the new matching after such an operation remains $\R$-feasible.

\begin{lemma}
\label{lem:admissible-feasible}
Given an $\R$-feasible matching $M,y(\cdot)$ on any bipartite graph $G(V,E)$, let $\dir{P}$ be a path or cycle in $G_M$ consisting of only admissible edges. Then, $M' \leftarrow\ M \oplus \dir{P}, y(\cdot)$ is also an $\R$-feasible matching. Furthermore, the slack on every edge $(u,v)$ of $\dir{P}$ with respect to $M', y(\cdot)$ is at least  $\delta_{uv}$. 
\end{lemma} 
\begin{proof}
First, consider any edge $(u,v) \in P\cap M$.\ Suppose $(u,v)$ is not a boundary edge, so $\delta_{uv} = 1$. Then the slack $s(u,v)$ is zero and 
\begin{eqnarray}
y(u)\ +\ y(v)\ - \dist(u,v)\ + \delta_{uv}  = 0, \\
y(u)\ +\ y(v)\ = \dist(u,v)\ - \delta_{uv}\ < \dist(u,v) + \delta_{uv}.
\end{eqnarray}
Thus, $(u,v)$ is $\R$-feasible with respect to $M'$ and the slack on the edge with respect to $M'$ is at least $\delta_{uv}$. Otherwise, $(u,v)$ is a boundary edge, and $\delta_{uv} \geq 2\sqrt{r}$. Since $(u,v) $ is an admissible edge in the matching $M$, $(u,v)$ will satisfy \begin{eqnarray}
s(u,v) =  y(u)\ +\ y(v)\ - \dist(u,v)\ + \delta_{uv}\ \leq \sqrt{r}, \\
y(u)\ +\ y(v)\ \le \dist(u,v)\ - \delta_{uv}\ +\ \sqrt{r} < \dist(u,v)\ + \delta_{uv}.
\end{eqnarray}
So, feasibility condition \eqref{eq:feas5} holds for $(u,v)$ with respect to $M'$ and the slack of $(u,v)$ is at least $\delta_{uv}$.

Next, consider $(u,v) \in P\setminus M$. Suppose $(u,v)$ is not a boundary edge, so $\delta_{uv}=1$. Then we have
 \begin{eqnarray*}
 \dist(u,v) + \delta_{uv} - y(u) - y(v) &=& 0,\\
 y(u) + y(v) = \dist(u,v) + \delta_{uv} &>& \dist(u,v) - \delta_{uv}.
 \end{eqnarray*}
 So, feasibility condition \eqref{eq:feas6} holds for $(u,v)$ with respect to the matching $M'$ and dual weights $y(\cdot)$ and the slack of $(u,v)$ is at least $\delta_{uv}$. Otherwise, $(u,v)$ is a boundary edge, and $\delta_{uv} \ge 2\sqrt{r}$. Then the admissible edge $(u,v)$ will satisfy \begin{eqnarray*}
  \dist(u,v)\ + \delta_{uv} - y(u) - y(v) &\leq& \sqrt{r}, \\
y(u)\ +\ y(v)\ \ge  \dist(u,v)\ + \delta_{uv}\ - \sqrt{r} &>& \dist(u,v) - \delta_{uv}.
\end{eqnarray*}
So, feasibility condition \eqref{eq:feas6} is satisfied with respect to the matching $M'$ and dual weights $y(\cdot)$ and the slack of $(u,v)$ is at least $\delta_{uv}$.  Since all edges of $G$ satisfy the $\R$-feasibility conditions \eqref{eq:feas5} and \eqref{eq:feas6}, the matching $M'$ and $y(\cdot)$ is $\R$-feasible.
\end{proof}

An \emph{$\R$-optimal} matching is a perfect matching that is $\R$-feasible. Our algorithm, for the graph $G(A\cup B,E)$ and the $r$-clustering $\R(G)$,  computes an $\R$-optimal matching $M$ along with its dual weights $y(\cdot)$. Note that the notion of $\R$-feasible matching can be defined for any edge induced subgraph of $G$. In the context of this paper, the  only induced subgraphs that we consider are pieces from the set $\{\R_1,\ldots,\R_l\}$ that are given by the $r$-clustering of the graph.  Our algorithm  will maintain an $\R$-feasible matching for each piece. Throughout this paper, we fix the $r$-clustering in the definition of $\R$-feasibility to be $\R(G)$. For any  $\R$-feasible matching, when obvious from the context, we will not explicitly mention the induced subgraph the matching is defined on.

For our algorithm, we will define the  \emph{net-cost} of an edge $(u,v)$, $\phi(u,v)$ as 
\begin{eqnarray}
& &\phi(u,v) = \dist(u,v) + \delta_{uv} \quad \mbox{for } (u,v) \notin M, \label{eq:netcost1}\\
& &\phi(u,v) = -\dist(u,v) + \delta_{uv} \quad \mbox{for } (u,v) \in M. \label{eq:netcost2}
\end{eqnarray}
For any set of edges $S$, we can define the net cost as
$$\phi(S) = \sum_{(u,v) \in S} \phi(u,v),$$ and the total slack of $S$ can be defined in a similar fashion. Our interest is in net costs for the case where $S$ is an augmenting path, alternating path, or alternating cycle. Consider if we have any matching $M$ and let $M' \leftarrow M \oplus S$. Then, 
\begin{equation}
\label{netcost3}
\phi(S) = \dist(M') - \dist(M) + \sum_{(u,v) \in S} \delta_{uv}.
\end{equation}
From the definitions of net-cost and slack, we get the following relation:
\begin{eqnarray}  s(a,b) &=& \phi(a,b) + y(a) + y(b)  \qquad \textrm{if }(a,b) \in M,\\  
s(a,b) &=& \phi(a,b) - y(a) - y(b) \qquad \textrm{ if } (a,b) \not\in M.
\end{eqnarray}

For any vertex $u \in A\cup B$, throughout this paper we use $\lambda_u = -1$ if $u \in A$ and $\lambda_u = 1$ if $u \in B$. Given any directed path (resp. cycle) $\dir{P}$ from a vertex $u$ to a vertex $v$ in $G_M$ if we add the above equations over all the edges of $\dir{P}$, we get

\begin{equation}
\label{slackcost}
\sum_{(a,b) \in \dir{P}} s(a,b) = \sum_{(a,b) \in \dir{P}} \phi(a,b) + \lambda_v y(v) - \lambda_u y(u) = \phi(\dir{P}) +\lambda_v y(v) - \lambda_uy(u).
\end{equation}

Despite allowing for slacks in the admissible edges, the following lemma shows that for any admissible path the difference in the dual weights of its first and last vertex is related to the change in the matching cost and the number of pieces visited by this path.  
\begin{lemma} \label{lem:ynetcost}
Given an $\R$-feasible matching $M, y(v)$,  suppose we have a simple alternating path or simple alternating cycle $\dir{P}$ from $u$ to $v$ ($u=v$ if $\dir{P}$ is a cycle) consisting only of admissible edges. Then,
\begin{equation}\label{eq:ynetcost}\lambda_u y(u) - \lambda_v y(v) \ge \wt(M\oplus \dir{P}) - \wt(M) + \sum_{(p,q) \in \dir{P}} \frac{\delta_{pq}}{2}. \end{equation} \end{lemma}
\begin{proof}
Plugging in \eqref{netcost3} into equation \eqref{slackcost}, we get 
\begin{eqnarray*}
\sum_{(p,q) \in \dir{P}} s(p,q)&=&\dist(M\oplus \dir{P}) - \dist(M) + \sum_{(p,q) \in \dir{P}} \delta_{pq}  +\lambda_v y(v) - \lambda_uy(u),\\
\lambda_u y(u)-\lambda_v y(v)&=& \dist(M\oplus \dir{P}) - \dist(M) + \sum_{(p,q) \in \dir{P}} \delta_{pq} - s(p,q).    
\end{eqnarray*}
Consider the case where $\delta_{pq} = 1$. Since $(p,q)$ is admissible, $s(p,q) = 0$ and we get $\delta_{pq} - s(p,q) = 1$. Otherwise, $\delta_{pq} \geq 2\sqrt{r}$, $(p,q)$ is a boundary edge, and $s(p,q) \leq \sqrt{r}$. Then we get $\delta_{pq} - s(p,q) \geq \delta_{pq} / 2$. Summing over both cases gives us
\begin{equation*}
\sum_{(p,q) \in \dir{P}} \delta_{pq} - s(p,q) \geq \sum_{(p,q) \in \dir{P}} \frac{\delta_{pq}}{2},
\end{equation*}
which gives us \eqref{eq:ynetcost}.
\end{proof}   

\begin{cor} \label{cpr:ynetcost}
Given an $\R$-feasible matching $M, y(v)$,  suppose we have a simple alternating path or simple alternating cycle $\dir{P}$ from $u$ to $v$ ($u=v$ if $\dir{P}$ is a cycle) consisting only of admissible edges and let $B(\dir{P})$ denote the edges  participating in $\dir{P}$ that are incident on boundary vertices. Then,
\begin{equation}\label{eq:ynetcost2}\lambda_u y(u) - \lambda_v y(v) \ge \wt(M\oplus \dir{P}) - \wt(M) + |B(P)|\sqrt{r}. \end{equation} \end{cor}
\begin{proof}
This follows easily from~\eqref{eq:ynetcost} and the fact that all values of $\delta_{pq}$ for boundary edges $(p,q)$ are at least $2\sqrt{r}$.
\end{proof}   

\section{Compressed Residual Graph}
\label{sec:compressed-graph}
In this section, we formally present the compressed residual graph and compressed feasibility. Note that there are a few key differences between our definition and the compressed residual graph as defined in~\cite{soda-18}. We highlight these differences in the discussion below. 

\paragraph{Active and inactive free vertices.}
During the \fastmatch\ procedure, the dual weight of a free vertex $b \in B$  may exceed a pre-determined upper bound  $\beta = \steptwoweight$. Such vertices are called \textit{inactive}. More specifically, a free vertex $b$ is inactive with respect to an $\R$-feasible matching $M, y(\cdot)$ if $y(b) \geq \beta$. All other free vertices of $B$ are \textit{active}. In our algorithm, each piece $\R_j$ will have a corresponding $\R$-feasible matching. Therefore, each piece has its own set of inactive vertices, $B_j^{\mathcal{I}}$, and its own set of active vertices, $B_j^{\mathcal{A}}$.
\paragraph{Compressed residual graph $H$.} 
We now define a compressed residual graph  $H$, which will be useful in the fast execution of the second step.  This definition is similar to the one in \cite{soda-18} with two differences. First, the vertex set of $H$ is modified to include inactive and active free vertices. Second, we will allow any vertex $v$ in $H$ to have an edge to itself, i.e., self-loops.

For a matching $M$, let $G_M$ denote the (directed) residual graph with respect to $M$.  Let $\R(G_{M})=\R(G)$ be the $r$-clustering of $G_M$ as given by Definition~\ref{def:r-division}. Let $A_F$ and $B_F$ denote the set of free vertices (vertices not matched by $M$) of $A$ and $B$ respectively. Let $B_F^\mathcal{A}$ (resp. $B_F^\mathcal{I}$) be the set of free vertices that are also active (resp. inactive). Our sparse graph $H$ will be a {\em weighted multi-graph} whose vertex set $V_H$ and the edge set  $E_H$ is defined next. 

We define the vertex set $V_j^H$ and the edge set $E_j^H$  for each  piece $\R_j$. The vertex set  $V_{H}$ and the edge set  $E_H$ are  simply the union of all the vertices and edges across all pieces.  For every piece $\R_j$, $V_j^H$ contains  the boundary vertices  $\boundary_j$. Also, if there is at least one internal vertex of $B$    that is also an active free vertex, i.e., $(V_j\setminus \boundary_j) \cap B_F^\mathcal{A} \neq \emptyset$, then we create a special vertex $b_j^{\mathcal{A}}$ to represent all vertices of this set in $V_j^H$. We also create a vertex $b_j^\mathcal{I}$ to represent all inactive free vertices of $(V_j \setminus \boundary_j)\cap B_F^\mathcal{I}$. Similarly, we create a vertex $a_j$  to represent all  the free vertices in $(V_j\setminus \boundary_j) \cap A_F$ if any exist.  We refer to the three additional vertices for $\R_j$ the \emph{free internal vertices} of $\R_j$ and refer to $b_j^\mathcal{A}$ (resp. $b_j^\mathcal{I}$) as the active (resp. inactive) free internal vertex. We set  $V_j^H = \boundary_j \cup \{a_j, b_j^{I}, b_{j}^{\mathcal{A}}\}$, $A_j^H = (\boundary_j \cap A) \cup \{a_j\}$, and $B_j^H = (\boundary_j \cap B)\cup \{b^\mathcal{A}_j, b_j^\mathcal{I}\}$. The free vertices of $B$  in piece $\R_j$ of the compressed graph $H$ are represented by $B_j^F= (B_F \cap \boundary_j) \cup \{b_j^\mathcal{A}, b_j^\mathcal{I}\}$ and the free vertices of $A$ in piece $\R_j$ of $H$ are represented by $A_j^F = (A_F \cap \boundary_j)\cup \{a_j\}$.  The vertex set $V_H$ of $H$ is thus given by $V_H = \bigcup_{j=1}^{l} V_{j}^H$. We also define the sets $B_H= \bigcup_{j=1}^l B_j^H$  and $A_H=\bigcup_{j=1}^l A_j^H$. The free vertices of  $B$ in $H$ are denoted by $B_H^F=\bigcup_{j=1}^l B_j^F$ and the free vertices of $A$ is denoted by  $A_H^F=\bigcup_{j=1}^l A_j^F$ . The vertices in $H$ represent sets of vertices in $G$. Given a vertex in $G$, we describe the corresponding vertex in $H$ as a \textit{representative vertex}.

\noindent Next we define the set of edges $E_j^{H}$ for each pieces $\R_j$.  Each edge $(u,v)$ in $E_j^H$ will represent a corresponding shortest net-cost path from $u$ to $v$ in $\R_j$. We denote this path as $\dir{P}_{u,v,j}$ and describe this mapping from an edge of $E_H$ to its corresponding path in $G$ as \textit{projection}. For any $u, v \in V_j^H$, where $u$ and $v$ are allowed to be the same vertex, there is an edge from $u$ to $v$ in each of the following four cases.

\begin{itemize}
\item[1)] $u, v\in \boundary_j$, i.e., $u$ and $v$ are boundary vertices and there is a directed path $\dir{P}$ from $u$ to $v$ in $G_M$ that only passes through the  edges of the piece $\R_j$. Let $\dir{P}_{u,v,j}$  be the path consisting only of edges of $\R_j$ that has the smallest net-cost. We denote this type of edge as a \emph{boundary-to-boundary} edge. When $u=v$, $\dir{P}_{u,v,j}$ is the smallest net-cost cycle  inside $\R_j$ that contains the vertex $u$. Since $\R_j$ is a bipartite graph, such a cycle must consist of at least $4$ edges.\ 

\item[2)] $u = \{b_j^\mathcal{A},b_j^\mathcal{I}\}$, $v \in \boundary_j$, and there is a directed path $\dir{P}$ in $G_M$ from some free vertex in $B_F^{\mathcal{A}}\cap (V_j\setminus \boundary_j)$ (if $u=b_j^\mathcal{A}$) or $B_F^\mathcal{I} \cap (V_j\setminus \boundary_j)$ (if $u = b_j^\mathcal{I}$) to $v$ that only passes through the  edges of $\R_j$. Let $\dir{P}_{u,v,j}$  be the path consisting only of  edges of $\R_j$ that has the smallest net-cost. 

\item[3)] $u \in \boundary_j $, $v = a_j$, and there is a directed path $\dir{P}$ in $G_M$ from $u$ to some free vertex in $A_F\cap (V_j\setminus \boundary_j)$  that only passes through the  edges of $\R_j.$ Let $\dir{P}_{u,v,j}$  be the path consisting only of  edges of $\R_j$ that has the smallest net-cost.
 
\item[4)] $u = \{b_j^{\mathcal{A}},b_j^\mathcal{I}\}$ and $v = a_j$ are free vertices and there is a directed path $\dir{P}$ in $G_M$ from some vertex in the set $B_F^{\mathcal{A}}\cap (V_j\setminus \boundary_j)$ (if $u = b_j^\mathcal{A}$) or $B_F^\mathcal{I} \cap (V_j\setminus \boundary_j)$ (if $u=b_j^\mathcal{I}$) to a vertex in the set $A_F\cap (V_j\setminus \boundary_j)$ that only passes through the  edges\ in $\R_j$. Let $\dir{P}_{u,v,j}$  be the path consisting only of  edges of $\R_j$ that has the smallest net-cost. 
\end{itemize}
\begin{figure}
  \centering
  \includegraphics[width=17cm,scale=1]{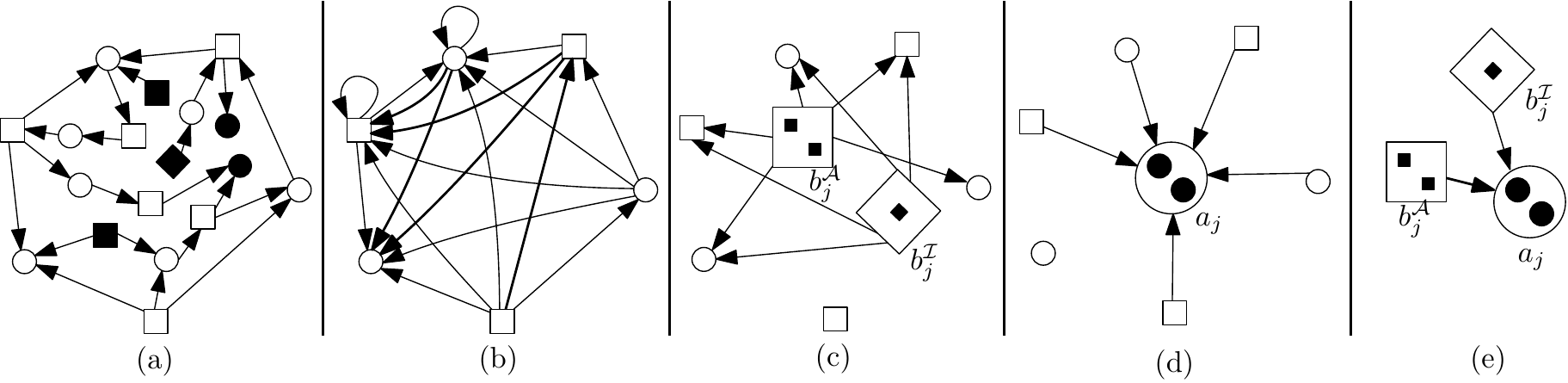}
  \caption{(a) A piece $\R_j$. The squares represent vertices of $B$ and the circles represent vertices of $A$. Filled vertices are free internal vertices. Specifically, the filled diamond represents an inactive free vertex, and the filled squares represent active free vertices. (b) The boundary-to-boundary edges and self-loop edges of $H$ for $\R_j$. (c) The edges from $b_j^{\mathcal{A}}$ and $b_j^{\mathcal{I}}$ to $\boundary_j$. (d) The edges from $\boundary_j$ to $a_j$. (e) There is a single edge from both $b_j^{\mathcal{A}}$ and $b_j^{\mathcal{I}}$ to $a_j$.}
\label{fig:H}
\end{figure}

See Figure~\ref{fig:H} for an example piece of $H$ from a piece of $G$. We set the weight of each edge to be $\phi(\dir{P}_{u,v,j})$.  We also refer to this edge $(u,v)$ as an edge of piece $\R_j$ in $H$ and denote the set of all edges of piece $\R_j$ as $E_j^H$.  The set of edges of $H$ is simply $E_H = \bigcup_j E_j^H$. Note that $H$ is a multi-graph as there can be directed path from $u$ and $v$ in multiple pieces. 

Note that the number of vertices in $H$ is $O(n/\sqrt{r}$). For each vertex $v \in \boundary$, let $\theta_v$ be the number of pieces in which it belongs. Counting multiplicity, the number of boundary vertices is $O(n/\sqrt{r})$ and therefore,$$\sum_{v\in \boundary} \theta_v = O(n/\sqrt{r}).$$
Any boundary vertex $v$ in $H$ can have edges to at most $\sqrt{r}$ boundary vertices inside any piece it participates in. Therefore, the total number of edges can be bounded by $$\sum_{v \in \boundary} \theta_v\sqrt{r} = O(n),$$   
leading to the following lemma.
\begin{lemma}
\label{lem:hsize}
The compressed residual graph $H$ has $O(n/\sqrt{r})$ vertices and $O(n)$ edges.
\end{lemma}
 This completes the description of the compressed residual graph $H$.
The compressed residual graph defined here differs from the one in~\cite{soda-18} in two ways: We classify free vertices as active and inactive and we allow for self-loops in the compressed residual graph $H$. We describe the reasons for introducing these changes. 

\ignore{
\paragraph{Need for active and inactive vertices.}
It is necessary to differentiate between active and inactive vertices because, unlike in our variant of Gabow-Tarjan, we cannot guarantee that every vertex in $H$ is visited only once in each phase of the algorithm. In fact, we can construct an example where a vertex can be visited an arbitrary number of times in the same  phase. In Figure~\ref{fig:create-new}, a DFS style search on the admissible edges of $H$ (those with slack $\le \sqrt{r}$) visits $u$ and then $z$. There are no admissible edges from $z$ and so, the search backtracks and goes to $v$ which leads to an augmenting path. After augmenting along the path, the edge $(y,z)$ is created in $H$. Now, in the same iteration, another DFS from a different free vertex may visit $y$ and then $z$ again (note that $(y,z)$ has a slack less than $\sqrt{r}$).  

Similar to the Gabow-Tarjan's variant, we can guarantee that every time a vertex in $H$ is visited, it is either part of an augmenting path (or alternating cycle/path), or its dual weight increases by $\sqrt{r}$.  A vertex that is visited many times during the execution of the algorithm, therefore, will accumulate a dual weight with a very large magnitude. When $|\ty{v}|$ for some  vertex $v$ in $B_H$  reaches the upper bound $\beta$ in the second step,  we flip edges along an alternating path to make $v$ an inactive free vertex, without decreasing the size of the matching. The vertex $v$ will then remain unused for the remainder of the second step.  
}
 \paragraph{Need for self loops.}  
Unlike Gabow and Tarjan's algorithm, we do not compute the shortest augmenting path. So, in order to guarantee that the paths we find in $H$ are simple, we need to actively look for possible cycles. Such a cycle may lie entirely within a piece and involve only a single boundary vertex. Including self-loops in $H$ is helpful in detecting such cycles.

\paragraph{Compressed feasibility.} Next, we will define the requirements for a \emph{compressed feasible matching}. We denote as $M_j$ the edges of $M$ that belong to piece $\R_j$.  $M = \bigcup_{j=1}^l M_j$. For any vertex $v \in V_H$, let $\lambda_v$ be $-1$ if $v \in A_H$ and $1$ if $v \in B_H$. For each piece $\R_j$, we maintain a dual weight $y_j(v)$ for every vertex in $v\in V_j$. These dual weights $y_j(\cdot)$ along with $M_j$ form an $\R$-feasible matching.  Additionally, we store a dual weight $\ty{v}$  for every $v \in V_{H}$.    
We say that the dual weights $\ty{\cdot}$ are \emph{$H$-feasible} if they satisfy the following conditions. For each piece $\R_j$,  and for every directed edge $(u,v) \in E_j^H$,   
\begin{eqnarray}
\lambda_u\ty{u} - \lambda_v\ty{v} &\le & \phi(\dir{P}_{u,v,j}).  \qquad \label{eq:uvpath1}
\end{eqnarray} 

 For any graph $G$, and an $r$-clustering $\R(G)$, we say that a matching $M$, a set of dual weights $y_j(\cdot)$ for the vertices of each piece $\R_j$, and a set of dual weights $\ty{\cdot}$ for the vertices $V_{H}$, form a compressed feasible matching if the following conditions (a)--(e) are satisfied.

\begin{enumerate}[(a)]
\item \label{cf-a} For every vertex $v \in A_H$, $\ty{v} \leq 0$ and for every free vertex $v \in A_H^F$, $\ty{v} = 0$.
\item \label{cf-b} Let $y_{max}$ be $\max_{v \in  \activeH} \ty{v}$. For every vertex $v \in B_{H}$, $\ty{v} \geq 0$ and for all active vertices $v \in \activeH$, $y_{max} - \sqrt{r} \le \ty{v}\le y_{max}$. For all inactive vertices $v\in \inactiveH$, $\ty{v} \geq \beta$.

\item \label{cf-c} For every piece $\R_j$, the matching $M_j $  and dual weights $y_j(\cdot)$ form an $\R$-feasible matching. 
\item \label{cf-d} The dual weights $\ty{\cdot}$ are $H$-feasible.

\item \label{cf-e} For each piece $\R_j$ and any $v \in \boundary_{j}$ ,  $|\ty{v}| \ge |y_j(v)|$. For every vertex $a \in (V_{j}\setminus \boundary_j)\cap A_F$, $|y_j(a)|=0 $. \end{enumerate}

Using (a) and (b), we can restate the $H$-feasibility conditions compactly as  
\begin{equation}
|\ty{u}| - |\ty{v}| \le \phi(\dir{P}_{u,v,j}). \label{eq:uvpathsum}
\end{equation}
and we can define the slack of an edge $(u,v) \in E_H$ to be $s_H(u,v) = \phi(\dir{P}_{u,v,j}) -|\ty{u}|+|\ty{v}|$. Let $u_0$ be the first vertex and $v_\ell$ be the last vertex of $\dir{P}_{u,v,j}$. Note that, if $\ty{u} = y_j(u_0)$ and $\ty{v} = y_j(v_\ell)$ then, using \eqref{slackcost}, 
\begin{equation}  \label{eq:slackHsumslacksG}
s_H(u,v) = \sum_{(u',v') \in \project{u}{v}{j}} s(u',v').
\end{equation}
We say that  an edge $(u,v) \in E_H$ is \emph{admissible} if the slack $s_H(u,v) \leq \sqrt{r}$. 

 Note that a boundary vertex has many different dual weights assigned to it, one for each of the pieces it belongs to. During the course of our algorithm,  the magnitudes of the dual weights of vertices in $H$ only increase (with a few exceptions).  As we do not immediately update the dual weights of all vertices in $G$,   for  some piece $\R_j$  the dual weight $y_j(\cdot)$ may not reflect the updated dual weight. This condition is captured by (e).  
\paragraph{Conventions for notations in a compressed residual graph.}For every boundary or free internal vertex, we define the \emph{representative} of $v$ to be $v$ itself if $v$ is a boundary node. Otherwise, if $v$ is a free internal vertex, then it is  one of the three free internal vertices $a_j$, $b_j^{\mathcal{I}}$  or $b_j^{\mathcal{A}}$ depending on whether $v$ is a free vertex of $A$, free inactive vertex or a free active vertex respectively. We denote  the representative of $v$ by $rep(v)$. For simplicity in exposition, wherever convenient, we will abuse notations and use $v$ to also denote this  representative $rep(v)$ in $H$. Following our convention, we set $H'$  to be a graph identical to $H$ with the weight of every edge replaced by its slack.

The compressed residual graph allows our algorithm to search for paths by modifying $\ty{\cdot}$ values in $H$ without explicitly modifying the $y_j(\cdot)$ for every piece $\R_j$.  In doing so, there may be a free vertex $b \in B_F\cap V_j$ such that $\ty{rep(b)}$ may exceed $\beta$ whereas $y_j(b)$ remains below $\beta$. In such a situation, our convention is to assume $b$ to be inactive with respect to $\R_j$. Thus, when a vertex of $H$ becomes inactive, all the vertices it represents also become inactive. 
This convention fits with the notion that the values $\ty{\cdot}$ can be seen as up to date, while the values $y_j(\cdot)$ are lazily updated. 

To overcome mild technical challenges encountered in the presentation of the algorithm, we introduce two useful procedures next.   

\paragraph{Procedures that reduce dual weight magnitudes.} We  introduce two  procedures, called \reduce$(b_j^{\mathcal{A}} /b_j^{\mathcal{I}} , \alpha)$ and \reduceslack$(v)$, that allow us to reduce the dual weights of  vertices of $B$  without violating compressed feasibility. During the second step of the algorithm, the dual weights may decrease in magnitude only within these two procedures. Otherwise, the magnitude of the dual weights only increase.  We describe these procedures next.

\reduce\ takes as input a free active (resp. inactive) internal vertex $b_j^{\mathcal{A}}$ (resp. $b_j^{\mathcal{I}}$), and a value $\alpha$ such that $0 \leq y_{max} - \sqrt{r} \leq \alpha \le \ty{b_j^{\mathcal{A}}}$ (resp. $\beta \le \alpha \le\ty{b_j^{\mathcal{I}}}$). For all $v \in (V_j \setminus \boundary_j) \cap B_F^{\mathcal{A}}$ (resp. $B_F^{\mathcal{I}}$), if $y_j(v) \geq \alpha$, it sets the dual weight $y_j(v) \leftarrow \alpha$ . Then it sets $\ty{b_j^{\mathcal{A}}}$ (resp. $\ty{b_j^{\mathcal{I}}}$) to $\alpha$. 

\reduceslack\ takes as input a matched vertex $v \in B$. Let $u\in A$ be the vertex that $v$ is matched to, and let $(u,v)$ belong to the piece $\R_j$. The procedure sets $y_j(v) \leftarrow  y_j(v) - s(u,v)$. If $v$ is also a boundary vertex, i.e., $v \in \boundary_j$, then it sets $\ty{v} \leftarrow y_j(v)$ and for every other piece $\R_{j'}$ such that $v\in V_{j'}$, it sets $y_{j'}(v) \leftarrow y_j(v)$.

For a discussion on why the \reduce\ and \reduceslack\ procedures do not violate the compressed feasibility conditions, see Section \ref{A:reduce} of the appendix. From that discussion, we get the following Lemma.

\begin{lemma}
\label{lem:reduce}
Invoking \reduce\ or \reduceslack\ procedures on a compressed feasible matching will  not violate any of the compressed feasibility conditions (a)--(e).
\end{lemma}

\section{Our Scaling Algorithm}
\label{sec:scaling}

As in the Gabow-Tarjan Algorithm, for every edge $(a,b) \in E$,  we redefine its weight to be $\dist^*(a,b) = (kn +1)\dist(a,b)$, where $k$ is a constant defined in the upcoming Lemma \ref{lem:optrel2}. Since this uniform scaling of edge costs preserves the set of optimal matchings,
Lemma~\ref{optrel2} implies that a $\R$-optimal matching of the vertices of
$A,B$ with edge weights $\dist^*(\cdot,\cdot)$ corresponds to an optimal matching with the original edge costs $\dist(\cdot,\cdot)$.
For any edge $(u,v)$, let $b_1, b_2\ldots b_\ell$ be the binary representation of $\dist^*(u,v)$. Let $\dist_{i}(u,v)$, correspond to the most significant $i$ bits of $\dist^*(u,v)$.
The following lemma bounds the cost of any $\R$-optimal matching on $G$.
\begin{lemma}
\label{optrel2}
\label{lem:optrel2}
For a bipartite graph $\Gr(A\cup B, E)$ with a positive integer edge cost function $\dist$, let $M$ be an $\R$-optimal matching and $M_{\opt}$ be some optimal matching. Then,
$\dist(M) \le \dist(M_{\opt}) + kn$
where $k = (2k_1 + 4k_2 + 1)$ is a constant.
\end{lemma}
\begin{proof}
Edges that are in both $M$ and $M_\opt$ need not be considered. Since $M$, $y(\cdot)$ is $\R$-optimal, all edges of $M \setminus M_\opt$ satisfy~\eqref{eq:feas6}, and we have
\begin{equation}
\dist(M \setminus M_\opt) = \sum_{(u,v) \in M \setminus M_\opt} \dist(u,v) \leq \sum_{(u,v) \in M \setminus M_\opt} y(u) + y(v) + \delta_{uv}. \label{eq:costm}
\end{equation}
Every edge in $M_{\opt} \setminus M$ satisfies~\eqref{eq:feas5} so we have
\begin{equation}
\dist(M_{\opt} \setminus M) = \sum_{(u,v) \in M_{\opt} \setminus M} \dist(u,v) \geq \sum_{(u,v) \in M_{\opt} \setminus M} y(u) + y(v) - \delta_{uv}. \label{eq:costm2}
\end{equation}
By subtracting \eqref{eq:costm2} from \eqref{eq:costm}, we have
\begin{equation}
\label{eq:24}
\dist(M) - \dist(M_{\opt}) \leq \sum_{(u,v) \in M \setminus M_\opt} y(u) + y(v) - \sum_{(u,v) \in M_\opt \setminus M}y(u) + y(v) + \sum_{(u,v) \in M \oplus M_\opt} \delta_{uv}.
\end{equation}
Since $M_\opt$ and $M$ are both perfect matchings, $$\sum_{(u,v) \in M \setminus M_\opt} y(u) + y(v) - \sum_{(u,v) \in M_\opt \setminus M}y(u) + y(v) = 0.$$Therefore, it is sufficient to bound $\sum_{(u,v) \in M \oplus M_\opt} \delta_{uv}$.

$\delta_{uv}$ can take one one of three values for every edge: $1, 2\sqrt{r}$, or $\frac{m_j n}{m\sqrt{r}}$. There are at most $k_1 \sqrt{r}$ boundary vertices per piece $\R_j$, and at most  one edge in both $M$ and $M_\opt$ adjacent to each vertex, so there are at most $2k_1\sqrt{r}$ edges $(u,v)$ in $M \oplus M_\opt$ such that $\delta_{uv} = \frac{m_j n}{m\sqrt{r}}$. There are at most $k_2 n/\sqrt{r}$ boundary vertices in the $r$-clustering, so there are at most $2k_2n / \sqrt{r}$ edges in $M \oplus M_\opt$ for which $\delta_{uv} = 2\sqrt{r}$. For the other at most $2n$ edges $(u,v) $ of $M \oplus M_\opt$ , $\delta_{uv} = 1$. Therefore,

\begin{eqnarray}
\sum_{(u,v) \in M \oplus M_\opt} \delta_{uv} \nonumber 
&\leq& \ (\sum_{j} 2k_1\sqrt{r}\frac{m_j n}{m \sqrt{r}}) + 2\frac{k_2n}{\sqrt{r}} (2\sqrt{r})+ 2n, \nonumber \\
&\leq& \ 2nk_1(\sum_{j} \frac{m_j}{m}) + (4k_2+2)n, \nonumber \\
&\leq& (2k_1 + 4k_2 + 2) n. \label{eq:errorsum}
\end{eqnarray}

Let $k = (2k_1 + 4k_2 + 2)$. Then we have.
\begin{equation*}
\dist(M) \leq \dist(M_{\opt}) + kn.
\end{equation*}
\end{proof}
Using an almost identical argument, we can also bound the cost of any $\R$-feasible matching $M$ (not necessarily perfect) and dual weights $y(\cdot)$ by $\dist(M_{\opt}) +O(n)$ so long as every free vertex of $A$ has a dual weight of $0$ and every free vertex of $B$ has a positive dual weight.
\begin{lemma}
\label{lem:feasrel2}
For a bipartite graph $\Gr(A\cup B, E)$ with a positive integer edge cost function $\dist$, let $M$ along with $y(\cdot)$ be an $\R$-feasible matching such that every free vertex $b$, $y(v) \ge 0$ and for every free vertex $a \in A$, $y(a) =0$,  and let $M_{\opt}$ be any optimal matching. Then,
$\dist(M) \le \dist(M_{\opt}) + kn$
where $k = (2k_1 + 4k_2 + 1)$ is a constant.
\end{lemma}
\begin{proof}
Using \eqref{eq:24}, it is sufficient to show that
$$\sum_{(u,v) \in M \setminus M_\opt} y(u) + y(v) - \sum_{(u,v) \in M_\opt \setminus M}y(u) + y(v) \leq 0.$$ Since $M_\opt$ is perfect, we can rewrite the left side as $\sum_{u \in F}y(u)$ where $F$ is the set of free vertices with respect to $M$. The fact that all free vertices have nonnegative dual weight gives the lemma.
\end{proof}

Our algorithm consists of \emph{scales}. The input to any scale $i$ is  a bipartite graph on $A, B,$ with a cost function $\dist_{i}(\cdot,\cdot)$, and a set of dual weights $y(v)$ for every node $v\in A \cup B$. Let the $\dist^i(u,v) = \dist_i(u,v) - y(u)-y(v)$ be the \textit{reduced cost} of $(u,v)$. Reduced costs satisfy the following properties: 
\begin{itemize}
\item[(E1)]For every edge $(u,v)$, $\dist^i(u,v) \ge \delta_{uv}$, and,
\item[(E2)] The cost of a $1$-optimal matching with respect to $\dist^i(\cdot,\cdot)$ is $O(n)$.

\end{itemize}

Given an input with these properties, the algorithm within a scale returns an $\R$-optimal matching $M$ and dual weights $y'(\cdot)$ with respect to the reduced costs $\dist^{i}(\cdot,\cdot)$. Reduced costs do not affect the optimal matching. Furthermore, for every vertex $v \in A\cup B$, let $y(v)$ be the sum of the dual weight $y'(v)$ and the  dual weight of $v$ that was provided as input to scale  $i$.   It can be shown that $M, y(\cdot)$ is also $\R$-optimal with respect to $\dist_i(\cdot,\cdot)$.

The input to the  first scale is the graph $G(A\cup B, E)$ with the cost $\dist_1(\cdot,\cdot)$ and a set of   dual weights of $-1$ on every internal vertex $u$ of $B$, dual weight of $-\max_{v \in V} \delta_{uv}$ for every boundary vertex $u$ of $B$ and  dual weights of $0$ on every vertex of $A$.  It is easy to see that $\dist^1(\cdot,\cdot)$ satisfies (E1) and (E2). For any scale $i$, using $\dist^i(\cdot,\cdot)$  as the cost, the algorithm for a scale (described below) computes a matching $M$ and dual weights $y'(\cdot)$ so that $M, y'(\cdot)$ is an $\R$-optimal matching.  

For any scale $i \ge 1$, we use the  $\R$-optimal matching $M, y(\cdot)$  returned by the algorithm for scale $i$ to generate an input for scale $i+1$ as follows. It is possible that for any matching edge $(u,v)$ of the $\R$-optimal matching $M$,  the sum of the dual weights $y(u)+y(v)$ greatly exceeds $\dist(u,v)$ (see~\eqref{eq:feas6}). Prior to moving the dual weights to scale $i+1$, the algorithm sets for all edges $(a,b) \in M$, where $a \in A$ and $b \in B$,
\begin{equation*}
y(b) \leftarrow y(b) - s(a,b).
\end{equation*}
It is easy to see that this results in an $\R$-optimal matching $M$ and new dual weights $y(\cdot)$ such that the slack of every matching edge is 0. Note that reducing $y(b)$ only increases the slack on adjacent nonmatching edges. The algorithm transfers the dual weights from scale $i$ to scale $i+1$ as follows.
For any vertex $v \in A \cup B$,

\begin{equation*}
y(v) \leftarrow 2y(v)- 2\max_{u \in N(v)} \delta_{uv}.
\end{equation*}

 Therefore, at the beginning of scale $i+1$, the  reduced cost $\dist^{i+1}(u,v)=\dist_{i+1}(u,v) - y(u)- y(v)$ on every edge is at least $\delta_{uv}$ implying (E1). The cost of the optimal matching for scale $i+1$, $M_{i}$, with respect to the new costs $\dist^{i+1}(\cdot,\cdot)$ is at most  $\sum_{(u,v) \in M_i} 2\delta_{uv} + \sum_{v \in A \cup B} 2\max_{u \in N(v)} \delta_{uv} + n$. Following similar steps as those used in showing~\eqref{eq:errorsum} gives that the cost of the optimal matching or $\R$-optimal matching with respect to $\dist^{i+1}(\cdot,\cdot)$ is $O(n)$, implying (E2).

 In the next section, we describe an algorithm for each scale. This algorithm takes a graph with positive integer edge costs where the cost of each edge $\dist(u,v) \ge \delta_{uv}$ (condition (E1))and  the optimal matching has a cost of $O(n)$ (condition (E2)). Given this input, it computes an $\R$-optimal matching in $\tilde{O}(nm^{2/5})$ time. After $O(\log{((kn+1)C))}$ scales, the $\R$-optimal matching returned by our algorithm will also be an optimal matching.

\section{Algorithm For Each Scale}
\label{algorithmforscale}

Our algorithm takes a bipartite graph $G(A\cup B, E)$ and its $r$-clustering as input. Each edge $(u,v)$ of this graph has a positive integer cost of $\dist(u,v)$ with  $\dist(u,v) \ge \delta_{uv}$ and the optimal matching has a cost no more than $O(n)$. Given such an input, it produces an $\R$-optimal matching.  

The algorithm has three steps. The \emph{first step} (also called the preprocessing step) of the algorithm will execute $\sqrt{r}$ iterations of a scale of the GT-Algorithm~\cite{gt_sjc89}. This can be executed in $O(m\sqrt{r})$ time. At the end of this step, the algorithm has a $1$-feasible matching $M$ and dual weights $y(\cdot)$ that satisfy the original dual feasibility conditions of Gabow and Tarjan (\eqref{eq:gtfeas5} and~\eqref{eq:gtfeas6}) and there are at most $O(n/\sqrt{r})$ free vertices. Additionally, from the properties of GT-Algorithm, for every free vertex $a \in A$, $y(a) =0$, and for every free vertex $b \in B$, $y(b) \geq \sqrt{r}$. Furthermore, the dual adjustments  performed during an iteration of GT-Algorithm only decrease dual weights of $A$ and increase dual weights of $B$. A 1-feasible matching also satisfies the requirements for an $\R$-feasible matching (\eqref{eq:feas5} and~\eqref{eq:feas6}). Therefore, at the end of the first step, we have the following.

\begin{lemma}
\label{lem:preprocess}
At the end of the first step of our algorithm, the matching $M$ and the dual weights $y(\cdot)$ form an  $\R$-feasible matching for the graph $G$, and the number of free vertices with respect to $M$ is at most $O(n/\sqrt{r})$. For every vertex $a \in A$, $y(a) \leq 0$ and for every vertex $b \in B$, $y(b) \geq 0$. For every free vertex $a \in A$, $y(a) =0$, and for every free vertex $b \in B$, $y(b)=\max_{b' \in B} y(b') \geq \sqrt{r}$.
\end{lemma}

To match the remaining $O(n/{\sqrt{r}})$ unmatched vertices, the algorithm will use the $r$-clustering to construct a compressed residual graph with $O(n/\sqrt{r})$ vertices and $O(n)$ edges (See Section~\ref{sec:compressed-graph} for relevant definitions). The \emph{second step} consists of iteratively calling the \fastmatch\ procedure (lines \ref{step2startline}--\ref{step2endline} of Algorithm \ref{alg:hmfree}) to match all but $O(\sqrt{n}/r^{1/4})$ vertices. The \emph{third step} then iteratively matches the remaining vertices one at a time (lines \ref{step3startline}--\ref{step3endline} of Algorithm \ref{alg:hmfree}).

At the end of the first step of our algorithm, we have a matching $M$ and a set of dual weights that form an $\R$-feasible matching. In Section~\ref{subsec:rtoplanar}, we describe a procedure called \construct\ that, given an $\R$-feasible matching on a piece in the $r$-clustering, computes and stores the edges of the compressed feasible matching for that piece.  In Section~\ref{subsec:planartor}, we describe another procedure called \sync\ that, given the edges of a compressed feasible matching of any piece, computes an $\R$-feasible matching. We can use these procedures to convert any $\R$-feasible matching into a compressed feasible matching and vice-versa. We also use these procedures in the second step of our algorithm in Section~\ref{subsec:step2}.\   

\subsection{Computing a compressed feasible matching from an $\R$-feasible matching}

\label{subsec:rtoplanar}
\ignore{At the end of the first step, we have an $\R$-feasible matching $M, y(\cdot)$ that also satisfies $y(u) = 0$ for all $u \in A_F$ and $y(v)=\max_{v' \in B} y(v')$ for all $v \in B_F$.} In this section, we will present an algorithm to compute a compressed residual graph and a compressed feasible matching from this $\R$-feasible matching $M$ and its set of dual weights $y(\cdot)$.  For every vertex $v \in A\cup B$, and for every piece $\R_j$ such that $v \in V_j$, we set $y_j(v) = y(v)$.\ For every boundary vertex $v \in \boundary$, we set $\ty{v}=y(v)$. We also set, for every piece $\R_j$, $\ty{a_j}=0$ and $\ty{b_j^\mathcal{A}}=\gamma$ where $\gamma \geq \sqrt{r}$ is the dual weight of all free vertices of $B_F$ from the first step. Note that $\beta > \gamma$ and therefore, there is no inactive free internal vertex at the end of the first step. So, we do not create a free inactive internal vertex for any piece. Note that, from Lemma~\ref{lem:preprocess}, conditions (a), (b) and (e) are trivially satisfied. The matching $M_j$ and $y_j(\cdot)$ form a 1-feasible matching. Edges satisfying $1$-feasibility conditions  also satisfy the $\R$-feasibility condition and so (c) is satisfied. The  next lemma shows that dual weights $\ty{\cdot}$ satisfy $H$-feasibility and therefore (d) holds.

\begin{lemma}
\label{lem:const1}
 Consider a matching $M_j$ and a set of dual weights $y(\cdot)$ for a piece $\R_j$ such that $M_j, y(\cdot)$ is $\R$-feasible. Suppose the dual weights of all vertices of $A$ in $V_j$ are non-positive and the dual weights of all vertices of $B$ in $V_j$ are non-negative. For any two vertices $u, v \in V_j$, let the directed path $\dir{P}_{u,v,j}$ be a minimum net-cost path from $u$ to $v$ in $\R_j$. Then,
\begin{eqnarray}
|y(u)| -| y(v)| &\le & \phi(\dir{P}_{u,v,j}). \label{eq:slack1-1}
\end{eqnarray}
Furthermore, 

\begin{equation} \label{eq:slack2-1}
    \sum_{(a,b) \in \dir{P}_{u,v,j}} s(a,b) = \phi(\dir{P}_{u,v,j}) - |y(u)| + |y(v)|.
\end{equation}

\end{lemma}
\begin{proof}
From equation~\eqref{slackcost} we have,
\begin{equation*}
\sum_{(a,b) \in \dir{P}} s(a,b) = \phi(\dir{P})  - \lambda_uy(u) +\lambda_v y(v).
\end{equation*}
From the fact that all vertices of $B$ have nonnegative dual weight and all vertices of $A$ have nonpositive dual weight, we get that $\lambda_u = |y(u)|$ and $\lambda_v = |y(v)|$. This gives~\eqref{eq:slack2-1}. From the fact that all slacks in an $\R$-feasible matching are nonnegative,~\eqref{eq:slack1-1} follows.

\end{proof}

\ignore{ 
  We can prove the following lemma using a proof similar to that of Lemma~\ref{lem:const1}. \begin{lemma}
\label{lem:const2}
Consider a matching $M$ and a set of dual weights $y(\cdot)$ on the vertices of the input graph $G(A\cup B,E)$ such that $M, y(\cdot)$ is $\R$-feasible.  Suppose all vertices of $A$ have a non-positive dual weight and all vertices of $B$ have a non-negative dual weight. For any two vertices $u, v \in A\cup B$, let the directed path $\dir{P}_{u,v}$ be a minimum net-cost path from $u$ to $v$ in $G_{M}$. Then,
\begin{eqnarray*}
|y(u)| -| y(v)| &\le & \phi(\dir{P}_{u,v}). 
\end{eqnarray*}
Furthermore, 

\begin{equation*} 
    \sum_{(a,b) \in \dir{P}_{u,v}} s(a,b) = \phi(\dir{P}_{u,v}) - |y(u)| + |y(v)|. 
\end{equation*}

\end{lemma}
} 

Using the following lemma, we will provide a procedure called \construct\ to compute all the edges of $H$.
   
\begin{lemma}
\label{lem:poswt}
Let $\R_{j}'$ be a directed graph identical to the  directed graph $\R_{j}$ except that the cost of any edge $(a,b)$ is set to be its slack $s(a,b)$. Then,
 for any two vertices $u, v$ in $\R_j$, the minimum net-cost directed path from $u$ to $v$ in $\R_j$ is the minimum cost directed path between $u$ and $v$ in $\R_j'$. The dual weights $y(u)$, $y(v)$ and the length of the shortest path in $\R_j'$ immediately give us the value of the minimum net-cost between $u$ and $v$ in $\R_j$. 
\end{lemma}
\begin{proof}
Let $P$ be any alternating path between vertices $u$ and $v$. We can write the net-cost of $P$ as in equation \eqref{eq:slack2-1},
\begin{eqnarray*}
\min_P \phi(P) 
&=& \min_P \sum_{(a,b)\in P\setminus M} s_{}(a,b) + |y_{}(u)| - |y_{}(v)|. 
\end{eqnarray*}
For every path, $y(u)$ and $y(v)$ are the same. Therefore, we conclude that computing minimum net-cost path is equivalent to finding the minimum-cost path $P^*$ between $u$ and $v$ in $\R_j'$. Furthermore, the sum of the cost of $P^*$ with $y(u)$ and $y(v)$ will give the value of the minimum net-cost between $u$ and $v$. 
\end{proof} 
We define the slack on any directed edge $(u,v) \in E^{H}_j$ to be $$s_H(u,v) = \phi(\dir{P}_{u,v,j}) - |\ty{u}| + |\ty{v}|.$$ 
From the Lemma \ref{lem:const1} above, it follows that slack of the edge $(u,v)$ is  non-negative, and exactly equal to $\sum_{(a,b) \in \dir{P}_{u,v,j}} s(a,b)$, provided $\ty{u} = y_j(u)$ and $\ty{v} = y_j(v)$. Following our convention, we use $H'$ to denote the compressed residual graph with the same edge set as $H$ but with the edge weights being replaced with their slacks. Our initial choice of $\ty{\cdot}$ is $H$-feasible. To assist in the  execution of the second step of our algorithm, we explicitly compute the edges of $H$ and sort them in increasing order of their slacks.

\paragraph{The \construct\ procedure.} Using Lemma \ref{lem:poswt}, we describe a procedure that computes the edges $E_j^H$ for a piece $\R_j$. This process will be referenced as the \construct\ procedure. Note that in some graphs, such as planar graphs, \construct\ could use a faster algorithm; see Section \ref{sec:planar}. This procedure takes a piece $\R_j$ of $G_{M}$ as input and constructs the edges of $E_j^H$. We next give a summary of how to accomplish this, for further details, see appendix section \ref{A:construct}. Let $\R_{j}^{'}$ be the graph of $\R_j$, with all edge weights converted to their slacks according to the current matching $M_j$ and the current dual assignment $y_j(\cdot)$. From equation \eqref{eq:slack2-1} and Lemma \ref{lem:poswt}, it is sufficient to compute the shortest path distances in terms of slacks between all pairs of vertices in $V_j^H$. These distances in slacks can then each be converted to net-costs in constant time. Therefore, using $O(\sqrt{r})$ separate Dijkstra searches over $\R_j'$, with each search taking $O(m_j + n_j \log n_j)$ time, the edges of $E_j^H$ can be computed. From the discussion in appendix section \ref{A:construct}, we get the following Lemma and Corollary.

\begin{lemma}
\label{lem:construct}
Given an $\R$ feasible matching $M_j$, $y_j(\cdot)$, the \construct\  procedure builds the edges of $E_j^H$  in $\bruteforceconstruct$ time. 
 \end{lemma} 

\begin{cor}
\label{cor:preprocess}
Let $\R$-feasible matching $M$, $y(\cdot)$ be the matching computed at the end of the first step of our algorithm. Given $M$, $y(\cdot)$, we can use the \construct\ procedure to compute the graph $H$ in $O(\sqrt{r}(m + n\log{n}))$ time.
\end{cor}

\subsection{Computing an $\R$-feasible matching from a compressed feasible matching}
\label{subsec:planartor}

In a compressed feasible matching, a boundary vertex has multiple dual weights, one corresponding to each of the pieces it belongs to. It also has a dual weight $\ty{\cdot}$ with respect to the graph $H$. We introduce a synchronization procedure (called \sync)\ that will take a compressed feasible matching along with a piece $\R_j$  and update the dual weights $y_j(\cdot)$ so that the new dual weights and the matching $M_j$ continue to be $\R$-feasible and for every boundary vertex $v \in \boundary_j$, $y_j(v) = \ty{v}$. We can convert a compressed feasible matching into an $\R$-feasible matching by repeatedly invoking this procedure for every piece. \
 
  The \sync\ procedure is implemented as follows: 
\begin{itemize} 
\item Recollect that the graph $\R_j'$ is a graph identical to $\R_j$ with slacks as the edge costs. Temporarily add a new source vertex $s$ to $\R_j'$ and add an edge from $s$ to every $v \in \boundary_j$. Also add an edge from $s$ to every unmatched internal vertex $v$ of $A$ and $B$ in $\R_j'$, i.e., $v \in ((V_j \setminus \boundary_j)\cap (A_F\cup B_F))$.   For any such vertex $v \in \boundary_j \cup (V_j \cap (A_F\cup B_F))$, let $\kappa_v = |\ty{v}|-|y_j(v)|$ and let $\kappa= \max_{v\in \boundary_j\cup (V_j \cap (A_F\cup B_F))}\kappa_v$. Set the weight of the newly added edge from $s$ to $v$  to $\kappa-\kappa_v$.  This new graph has only non-negative edge costs. 
 \item Execute Dijkstra's algorithm on this graph beginning from the source vertex $s$. Let $\ell_v$ be the length of the shortest path from $s$ to $v$ as computed by this execution of Dijkstra's algorithm. For each vertex $v \in V_j$, if $\ell_v > \kappa$, then do not change its dual weight. Otherwise, change the dual weight $y_j(v)\leftarrow y_{j}(v) +\lambda_v(\kappa - \ell_v$).
\ignore{
 if $v\in B$, then we set the new dual weight to be $y_j(v) \leftarrow y_j(v) + \kappa - \ell_v$  and if $v \in A$, we set $y_j(v) \leftarrow y_j(v) -\kappa + \ell_v$.}
\end{itemize} 
This completes the description of the \sync\ procedure. The \sync\ procedure executes Dijkstra's algorithm on $\R_j'$ with an additional vertex $s$ and updates the dual weights of  $O(n_j)$ vertices. The total time taken  for this is  $O(m_j + n_j\log{n_j})$ time. To prove the correctness of this procedure, we have to show the following: 
\begin{itemize}
\item[(1)] The new dual weights $y_j(\cdot)$ along with the matching $M_j$ form an $\R$-feasible matching.
\item[(2)] After the \sync\ procedure, for any vertex $v \in \boundary_j \cup (V_j \cap (A_F\cup B_F))$,  $\ty{v}=y_j(v)$.  
\end{itemize}

We give a proof that these two properties hold after executing \sync\ in Lemma \ref{synclemma} of the appendix. The following lemma establishes properties of the \sync\ procedure that will later be used to show that the projection computed by our algorithm are both simple and admissible. For its proof, see appendix section \ref{A:sync}.
\begin{lemma}
\label{lem:sync-zero-slack}
Consider a compressed feasible matching with dual weights $\ty{\cdot}$ assigned to every vertex of $V_H$. For any piece $\R_j$ and any vertex $v \in V_j$, let $\yjs(v)$ denote the dual weight prior to executing \sync, and for any edge $(u,v) \in E_j$, let $\sjs(u, v)$ be the slack prior to executing \sync. Let   $y_j(\cdot)$ denote the dual weights of $V_j$ after this execution.  For any edge $(u,v) \in E_j^H$ with a projection $\dir{P}_{u,v,j}=\langle u=u_{0},u_1,\ldots, u_t, u_{t+1}=v\rangle$, suppose $|\ty{u}|-|y_j^*(u)| \ge \sum_{q = 0}^{t} \sjs(u_q, u_{q+1})$.  Let  $\dir{P}_{s,u_{t},j} $ be any shortest path from $s$ to $u_t$ in $\R_j'$ . Then, 

\begin{enumerate}[(i)]  \item \label{5.10-i} If there exists a shortest path $\dir{P}_{s,u_t,j}$ in $\R_j'$ where $u$ is the second vertex on this path, then after the execution of \sync\ procedure,
 for every $1 \le i\le t-1$, $s(u_{i},u_{i+1})=0$  and $s(u_{t},v)\le |\ty{v}|-|\yjs(v)|$, 
\item \label{5.10-ii}Otherwise, there is no shortest path $\dir{P}_{s,u_t,j}$  in $\R_j'$ with $u$ as its second vertex. Consider $u^*$ to be the second vertex of some $\dir{P}_{s,u_{t},j}$ and $u^* \neq u$. Then,   $u^*\in (\boundary_j \cup (V_j \cap (A_F\cup B_F))$, 
 and   $|\ty{u^*}|- |\yjs(u^*)| > \sum_{(u', v') \in \dir{P}_{u^*, v, j}} \sjs(u', v')$. 
\end{enumerate}
\end{lemma}
Informally, Lemma \ref{lem:sync-zero-slack} states that if the dual weight $\ty{u}$ increased by a sufficiently large amount, then the path $\dir{P}_{s,u_t,j}$ must have 0 slack; i.e., all the slack of $\dir{P}_{u_1, v, j}$ is focused on its last edge $(u_t,v)$. Furthermore, either $u$ is immediately after $s$ on $\dir{P}_{s, u_t, j}$ (i.e., $u=u_1$) or some other vertex $u^*$ is immediately after $s$. In the first case, if the edge $(u,v)$ in $H$ is admissible, then all edges of its projection $\dir{P}_{u,v,j}$ are admissible after \sync. Otherwise, we can conclude that $u^*$ must have experienced a large increase in dual weight magnitude.

Given the correctness of the \sync\ procedure, we can convert a compressed feasible matching into an $\R$-feasible matching by simply applying the \sync\ procedure to all the pieces. This will guarantee that the dual weight of any vertex $v \in V$, is the same across all pieces and in $H$. Let $y(v)$ be this dual weight.  Every edge $(u,v)$ of the graph $\Gr$ belongs to some piece and therefore satisfies the $\R$-feasiblity conditions. Therefore the matching $M$ along with the dual weights $y(\cdot)$ form an $\R$-feasible matching.  \ignore{Also note that every free vertex $v \in A_F$ will have a dual weight $y(v)=0$ and every free vertex $v \in B_F$ will have a dual weight $y(v) = \max_{v' \in B} y(v')$.} Invocation of \sync\ on a piece $\R_j$ takes $O(m_j + n_j\log{n_j})$ time.

\begin{lemma}
\label{lem:feasible-convert}
Given a compressed feasible matching $M$, we can convert it into an $\R$-feasible matching $M$ with a set of dual weights $y(\cdot)$  in $O(m + n\log{n})$ time.
\end{lemma}

\subsection{Second step of the algorithm}
\label{subsec:step2}
\label{sec:algoritm}
In this section, given a compressed feasible matching with $O(n/\sqrt{r})$ free vertices, we show how to compute a compressed feasible matching with  $O\left(\numphases\right)$ unmatched vertices in $\tilde{O}(mn^{2/5})$ time. The second step of our algorithm will execute \fastmatch\ procedure $\numphases$ times. We refer to each execution of the \fastmatch\ procedure as one \textit{phase} of Step 2.  \ignore{After $\numphases$ phases, the algorithm will have a compressed feasible matching with $O(\unmatchedrem)$ unmatched vertices remaining. Finally, the third step of the algorithm matches the remaining vertices by performing $O(\unmatchedrem)$ iterations of a slightly modified version of Hungarian search.}
Next, we give the details of Step 2.  

In each phase of Step 2, the algorithm will invoke the \searchandswitch\ procedure on every active vertex $u \in \activeH$. This procedure takes as input a free active vertex $u$ (boundary or internal) and does a DFS-style search similar to the variant of GT-Algorithm from Section~\ref{subsec:gt}. This search will find paths or cycles of admissible edges in the compressed residual graph. 
When a path or cycle $P$ in the compressed graph is found, the algorithm invokes the \switch\ procedure. This procedure projects $P$ to obtain a path $\dir{P}$ in $G$ where $\dir{P}$ is either an:
\begin{itemize}
\item Augmenting path,
\item Alternating cycle, or
\item Alternating path from a free vertex $u' \in B_{F}$ to some matched vertex $v \in B$.
\end{itemize}
Note that we refer to $P$ (in $H$) as an alternating path, augmenting path, or alternating cycle based on its projection $\dir{P}$.

The \switch\ procedure then \textit{switches} along $\dir{P}$ by setting $M \leftarrow M\ \oplus \dir{P}$. Since switching along a path changes the residual graph, the \switch\ procedure updates the compressed graph accordingly. This process continues until at least one of the following holds:
\begin{enumerate}[(I)]
\item $\lambda_{u} \ty{u}$ has increased by $\sqrt{r}$, \label{stop-cond-1} in step (\ref{searchdualchange}) of \searchandswitch\ below,
\item $u $ is a matched boundary vertex, or, \label{stop-cond-2}

\item  $u$ was a free internal vertex $b_j^A$ of $H$, that no longer exists in $H$, i.e., there are no free active internal vertices in $\R_j$. \label{stop-cond-3}
\end{enumerate}

Recall that we define the slack of any edge $(u,v) \in E_H$ as $\phi(\project{u}{v}{j}) - |\ty{u}|\ +\ |\ty{v}|$. $(u,v)$\ is admissible in $H$ if it has at most $\sqrt{r}$ slack. For any vertex $v \in V_H$, let $\mathcal{A}_v$ be the set of admissible edges of $H$ going out of $v$. Next, we describe the algorithm.

The procedure will conduct a DFS-style search by growing a path $Q=\langle u_0, u_1, \ldots, u_s \rangle$ in\ $H$. Here, $u_0 = u$. The algorithm grows $Q$ by conducting a search at $u_s$ as follows:
\begin{enumerate}[(i)]
\item \label{searchdualchange} If $\mathcal{A}_{u_s} = \emptyset$, then remove $u_s$ from $Q$, set $\ty{u_s} \leftarrow \ty{u_s} +\lambda_{u_s }\sqrt{r}$. If $s> 0$, then continue the search from $u_{s-1}$.Otherwise,  if $s = 0$, the procedure terminates since stopping condition (\ref{stop-cond-1}) is satisfied.
\item  Otherwise, $\mathcal{A}_{u_s} \neq \emptyset$, then find the smallest slack edge in $\mathcal{A}_{u_s}$. 
There could be many edges with the same smallest slack. Among these edges, the algorithm will prefer any edge $(u_s, v)$ such that $v$ is already on the path $Q$.\ Such a vertex would form a cycle in $H$. If none of the smallest slack edges lie on the path, then choose an arbitrary edge with  the smallest slack. Let the chosen edge be $(u_s, v)$ with  $s_H(u_s, v) = \min_{(u_s,v') \in \mathcal{A}_{u_s}} s_H(u_s, v')$. Add $v$ to the path $Q$ as vertex $u_{s+1}$.
\begin{itemize}
\item  If $u_{s+1} \in Q$, then a cycle $C$ is created. Let $u_x = v$, and let $C = \langle u_x, u_{x+1}, ..., u_s, u_x =u_{s+1}\rangle$.  Then, set $Q \leftarrow \langle u_0, u_1, ..., u_{x-1} \rangle$, call the \switch\ procedure\ (described below) on $C$,  and continue the search from $u_{x-1}$. 
\item Otherwise $u_{s+1} \notin Q$, 
\begin{itemize}
\item If any of the following three conditions hold true, then set $P \leftarrow \langle u_0, u_{1}, ..., u_s, u_{s+1}\rangle$ and $Q \leftarrow \{u_0\}$, and invoke \switch\ on $P$.
\begin{itemize}
\item $u_{s+1} \in \boundary \cap B_H$ and $|\ty{u_{s+1}}| \geq \beta$,
\item $u_{s+1} \in \boundary \cap A_H$ and $|\ty{u_{s+1}}| \geq \beta + \max_{(t,u_{s+1}) \in E} \delta_{tu_{s+1}}$, or
\item $u_{s+1} \in A_H^F$.
\end{itemize}
\item If $u_{0}$ is no longer free (i.e., (\ref{stop-cond-2}) is satisfied), or $u_0$ no longer exists (i.e., (\ref{stop-cond-3}) is satisfied), then the \searchandswitch\ procedure terminates.  

\item Otherwise, continue the search from $u_{s+1}$.
\end{itemize}
\end{itemize}
\end{enumerate}

This completes the search portion of \searchandswitch. Next, we will describe the details of \switch\ procedure.

\paragraph{\switch\ procedure.}The \switch\ procedure takes an alternating path, augmenting path or an alternating cycle $P= \langle u_0,\ldots, u_{s+1}\rangle$ in $H$ and computes a path or cycle  $\dir{P}$ in $G$ by projecting every edge of $P$. It also updates the dual weights of every vertex on $\dir{P}$ so that $\dir{P}$ consists only of admissible edges. Then, the algorithm sets $M \leftarrow M \oplus \dir{P}$. This changes the residual graph and its compressed representation. Finally, the procedure updates the compressed graph $H$ to reflect the changes to the underlying residual graph.   

\begin{enumerate}[(a)]
\item \label{presync}
For every edge $(u,v) \in P$, mark the piece it belongs to as affected. Let $\mathbb{R}$ be the set of all affected pieces. Execute \sync\ on every piece $\R_j \in \mathbb{R}$. 

\item \label{updual} Set a value $\alpha \leftarrow \ty{u_0}$. For every $0\le i\le s$, set $\ty{u_i} \leftarrow \ty{u_i} + \lambda_{u_i} s_H(u_i,u_{i+1})$; here $s_H(u_i,u_{i+1})$ is the slack  before the dual weights are updated (i.e., prior to this execution of (\ref{updual})). Execute \sync\ again on every piece $\R_j \in \mathbb{R}$. 

 \item \label{projpath} Suppose $(u,v)$ is an edge of  piece $\R_j$. Project $(u,v)$ to obtain the path $\project{u}{v}{j}$. This can be done by executing Dijkstra's algorithm over $\R_j'$.  Next, combine all the projections to obtain a path or cycle $\dir{P}$ in the residual graph $G_M$. We show that this path or cycle is a simple path or cycle consisting only of admissible edges. 

\item \label{altpath} If $u_{s+1} \in A_H\setminus A_H^F$ and $P$ is an alternating path, then  $u_{s+1}$ is a matched vertex. Let $(u_{s+1},v_{s+1})$ be the edge in the matching $M$ belonging to the piece $\R_j$. \ignore{Set $y_j(u_{s+1}) \leftarrow y_j(u_{s+1}) - s(u_{s+1}, v_{s+1})$. Similarly, for each $\R_{j'}$ such that $u_{s+1} \in V_{j'}$ set $y_{j'}(u_{s+1}) \leftarrow y_{j'}(u_{s+1}) - s(u_{s+1}, v_{s+1})$. Also, set $\ty{u_{s+1}} \leftarrow \ty{u_{s+1}} - s(u_{s+1}, v_{s+1})$.} Execute \reduceslack$(v_{s+1})$. This makes the edge $(u_{s+1}, v_{s+1})$ admissible with respect to $M_j, y_j(\cdot)$ without violating compressed feasibility.  Also, add $(u_{s+1}, v_{s+1})$ to $\dir{P}$ and add $\R_j$ to $\mathbb{R}$. $\R_j$ is added to the affected set because the edge $(u_{s+1}, v_{s+1})$ will change to a non-matching edge during (\ref{augment}), which will affect the edges of $E_j^H$.

\item \label{augment} 
Update the matching $M$ along $\dir{P}$ by setting $M \leftarrow M \oplus \dir{P}$.  By Lemma \ref{lem:admissible-feasible}  the new matching is $\R$-feasible within all affected pieces. In the event that $\dir{P}$ was an alternating path to some internal vertex $v \in B_j \setminus \boundary_j$, $v$ is now an inactive free internal vertex in $G$. It is possible that the dual weights of the inactive free internal vertices in $\R_j$ differ.  Therefore, call \reduce$(b_j^{\mathcal{I}}, \beta)$. The residual graph changed during this step. Therefore, call the \construct\ procedure on every affected piece $\R_j \in \mathbb{R}$ and recompute the edges in $E_{j}^H$ along with their costs. 

\item \label{augment-reduce} If $P$ is an alternating path or augmenting path, and $u_0$  still exists as a vertex $b_j^{\mathcal{A}} \in B_H^F$, then execute \reduce$(u_0, \alpha)$. This effectively resets all dual weights associated with $u_0$ to their values prior to step (\ref{updual}). Since the dual weights $\ty{\cdot}$ of vertices of $\activeH$  differ by at most $\sqrt{r}$, and $\alpha \geq 0$, the preconditions to $\reduce$ are satisfied.

\end{enumerate}

\subsection{Third step of the algorithm}
After Step 2, the algorithm has a compressed feasible matching with $O(\unmatchedrem)$ unmatched vertices remaining. This can be converted into an $\R$-feasible matching in $O(m + n\log{n})$ time by Lemma \ref{lem:feasible-convert}. Next, we describe a procedure for computing an $\R$-optimal matching from this $\R$-feasible matching. In section \ref{sec:planar}, we will discuss further optimizations which lead to an even better running time for Step 3 on planar graphs, based on the results from~\cite{soda-18}.

\paragraph{Step 3.} Given an $\R$-feasible matching with $O(\unmatchedrem)$ unmatched vertices remaining, the algorithm can use $O(\unmatchedrem)$ iterations of Hungarian search to match the remaining vertices, matching one vertex each iteration. Let $G'$ be the graph $G$ with all edges having weight equal to their slack with respect to the $\R$-feasible matching. In each iteration of Step 3, the algorithm executes Dijkstra's algorithm on $G'$ from the vertices of $B_F$ to the vertices of $A_F$ in order to find the minimum total slack augmenting path. For each vertex $v$ in $G$, let $\ell_v$ be the distance assigned by Dijkstra's algorithm. Let $P$ be a shortest augmenting path found, and let $\ell_{max}$ be the distance to the free vertex of $A_F$ in $P$. Then for each vertex  with $\ell_v \leq \ell_{max}$, the algorithm sets $y(v) \leftarrow y(v) + \lambda (\ell_{max} - \ell_v)$. This dual weight change ensures all edges of $P$ are admissible, while also preserving $\R$-feasibility. The algorithm then sets $M \leftarrow M \oplus P$, which increases the matching size by one while preserving $\R$-feasibility.

\section{Analysis of Algorithm}
Step $1$ of the algorithm computes a $\R$-feasible matching $M$ with all but $O(n/\sqrt{r})$ unmatched vertices. This matching is then converted into a compressed feasible matching. Step 2 iteratively computes alternating paths, augmenting paths and cycles and switches the edges along them. While doing so, it maintains a compressed feasible matching. In $O(n^{3/2}/r^{1/4} + mr)$ time, we obtain a compressed feasible matching with no more than $O(\sqrt{n}/r^{1/4})$ unmatched vertices. Step 3 computes the remaining augmenting paths iteratively by doing a simple Hungarian search in $O(m\sqrt{n}/r^{1/4})$  time. For $r=n^{2/5}$, the running time of the algorithm is $O(mn^{2/5})$. To complete the analysis, we need to show the correctness and efficiency of Step 2 of the algorithm. We prove the correctness of Step 2 in Section~\ref{sec:step2-correctness} and the efficiency of Step 2 in Section~\ref{sec:step2-efficiency}.

\label{analysis}
\subsection{Correctness}
\paragraph{Overview of the proof.}
We show that the paths computed in the second step of the algorithm satisfy certain properties (P1)--(P3). Using these properties, we show that these paths have simple (see Lemma \ref{lem:no-intersect}) and admissible (see Lemma \ref{lem:zero-slack-project}) projections. Recollect that switching the edges on an admissible path maintains $\R$-feasibility, by Lemma \ref{lem:admissible-feasible}. Using this, we show that our algorithm maintains compressed feasibility. 
\label{sec:step2-correctness}
\paragraph{Properties of paths.} Our algorithm computes paths and cycles in $H$ that satisfy three properties stated below. Let $P=\langle u_1,u_2,\ldots,u_t\rangle$  be any path or cycle in $H$ such that,

\begin{enumerate}[(P1)]
\item \label{inv:A1} \label{P1} For any edge $(u_{i},u_{i+1})$ on $P$, $s_H(u_{i},u_{i+1}) \leq \sqrt{r}$, i.e., $(u_{i},u_{i+1})$ is admissible,
\item \label{inv:A2} \label{P2} For any edge $(u_{i}, u_{i+1})$ on $P$, $s_H(u_{i}, u_{i+1}) = \min_{(u_{i}, v') \in E_H} s_H(u_{i}, v')$,
\item \label{inv:A3} \label{P3} For any edge $(u_i, u_k)$ with $u_i, u_k \in P$ such that  $k < i$, $s_H(u_i, u_k) > s_H(u_i, u_{i+1})$. 
\end{enumerate}

The next two lemmas show that when the \switch\ procedure is called on any path $P$ that satisfies (P1)--(P3), its projection $\dir{P}$ is simple and consists of only admissible edges. 

\begin{lemma}
\label{lem:no-intersect}
Given a compressed feasible matching, let $P=\langle u_1,\ldots,u_t\rangle$ be any (not necessarily simple) path in $H$ that satisfies properties (P\ref{P2}) and (P\ref{P3}).   Then, for any two edges  $(u_i,u_{i+1})$ and $(u_k,u_{k+1})$ on $P$ with $i <k$ that belong to piece $E_{j}^{H}$, their projections  $\dir{P}_{u_i,u_{i+1},j}$ and $\dir{P}_{u_k,u_{k+1},j}$ are interior-disjoint.        
\end{lemma}
\begin{proof}
Without loss of generality, assume that the dual weights $\ty{\cdot}$ and $y_j(\cdot)$ are synchronized. For the sake of contradiction let $\dir{P}_{u_i,u_{i+1},j}$ and $\dir{P}_{u_k,u_{k+1},j}$  intersect in the interior at a vertex $x$ in some piece $\R_j$. Since $x$ is common to both the projections, it immediately follows that there is a path from $u_i$ to $u_{k+1}$  and a path from $u_k$ to $u_{i+1}$, both passing through $x$. This implies 
\begin{equation}
\label{eq:slacksum}
s_H(u_i,u_{i+1})+s_H(u_k,u_{k+1}) \ge s_H(u_i,u_{k+1})+s_H(u_k,u_{i+1}). \end{equation}

  From property (P\ref{P2}), we have that $s_H(u_i,u_{i+1}) \leq s_H(u_i,u_{k+1})$ and $s_H(u_k,u_{k+1}) \leq s_H(u_k,u_{i+1})$. This along with~\eqref{eq:slacksum} implies that  $s_H(u_k,u_{k+1}) = s_H(u_k,u_{i+1})$ contradicting  (P\ref{P3}) since $i < k$.
\end{proof}
The following is a straight-forward corollary of Lemma~\ref{lem:no-intersect}.
\begin{cor}
\label{cor:no-intersect}
Given a compressed feasible matching, let $P=\langle u_1,\ldots,u_t\rangle$ be a simple path (resp. simple cycle) in $H$ that satisfies properties (P\ref{P2}) and (P\ref{P3}). Then, the projection $\dir{P}$ of $P$ is a simple path (resp. cycle).  
\end{cor} Let $P$ be a path (or cycle) that satisfies (P2), Consider an execution of \switch\ on  path $P$ and let $\dir{P}$ be the projection computed  in step (\ref{projpath}) in \switch.  All edges  of $\dir{P}$ are admissible.

\begin{lemma}
\label{lem:zero-slack-project}
Let $P$ be the path (or cycle) that is projected during step (\ref{projpath}) of \switch. Assume that $P$ satisfies property (P1) and (P\ref{P2}) at the beginning of \switch. Then for every edge $(u,v)$ on $P$ with projection $P_{u,v,j}$, every edge of $\dir{P}_{u,v,j}$ is admissible.
\end{lemma}
\begin{proof}
At the end of step (b) of the \switch\ procedure, \sync\ is called on the piece containing $(u,v)$. In the \sync\ procedure, recollect that we add a new vertex $s$ and connect it with all the boundary vertices to create a graph $\R_j'$. For this \sync\ procedure (of step (b)), we use the notations from Lemma \ref{lem:sync-zero-slack}.  Recollect that, due to the execution of \sync\ procedure in step (a),   $\yjs(v)=\ty{v}$, for all $v \in V_j^H$. Step (b) will only increase the $\ty{\cdot}$ for vertices along $P$. Let $u_t$ be the vertex that appears before $v$ in $\dir{P}_{u,v,j}$. 

Suppose there is no shortest path from $s$ to $u_t$ in $\R_j'$ with $u$ as the second vertex, then let $u^*$ be the second vertex on some shortest path $\dir{P}_{s,u_t,j}$.  From    Lemma \ref{lem:sync-zero-slack} (ii), it follows that $|\ty{u^*}|-|\yjs(u^*)|  > s_H(u^*,v)$, i.e., the change in $\ty{u^*}$ in   Step (b) of the \switch\ procedure is greater than $s_H(u^*,v)$. Since step (b) updates the dual weight $\ty{u^*}$, $u^*$ is on the path $P$ and let $y$ be the vertex that appears after $u^*$ on $P$. The change in dual weight $\ty{u^*}$ is exactly $s_H(u^*,y)$. Therefore, $s_H(u^*,y) > s_H(u^*,v)$ contradicting (P2). We conclude that Lemma~\ref{lem:sync-zero-slack} (i) holds, and the vertex $u$ must be the second vertex on some shortest path from $s$ to $u_t$.
 
From Lemma~\ref{lem:sync-zero-slack} (i), it follows that every edge on $\dir{P}_{u,v,j}$ except the last edge has zero slack. Moreover the slack on the last edge is $|\ty{v}|-|\yjs(v)|$ which is less than or equal to $\sqrt{r}$ (by (P1)). Since $v$ is a boundary vertex, the last edge $(u_t,v)$ is also an admissible edge.
\end{proof}

The second step of the algorithm maintains the following invariants.

\begin{enumerate}[(A)]
\item \label{inv:min-slack} Let $Q$ be the search path in the \searchandswitch\ procedure. Then path properties (P\ref{P1}), (P\ref{P2}), (P\ref{P3}) hold for $Q$.
\item \label{inv:cf} After any step of \searchandswitch\ or \switch, the matching $M$ and the sets of dual weights $\bigcup_{\R_j}y_j(\cdot)$ and $\ty{\cdot}$ form a compressed feasible matching.
\end{enumerate}

Given Invariant (\ref{inv:min-slack}), it is easy to show that any projected path (or cycle) $\dir{P}$ in step (\ref{augment}) is both simple and admissible. From Corollary \ref{cor:no-intersect}, the projection of $Q$ is simple. In step (\ref{altpath}), a single matching edge $(a,b)$ may be added to $\dir{P}$. However, \reduceslack\ is called to ensure that this edge has 0 slack.

\begin{lemma}
\label{lem:admis-nointersect}
Let $P$ be a path or cycle in $H$ sent as input to \switch. Let $\dir{P}$ be the path or cycle that is a projection of $P$ prior to step (\ref{augment}) of \switch. Then $\dir{P}$ is a simple path or cycle consisting of admissible edges.
\end{lemma}

Next, we discuss the proof of Invariant (\ref{inv:min-slack}). The \searchandswitch\ procedure adds the smallest slack admissible edge going out of the last vertex on $Q$. In the case of a tie, the algorithm will prefer adding vertices already on the search path, in which case a cycle is detected immediately and \switch\ is invoked. Regardless, by construction, the edge added satisfies (P1)--(P3).  During the execution of the \searchandswitch\ procedure only $\ty{\cdot}$ values are modified for vertices from which the search backtracks. For any such vertex $v$ from which the search backtracks, the algorithm sets $\ty{v} \leftarrow \ty{v}+\lambda_v \sqrt{r}$. This increases the magnitude of the dual weight  $\ty{v}$. So, for any vertex $u \in Q$, the slack $s_H(u,v)$ only increases. Since $v$ is not on the path,   (P1)--(P3) continue to hold.\ 

During the execution of \switch\ procedure for an alternating path or an augmenting path, $Q$ is set to $\emptyset$. Therefore, (P1), (P2) and (P3) hold trivially. In the case of an alternating cycle, however, $Q$ may contain edges after \switch. Note that the magnitude of $\ty{\cdot}$ values increase for vertices not on $Q$ in  step (b) of \switch\ procedure. Since the magnitude only increases, (P1)--(P3)  holds. Since step (d) is not executed for a cycle, the only other step where slack on the edges of $H$ are changed are in step (e). In the following lemma, we show that, for any such $Q$, (P1)--(P3) hold after the execution of step (e) of the \switch\ procedure.

\begin{lemma} 
\label{lem:A1-helper}
Assume that the path $Q$ satisfies properties (P\ref{inv:A1}), (P\ref{inv:A2}) and (P\ref{inv:A3}) prior to executing step (\ref{augment}) of  \switch\ with a cycle $C$ as input. Then (P\ref{inv:A1}), (P\ref{inv:A2}) and (P\ref{inv:A3}) hold for $Q$ after  step (\ref{augment}).
\end{lemma}
\begin{proof}
Consider the case during the execution of \searchandswitch\ right after $v$ is added to $Q$ as $u_{s+1}$. Since before adding $u_{s+1}$ to $Q$, $u_{s+1}$ was already on $Q$, a cycle is created. By Lemma \ref{lem:no-intersect}, the projection of $Q$ after the addition of $u_{s+1}$ has no self intersection except at $u_{s+1}$. We then remove the cycle $C$ from $Q$ and call the \switch\ procedure on $C$.
It follows that the projection of $C$ and $Q$ (after the removal of cycle as described in (ii) of \searchandswitch) has no intersections.

Let  $(u,v)$ be an edge of $Q$ before execution of step (e).    Let $\dir{P}_{u,v,j}$ be the projection of $(u,v)$ prior to step (\ref{augment}), and let $\dir{P}'_{u,v,j}$ be the projection after step (\ref{augment}). For every $(y,z)\in C$, from the discussion above, we have that $\dir{P}_{u,v,j} \cap \dir{P}_{y,z,j} = \emptyset$. During step (e), only edges of  $\dir{P}_{y,z,j}$ change direction and therefore, $\dir{P}_{u,v,j}$ continues to be a directed path  after step (\ref{augment}). Therefore, $s(\dir{P}'_{u,v,j}) \leq s(\dir{P}_{u,v,j})$ and (P1) holds. 

  Next, we show that for any vertex $v'\in V_j^H$ (possibly $v'=v$), any new projection $\dir{P}'_{u,v',j}$  created after step (e) has $s(\dir{P}'_{u,v',j}) > s(\dir{P}_{u,v,j})$ implying (P2) and (P3) hold. Since $\dir{P}'_{u,v',j}$ was created from switching along $C$, $\dir{P}'_{u,v',j}$ must intersect with some projection $\dir{P}_{y,z,j}$ of an edge $(y,z)$ in $C$. Step (b) of the \switch\ procedure increases the magnitude of the dual weight of $z$ say by $\Delta_z$ and since $(u,v)$ was the smallest slack edge out of $u$, it follows that
$$s_H(u,v) \le s_H(u,z) - \Delta_z.$$ 
Let $\dir{P}^{R}_{z,y,j}$ be the path obtained by reversing the edges of $\dir{P}_{y,z,j}$. Let $(x,x')$ be the first edge in the intersection of  $\dir{P}'_{u,v',j}$ and $\dir{P}^{R}_{y,z,j}$ as we walk along $\dir{P}_{u,v',j}$ (any two alternating paths that intersect will intersect along at least one edge). Right before switching the edges of the cycle, from  Lemma~\ref{lem:zero-slack-project}, it follows that every edge on $\dir{P}_{y,z,j}$ is zero slack except for the edge incident on $z$ which has a slack of $\Delta_z$. Consequently,
$\dir{P}_{u,x,j}$ has a slack of $s_H(u,v)$,  i.e.,

$$s(\dir{P}_{u,x,j} )\ge s_H(u,v).$$

\noindent $(x,x')$ is on $\dir{P}^R_{y,z,j}$ , from Lemma~\ref{lem:admissible-feasible}, $(x,x')$ has a positive slack. Therefore, $s(\dir{P}_{u,v',j} )\ge s(\dir{P}_{u,x',j} )> s_H(u,v)$, as desired. \end{proof}

Finally, we show Invariant (B) and establish that the compressed feasibility conditions hold throughout the second step. Without loss of generality, let us assume that the condition holds at the start of an execution of \searchandswitch\ procedure. We will show that this execution of \searchandswitch\ procedure and the subsequent execution of \switch\ does not violate  compressed feasibility conditions (a)--(e).

\searchandswitch\ only changes the dual weights in case (\ref{searchdualchange}) of \searchandswitch, where the procedure sets $\ty{u_s} \leftarrow \ty{u_s} +\lambda_{u_s }\sqrt{r}$. Since this operation only increases the magnitude of $\ty{u_s}$, conditions (a), (b), and (e) of compressed feasibility are satisfied. Note that condition (b) also requires the dual weights $\ty{\cdot}$ to be at most $\sqrt{r}$ apart. However, this is satisfied because a free internal vertex executes (\ref{searchdualchange}) exactly once per phase. Also, note that at the beginning of each phase, all vertices $v \in \activeH$ have the same dual weight. Condition (c) is unaffected. Finally, observe that, since the dual weight change only occurs when there are no admissible edges outgoing from $u_s$, condition (\ref{cf-d}) continues to hold.

Suppose $P$ is the path or cycle sent to the \switch\ procedure and suppose $\dir{P}$ is its projection.  Before projecting $P$, the dual weight $\ty{v}$ for every vertex $v \in P$ is increased by the slack of the edge $(v,v')$ in $P$. From (P2), $(v,v')$ is the smallest slack edge out of $v$ and therefore the increase $\ty{v}\leftarrow \ty{v}+\lambda_v s_(v,v')$ reduces the slack of $(v,v')$ to $0$ and all other edges continue to have a non-negative slack. Therefore, the change does not violate $H$-feasibility and also preserves (a), (b) and (e).  The projection computed by \switch\ is a simple path or cycle consisting only of admissible edges. Switching edges ($M \leftarrow M\oplus \dir{P}$) on this path does not violate $\R$-feasibilty (Lemma~\ref{lem:admissible-feasible}) of any of the affected pieces. So, after switching the edges, the new matching in each of the affected pieces $\R_j$, along with the dual weights $y_j(\cdot)$ form an $\R$-feasible matching. The \construct\ procedure will recompute edges of $H$ which, from Lemma~\ref{lem:construct} satisfies $H$-feasibility. As discussed in Section~\ref{algorithmforscale},  \sync, \reduce\ and \reduceslack\ preserve compressed feasibility as well. 

Therefore, our algorithm iteratively matches vertices while maintaining compressed feasibility. In the following, we discuss the efficiency of our algorithm.

\subsection{Efficiency of Step 2}
\label{sec:step2-efficiency}
Step 2 of our algorithm invokes \searchandswitch\ on free internal vertices of $B_F^{\mathcal{A}}$. This\ procedure computes cycles and paths in $H$ and passes them to the \switch\ procedure. Let $\langle P_1, P_2, ..., P_N \rangle$ be the sequence of paths and cycles generated by the second step of the algorithm. These paths and cycles are sorted in the order in which they are computed in Step 2, with $P_i$ being the $i$th such path or cycle. Note that the \switch\ procedure is executed for each such $P_i$. Let $\dir{P}_i$ be the projection of $P_i$ as computed by the \switch\ procedure. Let $\Mi{0}$ be the matching at the start of Step 2.  Then $\Mi{i} \leftarrow \Mi{i-1} \oplus \dir{P}_i$. The operations conducted by the algorithm after the execution of the \switch\ procedure on $P_{i-1}$ until the end of the execution of the \switch\ procedure on $P_i$ is referred to as the $i$th \emph{iteration} of the algorithm. Let $B^i_F$ denote the free vertices of $B$ at the start of iteration $i$. Every compressed feasible matching can be converted into an $\R$-feasible matching by applying \sync\ to all pieces. For the proof,  let $y^{i}(\cdot)$ denote the dual weights of this $\R$-feasible matching at the start of iteration $i$. For any path $\dir{P}'$ in $G$, let $s(\dir{P}') = \sum_{u',v' \in \dir{P}'} s(u',v')$.

\paragraph{Efficiency of \searchandswitch.}To bound the time taken by the DFS search portion of the \searchandswitch\  procedure (i.e., the portion outside of the \switch\ procedure), it suffices if we bound the total time taken to find the smallest slack edge  from $u_s$ during all executions of \searchandswitch\ in Step 2. Recollect that if there are ties, we would like to pick the smallest slack edge to a vertex on the path. We accomplish this by explicitly maintaining, for every vertex $u$, a binary search tree (BST) of all the edges going out of $u$. The slack is used as the key value and ties broken by prioritizing edges for which the other vertex  is on the search path.  The \construct\ procedure can be modified to create and update this tree without any asymptotic increase in execution time.  

During the \searchandswitch\ procedure, dual weights of certain vertices may change, affecting the slacks on edges. When necessary, we update all affected BSTs to reflect the new slacks. Updating the dual weight of $u$ will uniformly change the slacks on all the edges going out of $u$. So, the relative ordering of these edges in the BST of $u$ does not change. However,  for an edge from $u$ to $v$, if (a) the dual weight $\ty{v}$ changes, or (b) $v$ enters the search path, then we have to update the BST of $u$. In case (b), since $v$ is on the path, we have to prioritize the edge $(u,v)$ over all other edges with the same slack. We will first bound the total time to update BSTs for case (a). In case (a), the dual weight of $v$ can change in two places: (i) a search backtracked from $v$ causing the magnitude of the dual weight $\ty{v}$ to increase by $\sqrt{r}$, and (ii)  $v$ lies on some path/cycle $P_i$ and the \switch\ procedure updated $\ty{v}$ in step (\ref{updual}) prior to switching the edges.  

 Note that, any vertex $v$  whose dual weight exceeds $(\beta+\max_{v' \in N(u)}\delta_{vv'})$ becomes inactive and so, the number of dual weight changes of $v$ of type (i) cannot be more than $(\beta+\max_{v' \in N(u)}\delta_{vv'})/\sqrt{r}$ per vertex in $H$. We can upper bound $\delta_{uv}$ by observing that $\frac{m_j n}{m\sqrt{r}} = O(\frac{r^2 n}{n\sqrt{r}}) = O(r^{3/2})$ and so the total number of dual updates of type (i) for $v$ is no more than $\beta/\sqrt{r}+r$. The dual updates of type (ii) over all vertices $v$  is bounded by the total length of all paths and cycles in $H$ computed by the \searchandswitch\ procedure, i.e, $O((n/\sqrt{r}) \log n)$ (see Corollary~\ref{cor:aplengths}).  

Whenever the dual weight of a vertex $v \in V_H$ changes, we update the BST any $u'$ such that $(u',v)\in E_H$. Therefore, the total number of BST updates is bounded by the in-degree of $v$ in $H$. Let $d_v$ be the in-degree of $v$ and recollect that $\theta_v$ is the number of pieces of the $r$-clustering that $v$ participates in. From the properties of an $r$-clustering,  $d_v\le \theta_v\sqrt{r}$ and $\sum_{v\in V_H}\theta_v = O(n/\sqrt{r})$.  

If $v \in B_H$, from Lemma~\ref{lem:searchcost1}, the in-degree of $v$ is no more than $\sqrt{r}$. Therefore, the total work done across all executions of the \searchandswitch\ procedure for dual updates of type (i) is $O((n/\sqrt{r})(\beta/\sqrt{r} +r)\sqrt{r})$ and the total for type (ii) is $O( ((n/\sqrt{r})\log n)\sqrt{r} ) $ for a combined total of $ O(n^{3/2}/r^{1/4} + nr+ n \log n)$.

If $v \in A_H$,  the total time to update the BST  due to a dual weight change of type (i)  is bounded by $ d_v(\beta/\sqrt{r}+r)$ for any vertex $v$ and $\sum_{v \in A_H}d_v(\beta/\sqrt{r}+r)$ across all vertices of $A_H$. For type (ii), since $v$ is on a path/cycle $P_i$ computed by the \searchandswitch\ procedure,  from Lemma~\ref{lem:searchcost}, the next visit to $v$ will trigger an increase of the dual weight $\ty{v}$ by $\sqrt{r}$.  Therefore, we can charge the time to update the BSTs  to the increase in dual weight of $v$ during the next visit. The total time to update the BSTs due to type (ii) dual weight changes of $v$ is $d_v (\beta/\sqrt{r} +r +1)$ and across all vertices of $A_H$, $\sum_{v\in A_H}d_v (\beta/\sqrt{r} +r +1)$. Combining the totals for cases (i) and (ii) gives the total work for vertices of $A_H$ as at most
\begin{eqnarray*}
(2(\beta/\sqrt{r}) + 2r + 1) \sum_{v\in V_H} d_v &=& O(((\beta/\sqrt{r}) + r + 1)\sqrt{r} \sum_{v\in V_H} \theta_v)\\
&=& O(((\beta/\sqrt{r}) + r + 1)n)\\ 
&=& O(n^{3/2}/r^{1/4}+nr).
\end{eqnarray*}

For case (b), we note that every vertex  $v$ that has entered the search path in \searchandswitch\  will either be backtracked from (type (i)) or be on some path or cycle $P_i$ (type (ii)). Using  identical arguments to cases (i) and (ii), we can bound the total time for BST updates across all executions of the \searchandswitch\ procedure by $O(n^{3/2}/r^{1/4}+nr)$.

\begin{lemma}
\label{lem:searchcost1}
For any vertex $v \in B_H$, the number of edges in $H$ that are directed towards $v$ is $O(\sqrt{r})$. 
\end{lemma}
\begin{proof}
Since  $v\in B_H$,  the in-degree of $v$ in the residual graph $\dir{G}_M$ is $1$. Let this edge be $(u,v)$ from the piece $\R_j$. Every edge of $H$ directed towards $v$ should contain $(u,v)$ in its projection. Therefore, all incoming edges of $v$ should be in $E_j^H$ implying that the in-degree of $v$ is $O(\sqrt{r})$. 
\end{proof}
\begin{lemma}
\label{lem:searchcost}
Consider a vertex $v \in  P_i$ where $v$ is a boundary vertex and $v \in A_H$. Then, after the execution of \switch\ procedure on $P_i$, $v$ does not have any admissible edge of $H$ going out of it. Therefore, if $v$ is visited again by the \searchandswitch\ procedure, the algorithm will immediately backtrack from $v$ and the magnitude of $\ty{v}$ will increase by $\sqrt{r}$.
\end{lemma}
\begin{proof}
Suppose $v \in A_H$, then let $(v,v')$ be the matching edge after the execution of \switch. Since $(v,v')$ was admissible prior to the execution of \switch, from Lemma~\ref{lem:admissible-feasible}, after the execution of \switch, the slack on $(v,v')$ is at least $\delta_{vv'} \ge 2\sqrt{r}$. Therefore, every edge going out of $v$ has a slack of at least $2\sqrt{r}$, implying that there are no admissible edges going out of $v$. Therefore, if the \searchandswitch\ procedure visits $v$ again, it will backtrack and $\ty{v}$ will increase by $\sqrt{r}$.     
\end{proof}
The \switch\ procedure synchronizes, projects, augments and then re-constructs the affected pieces. The most expensive of these operations is the \construct\ procedure; therefore, the time taken by \switch\ is upper bounded by the time taken to re-construct the edges of $H$ for every affected piece. The following sequence of Lemmas bounds the time taken by the \switch\ procedure.    

\begin{lemma}
Given a compressed-feasible matching before iteration $i$ of Step 2,
\begin{equation}
\label{eq:bfdelta}
|B^{i}_F|\Delta_{i} \le O(n).
\end{equation}
Here, $\Delta_{i}$ is the minimum dual weight among all free vertices of $B$.
\end{lemma}

\begin{proof}
Using the \sync\ procedure, we can create an $\R$-feasible matching $M^{(i-1)}, y^i(\cdot)$ from the compressed feasible matching. From compressed feasibility, we have that for every free vertex $a \in A^{i}_F$, $y(a) = 0$. Consider some optimal matching $M^*$. $M^{(i-1)} \oplus M^*$ forms $n-i$ augmenting paths and alternating cycles. Let $\csetsymdiff$ be the set of cycles in $M^{(i-1)} \oplus M^*$ and let $\psetsymdiff$ be the set of augmenting paths in $M^{(i-1)} \oplus M^*$.

From~\eqref{netcost3}, $\phi(M^{(i-1)}\oplus M^{*}) = \dist(M^{(i-1)}) - \dist(M^{*}) + \sum_{(u,v) \in M^{*}\oplus M^{(i-1)}} \delta_{uv}$. Cost of the optimal matching is $O(n)$ and using the arguments of~\eqref{eq:errorsum},   $\sum_{(u,v) \in M^{*}\oplus M^{(i-1)}} \delta_{uv}=O(n)$. Since the dual weights of free vertices of $A$ are $0$, from Lemma~\ref{lem:feasrel2}, the cost of $M^{(i-1)}$ is also $O(n)$. Therefore,
\begin{equation}
\label{eq:sumnetcosts}
\sum_{P \in \csetsymdiff \cup \psetsymdiff} \phi(P) \le O(n) . 
\end{equation} 

 Each augmenting path in $\psetsymdiff$ is a path between a free vertex $b$ of $B$ to a free vertex $a$ of $A$. From properties of compressed feasibility, we know that $y^i(b) \ge \Delta_{i}$ and $y^i(a)=0$. Plugging this in~\eqref{slackcost}, we get~\eqref{eq:bfdelta}. 
\end{proof}

After Step 2, $\Delta$ is at least $\beta$. Therefore, the number of unmatched vertices is at most $O(n/\beta) = \unmatchedrem$. As a corollary, we can show the following:
\begin{cor}
\label{sumDelta}
 Recollect that $\langle P_1,\ldots, P_N\rangle$ are the set of paths and cycles computed by Step 2 of our algorithm and $\dir{P}_i$ is the projection of $P_i$. Let $B_F^i$ be the free vertices  and let $y^i(\cdot)$ denote the dual weights before switching along $\dir{P}_i$. Define $\Delta_i = \min_{v \in B_F^i} y^i(v)$. Let $\kappa_i = 1$ if $\dir{P}_i$ is an augmenting path and $0$ otherwise. Then $\sum_{i=1}^N \kappa_i\Delta_i = O(n \log n)$.
\end{cor} 
\begin{proof}
Suppose $\dir{P}_i$ is an augmenting path. From equation \eqref{eq:bfdelta}, we have that $\Delta_i = O(n) / |B_F^i|$. $|B_F^i| = n - i+1$. After augmenting along $\dir{P}_i$, the number of free vertices reduce by $1$ and summing over all $i$ when $\kappa_i$ is $1$, yields a harmonic series in the denominator. Therefore,  $\sum_{i=1}^n \kappa_i\Delta_i = O(n \log n)$.
\end{proof}

\begin{lemma}
\label{lem:sumdeltapaths}
\begin{equation}
\label{eq:sumdeltapaths}
\sum_{i=1}^N \sum_{(u,v) \in \dir{P}_i}\delta_{uv} = O(n \log n).
\end{equation}
\end{lemma}
\begin{proof}
From Lemma~\ref{lem:admis-nointersect}, all edges on any projection $\dir{P}_i$ are admissible. 

Suppose $\dir{P}_i$ is an alternating cycle consisting of admissible edges and let u be any vertex on $\dir{P}_i$.  Then from~Lemma~\ref{lem:ynetcost}, $$0=\lambda_u y^{i}(u) - \lambda_u y^{i}(u) \ge \wt(M\oplus \dir{P}_i) - \wt(M) + \sum_{(p,q) \in \dir{P}_i} \delta_{pq}/2.$$
Suppose $\dir{P}_i$ is an alternating path and let $u$ be the first vertex and $v$ be the last vertex of $\dir{P}_i$. Then, we know that $\lambda_vy^{i}(v)> \beta$ and $\lambda_uy^i(u) <\beta$. Therefore,

$$0>\lambda_u y^{i}(u) - \lambda_v y^{i}(v) \ge \wt(M\oplus \dir{P}_i) - \wt(M) + \sum_{(p,q) \in \dir{P}_i} \delta_{pq}/2.$$
Suppose, $\dir{P}_i$ is an augmenting path with $u$ as its first vertex and $v$ as its last vertex. Then, from Lemma~\ref{lem:ynetcost} and the
fact that the dual weight of $y^i(v)= 0$ and $y^i(u)\le \Delta_i+\sqrt{r}$, we have $$\Delta_i +\sqrt{r}>\lambda_u y^{i}(u) - \lambda_v y^{i}(v) \ge \wt(M\oplus \dir{P}_i) - \wt(M) + \sum_{(p,q) \in \dir{P}_i} \delta_{pq}/2.$$   Let $\kappa_i$ be $1$ if $\dir{P}_i$ is an augmenting path and $0$ otherwise. Adding over all $1\le i \le N$ and since there are at most $n/\sqrt{r}$ augmenting paths, we immediately get
$$\sum_{i=1}^N \kappa_i\Delta_i +\sqrt{r}\frac{n}{\sqrt{r}}\ge \wt(M^{(N)})-\wt(M^{(0)})+ \sum_{i=1}^N \sum_{(p,q) \in \dir{P}_i}\delta_{pq}/2.$$
From Corollary~\ref{sumDelta} and the fact that $\dist(M^{(0)})$  and $\dist(M^{(N)})$ is $O(n)$, the lemma follows.
\end{proof}

\begin{cor}
\label{cor:aplengths}
\begin{equation}
\sum_{i=1}^N |P_i| = O((n/\sqrt{r})\log{n}). 
\end{equation}
\end{cor}
\begin{proof}
By \eqref{eq:sumdeltapaths}, the total $\delta$ of all projection edges is $O(n\log{n})$. Each boundary edge has a $\delta$ of at least $2\sqrt{r}$ therefore, there can be at most $O(n/\sqrt{r})$ boundary vertices in the projections. Every edge of $H$ has at least one boundary vertex, except the edges from vertices of $B_H^F$ to vertices of $A_H^F$. However, there can be at most $O(n/\sqrt{r})$ such edges used, since each such edge corresponds to an augmenting path, and there are only $O(n/\sqrt{r})$ free vertices at the start of Step 2.
\end{proof}

\begin{lemma}
The total time taken for all calls to \sync, all projections, and all calls to \construct\ during Step 2 is $O(mr\log{m}\log^2{n})$.
\end{lemma}
\begin{proof}
Other than a single call to \construct\ per piece at the beginning of the algorithm, and a single call to \sync\ per piece at the end of the algorithm, \sync, \construct, and projections only occur as part of the \switch\ procedure, once per affected piece. Out of all these procedures, the time taken for \construct\ dominates with total time, taking time $\bruteforceconstruct$ per piece, so bounding the time for \construct\ is sufficient. To account for different piece sizes, we first divide the pieces into $O(\log{m})$ groups, where the $g$th group contains pieces $\R_j$ with $2^g \leq m_j < 2^{g+1}$. Since there are at most $m$ edges in total, the $g$th group can contain at most $O(m / 2^g)$ pieces. We will show that the total work done for each group over all calls to \construct\ is $O(mr\log^2{n})$. 

First, consider any group where $2^g = O(\sqrt{r})$. By Corollary \ref{cor:aplengths}, the maximum number of affected pieces for $g$ is $O(n/\sqrt{r}) \log{n}$. Since the number of edges in each piece of group $g$ is $O(2^g)$, the \construct\ time is $O(r\log{n})$ per piece, and the total time for $g$ is $O(n\sqrt{r}\log^2{n})$. Next, consider any group $g$ containing pieces with number of edges much greater than $2^g$. Then for each piece $\R_j$ in $g$, each boundary edge $(u,v)$ in $\R_j$ has $\delta_{uv} \geq \frac{2^g n}{m \sqrt{r}}$. By Lemma \ref{lem:sumdeltapaths}, $\sum_{P \in \mathcal{P} \cup \mathcal{C} \cup \mathcal{Q}} \sum_{(u,v) \in \dir{P}} \delta_{uv} = O(n\log{n})$. Therefore, the number of times the pieces of $g$ are affected is $O(\frac{m \sqrt{r} \log{n}}{2^g})$. The time taken for each execution of \construct\ on a piece of group $g$ is $O(2^g \sqrt{r}\log{n})$. Therefore, the total time taken for \construct\ over all pieces of $g$ is $O(mr\log^2{n})$. Summing over all groups gives a total time for \construct\ during Step 2 for all pieces as $O(mr\log{m}\log^2{n})=O(mr\log^3{n})$.
\end{proof}

\noindent Combining this with the total time for all the search operations during the second step gives $\tilde{O}(mr + n^{3/2}/r^{1/4})$.

\paragraph{Efficiency of Steps 1 and 3.}
The first step of the algorithm executes $O(\sqrt{r})$ iterations of Gabow and Tarjan's algorithm on the entire graph. This takes $O(m\sqrt{r})$ time. Note that the time for the first step is dominated by the time for the second step.

After the second step, the algorithm has a compressed feasible matching with $O(n / \beta)$ unmatched vertices remaining. This is then converted into an $\R$-feasible matching in $O(m + n\log{n})$ time by Lemma \ref{lem:feasible-convert}. The remaining $O(\unmatchedrem)$ vertices are then matched one at a time by performing iterations of Hungarian search. Each iteration takes $O(m \log{n})$ time, giving a total complexity of $O(m \unmatchedrem \log{n})$ for the third step.

Combining the times taken for the first, second, and third steps of the algorithm gives $O(mr\log^3{n} + m \unmatchedrem \log{n})$. Setting $r=n^{2/5} $ gives a total complexity of $\tilde{O}(mn^{2/5} )$. This is the complexity for a scale of the algorithm. Since there are $O(\log(nC))$ scales, the total complexity is $\tilde{O}(mn^{2/5}\log(nC))$.

\paragraph{Extension to minimum-cost maximum-cardinality matching.}
The algorithm described thus far computes a perfect matching. However, we can use the following technique, described by Gabow and Tarjan~\cite{gt_sjc89}, to reduce the perfect weighted matching problem to the maximum weighted matching problem. The technique makes a copy of the graph $G$ (let this copy be $G'$), and, for every vertex $v \in V(G)$, connects $v$ to its counterpart $v' \in V(G')$ by an edge of large cost. This cost could be, for example, the total of all edge costs in the graph plus 1. The new graph has a perfect matching, but a minimum cost perfect matching on $G'$ corresponds to a minimum cost maximum matching on $G$. Furthermore, an $r$-clustering of $G$ can also be used as an $r$-clustering in the new graph $G'$. We note that while the technique preserves the $r$-clustering property, it does not preserve planarity. Therefore, this reduction technique does not directly extend to the planar graph matching algorithm described in Section \ref{sec:planar}.

\subsection{Regarding $r$-clusterings in $K_h$-minor free graphs} \label{subsec:hminor}
Using the result of Wulff-Nilsen~\cite{wulff2011separator}, one can obtain an $r$-clustering for $K_h$-minor free graphs. The total number of boundary vertices in their definition is $\tilde{O}(n/\sqrt{r})$ instead of $O(n/\sqrt{r})$. Similarly, the number of boundary vertices per piece is $\tilde{O}(\sqrt{r})$ instead of $O(\sqrt{r})$. This increases the sizes of both the vertex and edge sets of $H$ by a $\poly(\log{n})$ term. To handle the increase in the size of $H$, our algorithm reduces the error $\delta_{uv}$ on each edge to by a $\poly(\log{n})$ factor so that the product of $\delta_{uv}$ and the number of boundary vertices is $O(n)$, which guarantees that the optimal solution at the start of each scale is $O(n)$. For constant $h$, we can set $\delta_{uv}=O(\sqrt{r} / \poly(\log{n}))$. In \searchandswitch, instead of raising the dual weights by $\sqrt{r}$, we raise it by $O(\sqrt{r} / \poly(\log{n}))$. The convergence rate consequently slows down by a $\poly(\log{n})$ factor, with the algorithm taking $\tilde{O}(\sqrt{n}/r^{1/4})$ phases during the second step. Furthermore, from the efficiency discussion of \searchandswitch, the search takes $O(|E_H| + |V_H|) \times \tilde{O}(\sqrt{n}/r^{1/4}) = \tilde{O}(n^{3/2}/r^{1/4})$ time (Lemma~\ref{lem:hsize}). From Lemma~\ref{lem:sumdeltapaths}, we have that $\sum_{i=1}^N \sum_{(u,v) \in \dir{P}_i}\delta_{uv} = O(n \log n)$. Therefore, we get $(\sqrt{r} / \poly(\log{n})) \times \sum_{i=1}^N |P_i| = \tilde{O}(n/\sqrt{r})$, and the execution time is $\tilde{O}(n^{3/2}/r^{1/4} + nr)$, which is $\tilde{O}(n^{7/5})$ for $r=n^{2/5}$.

\newcommand{\CTotalPlanar}{O(n^{6/5}(\log^{12/5}{n}) \log{nC})}
\newcommand{\mongeraise}{\textsc{Raise}}
\newcommand{\findmin}{\textsc{FindMin}}
\newcommand{\rebuild}{\textsc{Build}}
\newcommand{\findmincolumn}{\textsc{FindMinInColumn}}
\newcommand{\raiserow}{\textsc{RaiseRow}}

\section{Planar Graph Matching Algorithm}
\label{sec:planar}
In this section, we give an improved version of the algorithm described in Section \ref{algorithmforscale} for planar graphs. The algorithm of ~\cite{soda-18} uses known results in planar shortest path computation to execute Hungarian search faster, while the algorithm described in Section \ref{sec:algoritm} uses a Gabow-Tarjan style algorithm to match potentially many vertices each phase, leading to fewer phases being executed. By combining these two approaches, we get an algorithm that both executes fewer phases, and executes each phase efficiently, leading to an $\tilde{O}(n^{6/5}\log{(nC)})$ algorithm.

The speedup for each phase comes from two main prior results in planar shortest paths data structures that were also used in~\cite{soda-18}. The first result is a multiple source shortest paths (MSSP) data structure by Klein \cite{klein_mssp_05}, which can be used in the \construct\ procedure to compute the edges between all pairs of boundary vertices in a piece in $O(r\log{r})$ time instead of $O(r^{3/2}\log{n})$ time. The second is a Monge property-based range searching data structure by Kaplan \etal ~\cite{kaplan_monge_12}, which allows step (\ref{searchdualchange}) of \searchandswitch\ to track the minimum slack outgoing edge of from a vertex in $H$ in amortized $O(\text{poly}(n))$ time per operation.

The algorithm in~\cite{soda-18}, only changes dual weights at the end of a phase, after the single augmenting path of that phase is found, which allows it to completely reconstruct all of the Monge range searching structures each phase. However, the algorithm described in Section \ref{sec:algoritm} dynamically changes the dual weights $\ty{\cdot}$, and therefore the slacks in $H$, throughout the \searchandswitch\ procedure, and the affected Monge range searching structures must be updated immediately to support this change. The result of Kaplan \etal~\cite{kaplan_monge_12} does not mention any sort of dynamic cost update operations. However, we give a procedure that allows the data structure to perform these updates efficiently within the setting of our algorithm.
\paragraph{$r$-division}
In the planar graph setting, we can efficiently compute a planar graph $r$-division, which satisfies stricter requirements than an $r$-clustering.  An $r$-division is a partition of the edge set of the graph into $O(n/r)$ pieces of size at most $r$ each having $O(\sqrt{r})$ boundary vertices. The total number of boundary vertices, counting multiplicities is $O(n/\sqrt{r})$. We will reuse the same notations described on the $r$-clustering for the $r$-division. For our algorithm, we require an additional property for the $r$-division; each piece $\R_j$ of the $r$-division has $O(1)$ holes, which are faces of $\R_j$ that are not faces of the original graph $G$. Given a constant degree planar graph, an $r$-division with few holes can be constructed in $O(n\log{n})$ time~\cite{klein_rdiv_98}. The constant degree assumption can be assumed without loss of generality for planar graphs; an explanation is given in \cite{soda-18}.

\ignore{
The organization of this section is as follows. In Subsection \ref{subsec:planar-background}, we give important background information on prior results for planar graphs that are used in our improved planar graph algorithm. In Subsection \ref{subsec:planar-updates} we describe the procedure that allows the Monge range searching data structures to support a specific type of cost update operation, which corresponds to dual weight changes. Subsection \ref{subsec:planar-algorithm} describes how these data structures are used for the improved algorithm, and Subsection \ref{subsec:planar-efficiency} describes the running time of this algorithm.
}

\subsection{Algorithm}
\label{subsec:planar-algorithm}

Our improved planar graph matching algorithm is mostly identical to that presented in Section \ref{sec:algoritm}. In this section, we only describe the modifications. For the second step, this includes using a nearest neighbor data structure to support faster augmenting path searches and speeding up the \construct\ procedure by using Klein's MSSP data structure. For the third step, the algorithm of ~\cite{soda-18} can be used almost directly.

We describe the modifications to the algorithm of Section \ref{sec:algoritm} under the assumption that we have access to a nearest neighbor data structure on the compressed residual graph $H$ that supports the following operations. 

\begin{itemize}
\item \findmin: Given a vertex $u \in V_H$, return the minimum slack outgoing edge $(u,v) \in E_H$.
\item \mongeraise: Given a vertex $v \in V_H$ whose dual weight magnitude increased by a value $c$, update the dual weight of $v$ in the data structure.
\item \rebuild: Given a piece $\R_j$, build nearest neighbor data structure for edges $E_j^H$ in $\R_j$.
\end{itemize}
We assume that such a structure can be constructed in $\tilde{O}(n/\sqrt{r})$ time. Specifically, the \rebuild\ operation takes $\tilde{O}(\sqrt{r})$ time per piece, and the data structure can be constructed by calling \rebuild\ on each of the $O(n/r)$ pieces. The \findmin\ operation can be implemented in $O(\poly(\log{r}))$. The time for \mongeraise\ is bounded in an amortized sense. After $k$ \mongeraise\ operations and $d$ \rebuild\ operations, the total time spent for \mongeraise\ is $\tilde{O}(k + \sqrt{r}d)$. This data structure is described in detail in Subsection \ref{subsec:nearest-neighbor}.

The first step of the algorithm is unchanged for the planar graph version; the algorithm will still execute $O(\sqrt{r})$ iterations of Gabow and Tarjan's algorithm, taking $O(n\sqrt{r})$ time. We next describe the changes to the second step of the algorithm.

The planar graph version of the second step has two main sources of improvement. The first source of improvement arises from speeding up the \construct\ procedure. It is easy to see that, using the planar graph MSSP data structure of Klein~\cite{klein_mssp_05}, the edges of a piece of $H$ can be rebuilt in $\tilde{O}(r)$ time. The same data structure was used to reconstruct pieces of $H$ in~\cite{soda-18}. As was the case for the algorithm in Section~\ref{algorithmforscale}, the total number of affected pieces is $\tilde{O}(n/\sqrt{r})$. Therefore, the total work done by \construct\ is $\tilde{O}(n\sqrt{r})$.

The algorithm in the \searchandswitch\ procedure described earlier takes $\tilde{O}(\sqrt{r})$ average time per vertex visit to identify the minimum slack outgoing edge from the end of the search path and update the sorted orderings. However, we can use the \findmin\ and \mongeraise\ procedures of the nearest neighbor data structure described above to reduce this to $O(\poly(\log{r}))$ amortized time per visit for the case of planar graphs. During step (\ref{searchdualchange}) of \searchandswitch, the minimum slack outgoing edge from the vertex $u$ at the end of the search path can be identified in $O(\poly(\log{r})$ time by executing a \findmin\ query on the nearest neighbor data structure. The second step invokes \findmin\ at most $O(\beta n/r) = O(n^{3/2}/r^{3/4})$ times for a total time of $O(n^{3/2}/r^{3/4})$. 

Whenever the dual weight $\ty{u}$ of a vertex $u \in V_H$ increases in magnitude during step (\ref{searchdualchange}) of \searchandswitch, the algorithm will execute \mongeraise\ on $u$. This can occur at most ${O}(\beta n/r)={O}(n^{3/2}/r^{3/4})$ times during the algorithm. The total complexity of a sequence of $k$ \mongeraise\ operations is $\tilde{O}(k)$, not counting the cost of $\tilde{O}(\sqrt{r})$ incurred from each execution of \rebuild. However, these additional $\tilde{O}(\sqrt{r})$ terms can easily be taxed on the costs for \rebuild\ itself. We conclude that the total time for all \mongeraise\ and \findmin\ operations, aside from that taxed on \rebuild,  is $\tilde{O}(n^{3/2}/r^{3/4})$. 

The dual weights $\ty{\cdot}$ also change during \switch. However, the algorithm can simply call \rebuild\ on each piece whose slacks changed during \switch\ at the end of \switch. The required number of calls to \rebuild\ is proportional to the number of affected pieces. When the dual weight of a boundary vertex changes during \switch, \rebuild\ must be called on each adjacent piece. However, since the graph is constant degree, this does not asymptotically increase the number of \rebuild\ operations. Therefore, the time taken for \rebuild\ is dominated by the time taken for \construct. 

After the second step, the algorithm has a compressed feasible matching with $O(\unmatchedrem)$ unmatched vertices remaining. Each of these remaining vertices can be matched one at a time using iterations of Hungarian search. The third step of the algorithm in Section \ref{sec:algoritm} implements each such search in $\tilde{O}(m)$ time. However, in the planar setting, we can make use of existing planar shortest path data structures to execute Hungarian searches more efficiently. The procedure for this improved Hungarian search is described extensively in~\cite{soda-18}, and applies to our setting with minimal modification. Using FR-Dijkstra\cite{fr_dijkstra_06, kaplan_monge_12}, each Hungarian Search is executed in $\tilde{O}(n/\sqrt{r})$ time. After finding an augmenting path during the third step, the algorithm must update $H$ in all pieces containing edges of the augmenting path. However, using the same arguments as those presented for the algorithm of Section~\ref{algorithmforscale} (or the similar argument in~\cite{soda-18}), the total number of such affected pieces can be shown as $O(n/\sqrt{r} \log{n}))$. Reconstructing a piece of $H$ requires $\tilde{O}(r)$ time, for a total time of $\tilde{O}(n\sqrt{r})$. Note that this is the same time complexity, ignoring log terms, as that for rebuilding the pieces of $H$ in the second step. 

Combining the times taken for the first, second and third steps gives $\tilde{O}(n^{3/2}/r^{3/4} + n\sqrt{r})$. Setting $r=n^{2/5}$ gives a time of $\tilde{O}(n^{6/5})$ per scale as desired. Over the $O(\log{(nC)})$ scales, the total time taken is $\tilde{O}(n^{6/5}\log{(nC)})$. It remains to describe the implementation details of the nearest neighbor data structure.

\subsection{Nearest neighbor data structure}
\label{subsec:nearest-neighbor}
We next describe the implementation details of the nearest neighbor data structure used in Subsection \ref{subsec:planar-algorithm}. We define a set of data structures for each piece $\R_j$, using a standard technique. Each data structure of a piece will support a \findmin-type operation as well as a \mongeraise-type operation on a subset of edges in $E_j^H$. 

The goal of \findmin\ is to identify the minimum slack outgoing edge from a vertex of $H$. Recall that the slack of any edge $(u,v) \in E_H$ can be computed by using $s_H(u, v) = \phi(u,v) - |\ty{u}| + |\ty{v}|$. Since all outgoing edges from $u$ have the same value $\ty{u}$, it is sufficient to find the edge $(u, v')$ that minimizes $\phi(u,v') + |\ty{v'}|$. For the purposes of identifying the minimum slack edge outgoing from $u$, we define the cost of any edge $(u,v)$ as $c(u,v) = \phi(u,v) + |\ty{v}|$. Observe that the net-costs only change during \switch, after which \rebuild\ is called on each affected piece. Therefore, the data structure only needs to support dynamic cost increase operations, corresponding to an increase in the magnitude of $\ty{v}$.

To identify the minimum slack outgoing edge from $u$ in the case where $u$ is a boundary vertex, the algorithm will split the  boundary-to-boundary edges of each piece into \textit{Monge groups} using the standard technique given by Fakcharoenphol and Rao in \cite{fr_dijkstra_06} and reiterated by Kaplan \etal\ in \cite{kaplan_monge_12}. Each group will have a corresponding Monge\footnote{For notational convenience, we do not distinguish between Monge and inverse Monge matrices in this description.} matrix. A matrix $M$ is Monge if for any pair of rows $i < j$ and any pair of columns $k < l$, $M_{ik} + M_{jl} \leq M_{il} + M_{jk}$. For any hole $h$ of a piece, we can define a cost matrix $M^h$ whose row and column orderings correspond to a clockwise ordering of the boundary vertices of $h$. Here $M^h_{ij}$ is the cost of the edge from the $i$th node to the $j$th node in the clockwise ordering. This matrix can be recursively divided into Monge submatrices with each vertex of $h$ belonging to $O(\log{r})$ submatrices. For any pair of distinct holes $h \neq h'$, we can define a cost matrix $M^{h,h'}$ whose rows correspond to the clockwise ordering of $h$ and whose columns correspond to the clockwise ordering of $h'$. $M^{h,h'}$ can be replaced by two Monge matrices. Since there are $O(1)$ holes per piece and each vertex is part of $O(1)$ pieces, each vertex belongs to $O(\log{r})$ Monge groups. For each of the Monge matrices, we make the common assumption that the Monge matrices are not explicitly represented in memory, rather, the cost of any $M_{ij}$ can be computed in $O(1)$ time by computing $\phi(u,v) + |\ty{v}|$. For each Monge matrix group, the algorithm will maintain a data structure that supports the following operations. 

\begin{itemize}
\item \findmincolumn: Given any column of $M$, return the minimum value in the column.
\item \raiserow: Given a row of $M$, increase the value of all entries in the row by a constant $c$.
\end{itemize}

A description of how to construct a data structure that efficiently supports \findmincolumn\ is given in~\cite{kaplan_monge_12}. However, their result does not explicitly support \raiserow. Since we will be defining the \raiserow\ function on their data structure, we describe their data structure in some detail in Subsection~\ref{sec:nearestneighbor}. For full details, see their paper. The Monge matrix data structure can be built on a $p$ by $q$ matrix in $\tilde{O}(p)$ time. It supports \findmincolumn\ queries in $\tilde{O}(\log{p})$ time. In the following section, we describe how any sequence of $k$ \raiserow\ operations can be implemented for this data structure in $\tilde{O}(p + k)$ time. 

Given these complexities, the complexities of the global nearest neighbor operations \findmin, \mongeraise, and \rebuild\ easily follow. Calling \rebuild\ on a piece $\R_j$ requires reconstructing all the Monge matrix data structures on the piece. Since each boundary vertex is represented in $O(\log{r})$ such Monge matrices, the total time spent is $\tilde{O}(\sqrt{r})$. To support the \findmin\ operation on a vertex $u \in V_H$, it is sufficient to query the $O(\log{r})$ Monge matrices that $u$ belongs to. This can be done in $O(\poly(\log{r}))$ time. Finally, the \mongeraise\ operation can be supported for a vertex $u \in V_H$ by calling \raiserow\ on all Monge matrix data structures containing $u$. When \rebuild\ is called on a piece, each of the Monge matrix data structures in the piece are reconstructed. In between two such reconstructions, the \raiserow\ cost associated with the number of rows $p$ could be accumulated once again. Hence, after $d$ \rebuild\ operations and $k$ \raiserow\ operations, the total time taken is $\tilde{O}(k + d\sqrt{r})$, as desired.

We note that this setup as described is only organized on boundary-to-boundary edges of $H$. However, it is easy to support a similar set of operations for edges adjacent to free internal vertices within the same time complexity. It remains to present the Monge data structure for a single Monge group.

\subsection{Data structure on a Monge matrix}
\label{sec:nearestneighbor}
This section gives details of how to implement a data structure on each Monge matrix group that supports the operations \findmincolumn\ and \raiserow. The majority of the data structure uses the result of~\cite{kaplan_monge_12}; our only contribution is the \raiserow\ procedure. We describe the inner workings of the structure in some detail; for full details, see~\cite{kaplan_monge_12}.

The structure takes as input a $p$ by $q$ Monge matrix $M$ and supports the following query in $O(\poly(\log{p}))$ time: for any submatrix of $M$ consisting of any one column of $M$ and any contiguous interval of rows in $M$, what is the minimum\footnote{The result of~\cite{kaplan_monge_12} mainly describes finding the maximum value, but either can be computed using the same method.} value within that submatrix? 

By plotting the values of any row of $M$, and linearly interpolating between points, we can obtain a set of pseudo-lines $L$. Let $\ell_y \in L$ be a pseudo-line with respect to a row $y$. $\ell_y$ is effectively a function $\ell_y(x)$, where $x$ is a row and $\ell_y(x)=M_{yx}$, although by linearly interpolating between points, $x$ can also be seen as a real number. From the Monge property, it can be shown that any pair of pseudo-lines cross at most once. 

The lower envelope of $M$ a function $\mathcal{E}(x) \mid x \in \mathbb{R}$, where $\mathcal{E}(x) = \min_{\ell_y \in L}\ell_y(x)$. The lower envelope is made up of portions of pseudolines; specifically, each pseudoline is part of the lower envelope over at most one contiguous interval. A breakpoint is an intersection of two pseudolines along the lower envelope, and there are at most $O(p)$ breakpoints at any given time. Thus, the lower envelope can be compactly stored using $O(p)$ intervals. Given such a representation, one can find, for any column $x$, the row $y$ that minimizes $M_{yx}$ in $O(\log p)$ time by using binary search over the intervals. 

To construct the lower envelope, the approach of~\cite{kaplan_monge_12} builds a balanced binary range tree $T$ on the rows of $M$. The leaves of $T$ represent the rows themselves, and internal nodes represent sets of all their descendants in $T$. Each node of $T$ will store the lower envelope for the set of rows it represents. These lower envelopes are computed in a bottom-up fashion, starting from the leaves. The lower envelope of a node representing a set of size $k$ can be computed from the lower envelopes of its two children in $O(k + \log k \log q)$ time. Summing over the entire tree gives a construction time of $O(p(\log{p} + \log{q}))$. Using the range tree $T$ one can, for any range of rows and any column find the minimum element in $O(\poly(\log{p}))$ time. This is done by taking the minimum over all $O(\log{p})$ canonical subsets of the range.

\paragraph{Updating the Monge data structures.} To help facilitate dual weight magnitude increases, we describe an additional procedure called \raiserow\ for use with the data structure of Kaplan \etal\ This procedure will allow us to, for any row $y$ of the matrix $M$, increase the cost of every entry in the row by a constant $c$. We give a procedure for repairing the affected portion of the lower envelope as a result of this change. The total time taken will be bounded in an amortized sense; after a sequence of $k$ \raiserow\ operations, the total time taken is $\tilde{O}(p + k)$; recall that $p$ is the number of rows.

Increasing the entry of all elements in row $y$ is equivalent to raising the pseudo-line $\ell_y$ up by $c$. This may introduce new breakpoints into the lower envelope, and may remove the presence of $\ell_y$ from the lower envelope entirely. Such changes may occur to the lower envelopes of any of the $O(\log{p})$ nodes of the range tree $T$ that contain $y$ as a descendant; the other nodes of $T$ are unaffected. We describe how to repair the information starting at the bottom of the tree. 

Assume we are given an internal node $t$ of the tree $T$ whose lower envelope information needs to be repaired, and that the lower envelope information of its two children is accurate. Let $M'$ be the submatrix consisting of the rows represented by $t$. Assume $M'$ is a $p'$ by $q$ matrix, where $p'$ is the number of rows represented by $t$, and let $\mathcal{E}'(x)$ be the lower envelope of this submatrix. Let $[i, j]$ be the interval of values such that $\mathcal{E}'(x) = \ell_y$. The envelope will only change in this interval, and some new breakpoints may need to be created. Given any value of $x$, we can find the pseudo-line that contains $x$ on the lower envelope after raising $\ell_y$ in $O(\poly(\log{p'}))$ time by executing two range minimum queries on the subtree of $T$ rooted at $t$. The first query interval will consist of all rows above $y$ and the second query will consist of all rows below $y$. By taking the minimum over the results of these two range queries with the new value $\ell_y(x)$, we obtain the pseudo-line on the lower envelope that contains $x$ after the \raiserow\ operation. Using this strategy, we can use binary search to find the left-most breakpoint in the interval $[i,j]$ in $O(\poly(\log{p'}))$ time. This process can be repeated for each successive breakpoint, until no new breakpoints are found. The time complexity is therefore proportional to the number of new breakpoints formed as a result of raising $\ell_y$. Let the number of breakpoints formed be $\alpha$. Then the complexity of \raiserow\ is $O(\alpha \poly(\log{|p'|}))$ for the node $t$. Next, observe two facts. First, the maximum number of breakpoints in $\mathcal{E}'(x)$ is $O(p')$. Second, each \raiserow\ operation reduces the number of breakpoints in $\mathcal{E}'(x)$ by at most 1. Therefore, after a sequence of $k$ \raiserow\ operations, the total time taken for node $t$ is $\tilde{O}(p' + k)$. Summing over all nodes of $T$ gives us the desired total time of all \raiserow\ operations as $\tilde{O}(p' + k)$, as desired.

\section{Conclusion}
\label{sec:conclusion}
In this paper, we give an $\tilde{O}(n^{7/5}\log{(nC)})$ (resp. $\tilde{O}(n^{6/5}\log{(nC)}) $ time algorithm for computing minimum-cost matchings in $K_h$-minor free (resp. planar graphs). We conclude by asking the following open questions. 
\begin{itemize}
\item Can we improve the \construct\ procedure for $K_h$-minor free graphs from $O(r^{3/2})$ to $O(r)$ for each piece? This can be done by designing a shortest path data structure for $K_h$-minor free graphs that is similar to the MSSP data structure of Klein~\cite{klein_mssp_05}. Such a data structure will improve the running time of our algorithm to $\tilde{O}(n^{4/3})$. 
\item Can we bridge the gap between our $\tilde{O}(n^{6/5}\log{(nC)})$ time weighted planar matching algorithm and the previously existing $\tilde{O}(n)$ unweighted planar matching algorithm~\cite{multiple_planar_maxflow}? 
\end{itemize}

\appendix

\section{Discussion of Correctness for \reduce\ and \reduceslack}
\label{A:reduce}
The following discussion demonstrates that any call to \reduce\ or \reduceslack\ that satisfies the preconditions in those procedures' definitions will not violate compressed feasibility. For both procedures, the discussion argues conditions (a)--(e) of compressed feasibility hold, establishing Lemma \ref{lem:reduce}.

The \reduce\ procedure sets the dual weight $\ty{\cdot}$  so that conditions (b) and (e) of compressed feasibility hold based on its preconditions. No vertex of $A_H$ has a change in dual weight  during \reduce\ and so condition (a)\ holds. For (d), observe that all edges of $b_j^{\mathcal{A}}$ (resp. $b_j^{\mathcal{I}}$) are outgoing. A reduction in the dual weight $\ty{\cdot}$ will only increase the slack on every edge going out of this vertex, and so  the edges in $E_j^H$ remain $H$-feasible. Similarly, every edge incident on any $v  \in (V_j \setminus \boundary_j) \cap \activeH$ (resp. $\inactiveH$) is not in the matching and therefore a reduction of dual weight of $v$ will only increase the slack on the edge, implying (c).

\reduceslack\ does not violate (c) for any piece that $v$ participates in. This is because all edges except $(u,v)$ are edges that are not in the matching, and so, a reduction  of the dual weight only increases the slack on the other edges, and condition of equation \eqref{eq:feas5} holds. From the definition of slack for matched edges, it follows that the new dual weight of $y_j(v)-s(u,v)$, is   non-negative and the slack of $s(u,v)$ after the dual update is $0$.  Therefore the condition of equation \eqref{eq:feas6} is satisfied, and (c) holds. For conditions (a), (b), (d) and (e), if $v$ is an internal vertex, then $\ty{\cdot}$ values are not modified by the procedure and so (a), (b), (d), and (e) hold trivially. Otherwise, if $v $ is a boundary vertex, since $v \in B$, (a) holds trivially.  and since the updated dual weight  $\ty{v}$ is non-negative,  (b) holds.  $\ty{v}$ and $y_j(v)$ are updated so that (e) holds. Finally, for condition (d), we need to show $H$-feasibility of edges going out of $v$, we address the case where $v \in V_H$. First, consider any edge $(u', v) \in E_H$ incoming to $v$. The projection of $(u',v)$  must contain the edge $(u,v)$, since $(u,v)$ is the only edge in the residual graph that is directed into $v$. Therefore, the slack $s_H(u',v) \ge s(u,v)$. The procedure decreases the dual weight $\ty{v}$. This reduces the slack on $(u', v)$ by at most $s(u,v)$ implying that the slack on the edge $(u',v)$ remains non-negative. Every other edge of $H$ incident on $v$ is directed away from $v$, so a reduction of dual weight only increases the slack on these edges, implying (d).

\section{Details of the \construct\ Procedure}
\label{A:construct}

This section describes in further detail how to implement the \construct\ procedure defined in Section \ref{subsec:rtoplanar}. The input to the \construct\ procedure is an $\R$-feasible  matching $M_j$ and the dual weights $y_j(\cdot)$.
 Let 
$\R_{j}^{'}$
 be the graph of $\R_j$, with all edge weights converted to their slacks according to the current matching $M_j$ and the current dual assignment $y_j(\cdot)$. We note that all edges in $\R_{j}^{'}$ are non-negative. Since the dual assignment is feasible with respect to $M_j$, we know by Lemma \ref{lem:poswt} that the path of minimum net-cost between two vertices is also the path of minimum total slack in $\R_{j}^{'}$. Therefore, it is sufficient to compute the shortest path lengths in $\R_{j}^{'}$ and use \eqref{eq:slack2-1}  to compute the minimum net-cost path in constant time.

Recall that there are four types of edges in $E_j^H$. To compute the boundary-to-boundary edges $(u,v) \in E_j^H$, for each $u\in \boundary_j$, we execute a Dijkstra search over $\R_j'$ from $u$ to obtain the length of the shortest slack path from $u$ to every other boundary node.  To compute an edge from a boundary node $u$ to itself,  for each such boundary vertex $u \in \boundary_j$, create a duplicate vertex $u'$ and add an edge from (resp. to)  $u'$ to (resp. from) any other vertex $v\in V_j   $ if and only if there is an edge from (resp. to)  $u$ to (resp. from) $v$ in $E_j$ with the same cost, i.e., slack $s(u,v)$. Execute Dijkstra's algorithm from $u$ over $\R_j'$ to find the distance to $u'$. 

For piece $\R_j$, we describe how to compute the edges from a vertex $b_j^A$ to a boundary node or the free internal vertex $a_j$. A similar argument also applies for computing edges from $b_j^I$. We  add  a new vertex $s$ to $\R_j'$ and connect them to every free internal vertex $v \in (V_j \setminus \boundary_j)\cap B_F^A$. The cost of the newly added edges is set to zero.   Then, we execute Dijkstra's algorithm from $s$. For every boundary node $v \in \boundary_j$, we add an edge from $b_j$ to $v$ if there is a path between $s$ and $v$   and compute the cost $\phi(\dir{P}_{b_{j},v, j})$ by using Lemma~\ref{lem:poswt}. We add an edge between $b_j$ and $a_j$ if there is a path from $s$ to some $v \in(V_j \setminus \boundary_j)\cap A_F$ and among all such vertices which have a path from $b_j$, use the one that has the smallest cost path from $s$.  Using Lemma~\ref{lem:poswt}, we can obtain the weight $\phi(\dir{P}_{b_j,a_j,j})$. We can use an identical algorithm to compute edges incident on $a_j$ by applying Dijkstra's algorithm on the graph $\R_j''$ where $\R_j''$ is a graph identical to $\R_j'$ except that every edge is in the reverse direction. Together, these three searches compute the remaining edges of $E_j^H$. The total time taken to compute all the edges of $E_j^H$ is $\bruteforceconstruct$ time because $O(\sqrt{r})$ Dijkstra searches over $\R_j'$ are executed.

From this discussion, Lemma \ref{lem:construct} follows. By summing over the sizes of all pieces of the $r$-clustering, we get Corollary \ref{cor:preprocess}.

\section{Proofs for Properties of \sync}
\label{A:sync}
This section provides proofs for properties (1) and (2) of the \sync\ procedure given in Section \ref{subsec:planartor}. It also presents Lemma \ref{lem:syncprop1} which is used in the proof of Lemma \ref{lem:sync-zero-slack}. Finally, we give the proof for Lemma \ref{lem:sync-zero-slack}.

\begin{lemma}
\label{synclemma}
At the end of the \sync\ procedure, both (1) and (2) hold.  
\end{lemma}
\begin{proof} 
Let us denote the dual weights before and after applying the \sync\ procedure as $\yjs(\cdot)$ and $y_j(\cdot)$. We also denote the slack on any edge $(u,v)$ with respect to the original dual weights $\yjs(\cdot)$ as $\sjs(u,v)$.
 To prove (1), we need to show that the  matching $M_j$ along with the new dual weights $y_j(\cdot)$ are $\R$-feasible. We first note that the dual weights change only if  $\ell_v \leq \kappa$ and the change is by $\lambda_v(\kappa - \ell_v)$. This change is positive for vertices of $B$ and negative for vertices of $A$. Therefore, the magnitude of dual weights does not decrease from the procedure.  We show \eqref{eq:feas5} and \eqref{eq:feas6} next. For any edge $(u,v)$ directed from $u$ to $v$, we know from the properties of shortest paths that $\ell_v \le \ell_u +\sjs(u,v)$, or,
\begin{eqnarray} 
\ell_v - \ell_u \le \sjs(u,v), \nonumber\\
(\kappa - \ell_u) - (\kappa - \ell_v) \le \sjs(u,v).
\end{eqnarray}  
If $(u,v) \in M_j$, then $u \in A$, $v \in B$ and $\sjs(u,v) = \yjs(u) + \yjs(v) -\dist(u,v) + \delta_{uv}$. We can rewrite the above equation as 
\begin{eqnarray*}
(\kappa - \ell_u) - (\kappa - \ell_v) &\le& \yjs(u) + \yjs(v) -\dist(u,v) + \delta_{uv},\\
(\yjs(u) + \lambda_u (\kappa - \ell_u))+(\yjs(v)+ \lambda_v(\kappa - \ell_v)) &\ge& \dist(u,v) - \delta_{uv},\\
y_j(u) + y_j(v) &\ge& \dist(u,v) - \delta_{uv},
\end{eqnarray*}
satisfying \eqref{eq:feas6}. 
If the edge $(u,v)$ directed from $u$ to $v$ is not  in the matching, then $u \in B$, $v \in A$,  and $\sjs(u,v) = \dist(u,v) + \delta_{uv}-\yjs(u) - \yjs(v) $ Therefore,
\begin{eqnarray*}
(\kappa - \ell_u) - (\kappa - \ell_v) &\le& \dist(u,v) + \delta_{uv}-\yjs(u) - \yjs(v), \\
(\yjs(u) + \lambda_u (\kappa - \ell_u))+(\yjs(v)+ \lambda_v(\kappa - \ell_v)) &\le& \dist(u,v) + \delta_{uv},\\
y_j(u) + y_j(v) &\le& \dist(u,v) + \delta_{uv},
\end{eqnarray*}
implying that the edge $(u,v)$ satisfies \eqref{eq:feas5} implying $M_j, y_j(\cdot)$ is $\R$-feasible.

To prove (2), we need to show that for any vertex $v \in \boundary_j \cup ((V_j\setminus \boundary_j)\cap (A_{F}\cup B_F ))$, the new shortest path from $s$ to $v$ in $\R_j'$ is the direct edge from $s$ to $v$. If we show this, then $\ell_v = \kappa - \kappa_v = \kappa - |\ty{v}| + |\yjs(v)|$ or $\lambda_v \ty{v} -\lambda_v \yjs{(v)} = \kappa - \ell_v$. This gives $\ty{v} = \yjs(v) + \lambda_v(\kappa- \ell_v) = y_j(v) $ because $\lambda_v \in \{1,-1\}$.

Therefore, we will show that the shortest path from $s$ to $v$ is no less than the cost of the edge from $s$ to $v$. For the sake of contradiction, let the shortest path, $\dir{P}_{s,v}$ from $s$ to $v$  be strictly less than the cost of the edge $(s,v)$. Let $u$ be the first vertex that appears on $\dir{P}_{s,v}$ after $s$. By the optimal substructure property of shortest paths, $\dir{P}_{s,v}$ with the vertex $s$ removed forms a shortest path from $u$ to $v$. Since all edge costs are the original slacks $\sjs(\cdot)$, this path was also the shortest slack path, $\project{u}{v}{j}$, prior to \sync. We know that the length of $\dir{P}_{s,v}$ is $\kappa -|\ty{u}|+ |\yjs(u)| + \sum_{(a,b) \in \project{u}{v}{j}}\sjs(a,b)$. Since the length of $\dir{P}_{s,v}$ is smaller than the cost of the direct edge from $s$ to $v$, we have\\ 
\begin{eqnarray*}
 \kappa -|\ty{v}| +|\yjs(v)| &>& \kappa -|\ty{u}|+ |\yjs(u)| + \sum_{(a,b) \in \project{u}{v}{j}}\sjs(a,b),\\
|\ty{u}|-|\ty{v}| &>&  |\yjs(u)|-|\yjs(v)| +\sum_{(a,b) \in \project{u}{v}{j}}\sjs(a,b)\\
&=& \phi(\project{u}{v}{j}).  
\end{eqnarray*}  
The last equality holds because, from Lemma \ref{lem:const1} and the fact that $M_j,\yjs(\cdot)$ was an $\R$-feasible matching,  $|\yjs(u)|-|\yjs(v)|+ \sum_{(a,b) \in \project{u}{v}{j}}\sjs(a,b)$ will be equal to $\phi(\project{u}{v}{j})$. The inequality $|\ty{u}|-|\ty{v}| > \phi(\project{u}{v}{j})$ contradicts the $H$-feasibility of the input to \sync.
\end{proof}

\begin{lemma}
\label{lem:syncprop1}
Suppose we are given a piece $\R_j$ and a dual weight $y_j(\cdot)$ for every vertex in $V_j$ and $\ty{\cdot}$ for each vertex in $V_j^H$. Upon applying the \sync\ procedure, let $v$ be any vertex in $V_j$  for which $\ell_v \le \kappa$. Let $P = \langle s= u_0, u_1,\ldots, u_t=v\rangle$ be the shortest path from $s$ to $v$ in $\R_j'$. Then, after the \sync\ procedure, the slack on every edge $(u_q,u_{q+1})$ with respect to the updated dual weights $y_j(\cdot)$ for $1\le q< t$ is zero.  
\end{lemma}
\begin{proof}
 As in the previous proof, we denote the dual weights prior to the execution of the \sync\ procedure by $\yjs(\cdot)$ and the dual weights after by $y_j(\cdot)$. Also, let $\sjs(\cdot,\cdot)$ denote the slack of an edge in $G$ with respect to the dual weights $\yjs(\cdot)$. Since, $P$ is the shortest path from $s$ to $v$\ , for any directed edge on this path from $u_{q}$ to $u_{q+1}$, 

\begin{equation}
\ell_{u_{q+1}} = \ell_{u_{q}} + \sjs(u_q,u_{q+1}). \label{eq:syncprop1}
\end{equation} 

\noindent The shortest path cost from $s$ to $v$, $\ell_v$ is at most $\kappa$ and so, for any vertex $u_{q}$ on the shortest path to $v$, the shortest path to $\ell_{u_q}$ is at most $\kappa$. \sync\ sets dual weights such that $y_j(u_{q}) = \yjs(u_{q}) + \lambda_{u_{q}} (\kappa - \ell_{u_q})$. Therefore, for any edge $(u_q, u_{q+1})$ on $P$ we have

\begin{eqnarray}
y_j(u_{q+1}) = \yjs(u_{q+1}) + \lambda_{u_{q+1}}(\kappa - \ell_{u_{q+1}}),\label{eq:syncprop2}\\
y_j(u_{q}) = \yjs(u_{q}) + \lambda_{u_{q+1}}(\kappa - \ell_{u_{q}}). \label{eq:syncprop3}
\end{eqnarray}

We consider the cases where $(u_q, u_{q+1}) \in M$, and $(u_q, u_{q+1}) \notin M$. First, we consider the case where $(u_q, u_{q+1}) \in M$. Matching edges are directed from a vertex of $A$ to a vertex of $B$, and so, $u_q \in A$ and $u_{q+1} \in B$. By the definition of slack for matching edges, we have

\begin{eqnarray*}
s(u_q, u_{q+1}) &=& y_j(u_{q}) + y_j(u_{q+1}) - c(u_q, u_{q+1}) + \delta_{u_q u_{q+1}} \\
&=& \yjs(u_{q}) + \yjs(u_{q+1}) - c(u_q, u_{q+1}) + \delta_{u_q u_{q+1}} + \lambda_{u_q} (\kappa - \ell_{u_q}) + \lambda_{u_{q+1}} (\kappa - \ell_{u_{q+1}})\\
&=& \sjs(u_q, u_{q+1}) - (\kappa - \ell_{u_q}) + (\kappa - \ell_{u_{q+1}})\\
&=& \sjs(u_q, u_{q+1}) + \ell_{u_q} - \ell_{u_{q+1}} = 0.
\end{eqnarray*}
The last two equations follow from \eqref{eq:syncprop1} and the fact that $\lambda_{u_q} = -1$ and $\lambda_{u_{q+1}} = 1$.

Next, we consider the case where $(u_q, u_{q+1}) \notin M$. Edges that are not in the matching are directed from a vertex of $B$ to a vertex of $A$, and so, $u_q \in B$ and $u_{q+1} \in A$. By the definition of slack for  edges that are not in the matching, we have

\begin{eqnarray*}
s(u_q, u_{q+1}) &=& c(u_q, u_{q+1}) + \delta_{u_q u_{q+1}} - y_j(u_q) - y_j(u_{q+1})   \\
&=& c(u_q, u_{q+1}) + \delta_{u_q u_{q+1}} - \yjs(u_{q+1}) - \yjs(u_q) - \lambda_{u_q} (\kappa - \ell_{u_q}) - \lambda_{u_{q+1}} (\kappa - \ell_{u_{q+1}})\\
&=& \sjs(u_q, u_{q+1}) - (\kappa - \ell_{u_q}) + (\kappa - \ell_{u_{q+1}})\\
&=& \sjs(u_q, u_{q+1}) + \ell_{u_q} - \ell_{u_{q+1}} = 0.
\end{eqnarray*}
The last two equations follow from \eqref{eq:syncprop1} and the fact that $\lambda_{u_q} = 1$ and $\lambda_{u_{q+1}} = -1$.

\end{proof}

Next, using Lemma \ref{lem:syncprop1}, we give a proof for Lemma \ref{lem:sync-zero-slack}. We first restate verbatim the claim of Lemma \ref{lem:sync-zero-slack}.\\

\begin{lemma} Consider a compressed feasible matching with dual weights $\ty{\cdot}$ assigned to every vertex of $V_H$. For any piece $\R_j$ and any vertex $v \in V_j$, let $\yjs(v)$ denote the dual weight prior to executing \sync, and for any edge $(u,v) \in E_j$, let $\sjs(u, v)$ be the slack prior to executing \sync. Let   $y_j(\cdot)$ denote the dual weights of $V_j$ after this execution.  For any edge $(u,v) \in E_j^H$ with a projection $\dir{P}_{u,v,j}=\langle u=u_{0},u_1,\ldots, u_t, u_{t+1}=v\rangle$, suppose $|\ty{u}|-|y_j^*(u)| \ge \sum_{q = 0}^{t} \sjs(u_q, u_{q+1})$.  Let  $\dir{P}_{s,u_{t},j} $ be any shortest path from $s$ to $u_t$ in $\R_j'$ . Then, 

\begin{enumerate}[(i)]  \item  If there exists a shortest path $\dir{P}_{s,u_t,j}$ in $\R_j'$ where $u$ is the second vertex on this path, then after the execution of \sync\ procedure,
 for every $1 \le i\le t-1$, $s(u_{i},u_{i+1})=0$  and $s(u_{t},v)\le |\ty{v}|-|\yjs(v)|$, 
\item Otherwise, there is no shortest path $\dir{P}_{s,u_t,j}$  in $\R_j'$ with $u$ as its second vertex. Consider $u^*$ to be the second vertex of some $\dir{P}_{s,u_{t},j}$ and $u^* \neq u$. Then,   $u^*\in (\boundary_j \cup (V_j \cap (A_F\cup B_F))$, 
 and   $|\ty{u^*}|- |\yjs(u^*)| > \sum_{(u', v') \in \dir{P}_{u^*, v, j}} \sjs(u', v')$. 
\end{enumerate}
\end{lemma}
\begin{proof}
Let $\yjs(\cdot)$ and $y_j(\cdot)$ be the dual weights before and after   the execution of the \sync\ procedure.  Also, let $\sjs(\cdot,\cdot)$ denote the slack of an edge in $G$ with respect to the dual weights $\yjs(\cdot)$. Note that the dual weights $\ty{\cdot}$ do not change from the execution of the \sync\ procedure. First, we establish that for every vertex $u_i$ along $\project{u}{v}{j}$, $\ell_{u_i} \leq \kappa$.
By our assumption,
\begin{equation}
\label{eq:syncprop0}
\kappa_u = |\ty{u}| - |y_j(u)| \geq \sum_{q = 0}^{t} \sjs(u_q, u_{q+1}).
\end{equation}

For any $i$, such that $0 \leq i \leq t+1$, consider the cost of the path $\langle s, u_0, u_1, ..., u_i \rangle$. From \eqref{eq:syncprop0}, the cost of the edge $(s, u_0) = \kappa - \kappa_{u} \le  \kappa - \sum_{q = 0}^{t} \sjs(u_q, u_{q+1})$, and for any $0 \leq q < i$, the cost of the edge $(u_q, u_{q+1}) = \sjs(u_q, u_{q+1})$. Therefore, the cost of the path $\langle s, u_0, u_1, ..., u_i \rangle$ is at most 
\begin{equation}
\label{eq:syncprop00}
\ell_{u_i} \leq \kappa - \sum_{q = 0}^{t} \sjs(u_q, u_{q+1}) + \sum_{q = 0}^{i-1} \sjs(u_q, u_{q+1}) = \kappa - \sum_{q = i}^{t} \sjs(u_q, u_{q+1}).  
\end{equation}
This implies $\ell_{u_i} \leq \kappa$, and the dual weight of $u_i$ is updated by the \sync\ procedure.

From \eqref{eq:syncprop00}, $\ell_{u_{i}} \leq \kappa - \sum_{q = i}^{t} \sjs(u_q, u_{q+1})$. The new dual weight of $u_i$ as updated by the \sync\ procedure is $y_j(u_i) \leftarrow \yjs(u_i) + \lambda_{u_i} (\kappa - \ell_{u_i})$, or,

\begin{equation}
\label{eq:syncprop02}
|y_j(u_i)| - |\yjs(u_i)| = \kappa - \ell_{u_i} \geq \sum_{q = i}^{t} \sjs(u_q, u_{q+1}).
\end{equation}

Note that if $\langle s, u_0, u_1, ..., u_i \rangle$ is not the shortest path from $s$ to $u_i$ in $\R_j'$, then inequalities~\eqref{eq:syncprop00} and~\eqref{eq:syncprop02} are strict inequalities.

First, we address the case where there is a shortest path $\dir{P}_{s,u_{t},j}$ in $\R_j'$ with $u$ as its second vertex.  Since the path $\dir{P}_{u,v,j}$ is the shortest path from $u$ to $v$ in $\R_j'$, from the optimal substructure property, a shortest path from $s$ to $u_{t}$ is $\langle s,u,u_{2}\ldots, u_{t-1}, u_t \rangle$. Since $\ell_{u_i} \le \kappa$  , from Lemma~\ref{lem:syncprop1}, every edge on this path will have a zero slack. 

From~\eqref{eq:syncprop02}, the change in dual weight for $u_t$ is 
\begin{equation} 
\label{eq:ut}
|y_j(u_t)|-|\yjs(u_t)| \ge s^*(u_t,v).
\end{equation} 
The slack for the edge $(u_t,v)$ is given by
\begin{eqnarray*}
s(u_t,v) &=& s^*(u_t,v) - (|y_j(u_t)|-|\yjs(u_t)|) + (|y_j(v)|-|\yjs(v)|)\\
&\le& |y_j(v)|-|\yjs(v)|\\
&=& |\ty{v}|-|\yjs(v)|. 
\end{eqnarray*}
This completes the proof for (i).

Next, we address case (ii), where $u^{*}\neq u$. Note that the only edges that are leaving $s$ are to the vertices of $\boundary_j \cup \{(A_F \cup B_F) \cap V_j\}$. So, $u^*$ has to be a vertex of this set. Next, let the path $\project{s}{u_t}{j} = \langle s=s_0, s_1,\ldots,s_\alpha=u_t  \rangle$ be the shortest path from $s$ to $u_t$ with $s_1=u^*$. Therefore, from Lemma \ref{lem:syncprop1}, all edges on this path have zero slack with respect to the dual weights $y_{j}(\cdot)$. From \eqref{slackcost}, 
\begin{equation}
\label{eq:syncprop03}
|y_j(u^{*})| - |y_j(u_t)| = \phi(\project{u^{*}_}{u_t}{j}).
\end{equation}
Before the execution of the \sync\ procedure, from \eqref{slackcost} we have,
\begin{equation}
\label{eq:syncprop04}
(\sum_{q=0}^{\alpha - 1} \sjs(s_q, s_{q+1})) + |\yjs(u^{*})| - |\yjs(u_t)| = \phi(\project{u^{*}_}{u_i}{j}).
\end{equation}
Subtracting \eqref{eq:syncprop04} from \eqref{eq:syncprop03} gives,
\begin{eqnarray*}
(|y_j(u^*)| - |\yjs(u^*)|) - (|y_j(u_t)| - |\yjs(u_t)|) &=& \sum_{q=0}^{\alpha - 1}\sjs(s_q, s_{q+1}), \\
(|y_j(u^*)| - |\yjs(u^*)|) &=& (\sum_{q=0}^{\alpha - 1} \sjs(s_q, s_{q+1}))+ (|y_j(u_t)| - |\yjs(u_t)|). 
\end{eqnarray*}

Note that if $\langle s, u_0, u_1, ..., u_t \rangle$ is not the shortest path from $s$ to $u_t$ in $\R_j'$, then, as stated before inequalities~\eqref{eq:syncprop00} and~\eqref{eq:syncprop02} are strict inequalities. From applying \eqref{eq:syncprop02} we get that $|y_j(u_t)| - |\yjs(u_t)| > \sjs(u_t,v)$. Therefore,
\begin{equation}
(|y_j(u^{*})| - |\yjs(u^{*})|) >  (\sum_{q=0}^{\alpha - 1} \sjs(s_q, s_{q+1})) +  \sjs(u_t, v) \ge \sum_{(u', v') \in \dir{P}_{u^{*},v,j}} \sjs(u', v').
\end{equation}
By property (2) of the \sync\ procedure, $y_j(s_1)=\ty{s_1}$, and therefore, (ii) follows.

\end{proof}
\bibliographystyle{mystyle}
\bibliography{planarmatch}

\end{document}